\documentclass[11pt,a4paper,leqno]{article}
\usepackage{subfig}
\usepackage{graphicx}
\usepackage{makeidx}
\usepackage{pst-node}
\usepackage{amsmath,amsthm,amssymb}
\usepackage{enumerate}
\usepackage{verbatim}
\usepackage{framed}
\usepackage{mdframed}
\usepackage{ulem}

\usepackage{appendix}

\usepackage{natbib}

\newtheorem{theorem}{Theorem}
\newtheorem{lemma}[theorem]{Lemma} 
\newtheorem{proposition}[theorem]{Proposition} 
\newtheorem{corollary}[theorem]{Corollary}
\newtheorem{claim}{Claim}

\theoremstyle{definition}
\newtheorem{remark}{Remark}

\def\R{{\bf \mbox{I\hspace{-.17em}R}}}

\title{Inheritance of Convexity for the $\mathcal{P}_{\min}$-Restricted Game}

\author{A. Skoda\thanks{
Corresponding author.
Universit\'e de Paris I, Centre d'Economie de la Sorbonne,
106-112 Bd de l'H\^opital, 75013 Paris, France.
E-mail: {\tt alexandre.skoda@univ-paris1.fr}} }

\begin{document}

\maketitle

\begin{abstract}
We consider restricted games on weighted  graphs associated with minimum partitions.
We replace in the classical definition of Myerson restricted game the connected components of any subgraph
by the sub-components corresponding to a minimum partition.
This minimum partition $\mathcal{P}_{\min}$ is induced by the deletion of the minimum weight edges.
We provide
a characterization of the graphs satisfying
inheritance of convexity from the underlying game to 
the restricted game associated with $\mathcal{P}_{\min}$.
Moreover,
we prove that these graphs can be recognized in polynomial time.
\end{abstract}




\section{Introduction}
\label{SectionIntroduction}
A cooperative game is a pair $(N,v)$
where $N$ is a finite set of players
and
$v: 2^N \rightarrow \R $ is a set function
assigning a worth to each coalition of players
with $v(\emptyset) = 0$.
For any coalition $A \subseteq N$,
$v(A)$ represents the worth that players in $A$ can generate by cooperation.
However,
in many situations,
the cooperation of players may be restricted by some communication or social structures.
Then, the worths of coalitions have to be modified 
to take these restrictions into account,
leading to the introduction of restricted games.
Lots of restricted games considered in the literature 
can be described by the $\mathcal{P}$-restricted game introduced by \citet{Skoda201702}.
$\mathcal{P}$ is a correspondence associating with any subset $A$ of $N$
a partition $\mathcal{P}(A)$ of $A$.
The partition restricted game $(N, \overline{v})$
associated with $\mathcal{P}$,
called $\mathcal{P}$-restricted game,
is defined by:
\begin{equation}
\overline{v} (A) = \sum_{F \in \mathcal{P}(A)} v(F),  \; \; \textrm{for all}\; A \subseteq N.
\end{equation}
This partition restricted game appears
in many different works with specific correspondences.
A founding example is the graph-restricted game
introduced by \cite{Myerson77} for communication games.
Communication games are cooperative games $(N, v)$ defined on the set of vertices $N$ of an undirected graph
$G = (N,E)$.
Assuming that
the members of a given coalition can cooperate
if and only if they are connected in $G$,
the Myerson graph-restricted game
$(N, v^{M})$ is
defined by
$v^{M} (A) = \sum_{F \in \mathcal{P}_{M}(A)} v(F)$
for all $A \subseteq N$,
where $\mathcal{P}_M (A)$ is the partition of $A$ into connected components.
Many other correspondences have been considered
to define restricted games
(see, e.g.,
\citet{AlgabaBilbaoLopez2001,Bilbao2000,Bilbao2003,Faigle89,GrabischSkoda2012,Grabisch2013}).
For a given correspondence,
a classical problem is to study the inheritance of convexity
from the initial game $(N,v)$ to the restricted game $(N,\overline{v})$.
Inheritance of convexity is of particular interest
as it implies that good properties are inherited,
for instance the non-emptiness of the core,
and that the Shapley value is in the core.
\citet{Skoda201702} established
an abstract characterization of inheritance of convexity
for an arbitrary correspondence~$\mathcal{P}$.
This characterization can be used to derive
several well-known results for inheritance of convexity with specific correspondences,
in particular
the characterization of inheritance of convexity for Myerson restricted game
established by \cite{NouwelandandBorm1991}.
Of course,
due to its generality
the characterization given by~\citet{Skoda201702}
does not give straightforward
insights into the precise structure of a given correspondence.
More direct characterizations have to be found to check
inheritance of convexity in practice.
In this paper,
we present a characterization
of inheritance of convexity for the correspondence $\mathcal{P}_{\min}$
introduced by~\citet{GrabischSkoda2012}
for communication games on weighted graphs.
The correspondence $\mathcal{P}_{\min}$
is defined on the set $N$ of nodes of
a weighted graph $G=(N, E, w)$
where $w$ is a weight function defined on the set $E$ of edges of $G$.
For a subset $ A \subseteq  N$,
$\mathcal{P}_{\min}(A)$ corresponds to the set of connected components of the subgraph $(A,E(A) \setminus \Sigma(A))$
where $\Sigma(A)$ is the set of minimum weight edges in the subgraph $G_A = (A,E(A))$.
Then,
the $\mathcal{P}_{\min}$-restricted game $(N,\overline{v})$
is defined  by:
\begin{equation}
\overline{v} (A) = \sum_{F \in \mathcal{P}_{\min}(A)} v(F),  \; \; \textrm{for all}\; A \subseteq N.
\end{equation}
Compared to the initial game $(N,v)$,
the $\mathcal{P}_{\min}$-restricted game $(N,\overline{v})$
conforms to the common conception that
members of a coalition have to be connected to cooperate
but also takes into account
the weights of the links between players and therefore differents aspects of cooperation restrictions.
Assuming that the edge-weights reflect the strengths of relationships between players,
$\mathcal{P}_{\min}(A)$ gives a partition of a coalition $A$
into connected coalitions where players are in privileged relationships
(with respect to the minimum relationship strength in $G_{A}$).
\citet{GrabischSkoda2012} first established
that there is always inheritance of superadditivity
from $(N,v)$ to $(N, \overline{v})$ for the correspondence $\mathcal{P}_{\min}$.
In contrast,
they observed that
inheritance of convexity
requires very restrictive conditions on the underlying graph and its edge-weights
giving simple counter-examples to inheritance of convexity 
with graphs with only two or three different edge-weights.
\citet{GrabischSkoda2012} established three necessary conditions
on the underlying weighted graph
but they also pointed out that these conditions are not sufficient 
and that contradictions to inheritance of convexity
are easily obtained with non-connected coalitions.
Following alternative definitions of convexity
in combinatorial optimization
and game theory
when restricted families of subsets
(not necessarily closed under union and intersection)
are considered
(see, \emph{e.g.},
\citet{EdmondsGiles77,Faigle89,Fujishige2005}),
\citet{GrabischSkoda2012} introduced
the $\mathcal{F}$-convexity by restricting convexity to the family $\mathcal{F}$
of connected subsets of $G$.
\citet{Skoda2017} characterized
inheritance of $\mathcal{F}$-convexity for $\mathcal{P}_{\min}$
by five necessary and sufficient conditions on the edge-weights
of specific subgraphs of the underlying graph $G$.
Of course,
the study of inheritance of $\mathcal{F}$-convexity
is a first key step to obtain a characterization of graphs satisfying
inheritance of classical convexity with $\mathcal{P}_{\min}$.
Inheritance of $\mathcal{F}$-convexity is also interesting in itself
as it corresponds to the common restriction to connected subsets
for communication games.
Moreover,
it is satisfied for a much larger family of graphs
than inheritance of convexity.
In particular,
inheritance of $\mathcal{F}$-convexity
allows an arbitrary number of edge-weights,
in contrast to inheritance of convexity
as observed in the present paper.
\citet{Skoda2017} also highlighted
that Myerson restricted game
can be obtained as a restriction
of the $\mathcal{P}_{\min}$-restricted game
associated with graphs with only two different edge-weights,
and proved that inheritance of convexity for Myerson restricted game is equivalent
to inheritance of $\mathcal{F}$-convexity
for the $\mathcal{P}_{\min}$-restricted game
associated with these specific weighted graphs.

In the present paper,
we consider inheritance of classical convexity for the correspondence $\mathcal{P}_{\min}$.
As convexity implies $\mathcal{F}$-convexity,
the conditions established by~\citet{Skoda2017} are necessary.
Let us recall that these last conditions are also sufficient
for inheritance of $\mathcal{F}$-convexity,
but we will only use their necessity throughout the paper.
Now dealing with disconnected subsets of $N$,
we establish supplementary necessary conditions.
As it was foreseeable by taking into account examples given by~\citet{GrabischSkoda2012} and~\citet{Skoda2017},
we get very strong restrictions on edge-weights and on the combinatorial structure of the underlying graph.
These supplementary conditions are much more straightforward than the conditions
established by~\citet{Skoda2017} for inheritance of $\mathcal{F}$-convexity.
In particular,
we obtain that edge-weights can have at most three different values
(Proposition~\ref{propAtMost3DifferentEdgeWeights})
and that many cycles have to be complete or dominated in some sense by two specific vertices.
The constraint on the number of edge-weight values implies that the family of weighted graphs
satisfying inheritance of convexity is
drastically smaller than the family of weighted graphs satisfying inheritance of 
$\mathcal{F}$-convexity.
Moreover,
edges have 
precise positions
according to their weights.
For example,
in the case of three different values $\sigma_1< \sigma_2 < \sigma_3$,
there exists only one edge $e_1$ of minimum weight $\sigma_1$
and all edges of weight $\sigma_2$ are incident to the same end-vertex of $e_1$.
Using these supplementary conditions,
we obtain simple necessary and sufficient conditions.
We give a complete characterization of the connected weighted graphs
satisfying inheritance of convexity with $\mathcal{P}_{\min}$
in Theorems~\ref{lemInheritanceOfConvexityForPminInheritanceOfConvexityForPMG1Cycle-Complete}, \ref{PropInheritanceOfConvexityForPminOnGIffCofGWithWeightwEitherCompleteOrAllVerticesLinkedToTheSameEndVertexOfe1} and~\ref{PropInheritanceOfConvexityForPminOnGIffCofGWithWeightw2EitherCompleteOrAllVerticesLinkedToTheSameEndVertexOfe12}.
Though these graphs are very particular,
they seem quite interesting.
For instance,
when there are only two values and at least two minimum weight edges,
we obtain weighted graphs similar to the ones defined by~\citet{Skoda2017}
relating Myerson restricted game to the $\mathcal{P}_{\min}$-restricted game.
Moreover,
we prove that
these graphs can be recognized in polynomial time.
Theorems~\ref{lemInheritanceOfConvexityForPminInheritanceOfConvexityForPMG1Cycle-Complete}, \ref{PropInheritanceOfConvexityForPminOnGIffCofGWithWeightwEitherCompleteOrAllVerticesLinkedToTheSameEndVertexOfe1} and~\ref{PropInheritanceOfConvexityForPminOnGIffCofGWithWeightw2EitherCompleteOrAllVerticesLinkedToTheSameEndVertexOfe12}
also imply that
inheritance of convexity 
and inheritance of convexity restricted to unanimity games are equivalent for $\mathcal{P}_{\min}$.
This last result was already observed
as a consequence of the general characterization established by~\citet{Skoda201702}
for arbitrary correspondences.

The article is organized as follows.
In Section~\ref{SectionPreliminaryDefinitions},
we give preliminary definitions and
results established by~\citet{GrabischSkoda2012}
and \citet{Skoda201702}.
In particular,
we recall the definition of convexity, $\mathcal{F}$-convexity
and general conditions on a correspondence to have inheritance of superadditivity,
convexity or $\mathcal{F}$-convexity.
In Section~\ref{SectionNecessaryCondF-ConvPmin},
we recall necessary conditions
on a weighted graph established by~\citet{Skoda2017}
to ensure inheritance of $\mathcal{F}$-convexity
with the correspondence $\mathcal{P}_{\min}$.
In Section~\ref{SectionInheritanceOfConvexityWithPmin},
we establish new, very restrictive conditions necessary
to ensure inheritance of classical convexity with $\mathcal{P}_{\min}$.
Section~\ref{SectionWeighthedgraphsforwhichInheritanceofconvexity} contains the main results.
We first provide characterizations of connected weighted graphs
satisfying inheritance of convexity with $\mathcal{P}_{\min}$.
Then,
the case of disconnected graphs is considered. 
Finally,
we prove that it can be decided in polynomial time 
whether a graph satisfies one of the previous characterizations.

\section{Preliminary definitions and results}
\label{SectionPreliminaryDefinitions}

Let $N$ be a given set with $|N| = n$.
We denote by $2^{N}$ the set of all subsets of $N$.
A game $(N,v)$ is
\textit{zero-normalized}\label{definitionZeroNormalizedGame}
if $v(\lbrace i \rbrace) = 0$ for all $i \in N$.
Throughout this paper,
we consider only zero-normalized games.
A game $(N,v)$ is 
\textit{superadditive}
if, for all $A,B \in 2^{N}$ such that $A \cap B = \emptyset$, $v(A \cup B) \geq v(A) + v(B)$.
For any given subset $\emptyset \not= S \subseteq N$,
the unanimity game $(N, u_{S})$ is defined by
\begin{equation}
u_{S}(A) =
\left \lbrace
\begin{array}{cl}
1 &  \textrm{ if } A \supseteq S,\\
0 & \textrm{ otherwise.}
\end{array}
\right.
\end{equation}
We note that $u_{S}$ is superadditive for all $S \not= \emptyset$.

Let us consider a game $(N,v)$.
For arbitrary subsets $A$ and $B$ of $N$,
we define the value
\[
\Delta v(A,B):=v(A\cup B)+v(A\cap B)-v(A)-v(B).
\]
A game $(N,v)$ is \textit{convex}
if its characteristic function $v$ is supermodular,
\emph{i.e.},
$\Delta v(A,B)\geq 0$
for all $A,B \in 2^{N}$.
We note that $u_{S}$ is supermodular for all $S \not= \emptyset$.
Let $\mathcal{F}$
be a
\emph{weakly union-closed family}\footnote{
$\mathcal{F}$ 
is weakly union-closed if $A \cup B \in \mathcal{F}$ for all $A$, $B \in \mathcal{F}$
such that $A \cap B \not= \emptyset $ \citep{FaigleGrabisch2010}.
Weakly union-closed families were introduced
and analysed by \citet{Algaba98}
(see also \citet{AlgabaBilbaoBormLopez2000})
and called union stable systems.
}
of subsets of $N$
such that $\emptyset \notin \mathcal{F}$.
A game $v$ on $2^{N}$ is said to be \emph{$\mathcal{F}$-convex}
if $\Delta v(A,B) \geq 0$,
for all $A,B \in \mathcal{F}$
such that $A \cap B \in \mathcal{F}$.
Let us note that
a game $(N,v)$ is convex if and only if
it is superadditive and $\mathcal{F}$-convex
with $\mathcal{F} = 2^{N} \setminus \lbrace \emptyset \rbrace$.

For a given graph $G = (N, E)$,
we say that a subset
$A \subseteq N$ is connected
if the induced graph $G_{A} = (A,E(A))$
is connected.

A \textit{correspondence} $f$ with domain $X$
and range $Y$
is a map that associates to every element $x \in X$
a subset $f(x)$ of $Y$,
\emph{i.e.},
a map from $X$ to $2^Y$.
Throughout this paper,
we consider correspondences $\mathcal{P}$
with domain and range $2^N$,
such that for every subset $\emptyset \not= A \subseteq N$,
the family $\mathcal{P}(A)$ of subsets of $N$ corresponds to a partition of $A$.
We set $\mathcal{P}(\emptyset) = \lbrace \emptyset \rbrace$.
For a given correspondence $\mathcal{P}$ on $2^N$
and subsets $A \subseteq B \subseteq N$,
we denote by $\mathcal{P}(B)_{|A}$
the \textit{restriction} of the partition $\mathcal{P}(B)$ to $A$.
More precisely,
if $\mathcal{P}(B) = \lbrace B_1, B_2, \ldots, B_p \rbrace$,
then
$\mathcal{P}(B)_{|A} =
\lbrace B_i \cap A \mid i = 1, \ldots, p, B_i \cap A \not= \emptyset\rbrace$.

For two given subsets $A$ and $B$ of $N$,
$\mathcal{P}(A)$ is a \textit{refinement} of $\mathcal{P}(B)$
if every block of $\mathcal{P}(A)$ is a subset of some block of $\mathcal{P}(B)$.

We recall the following results established by~\citet{GrabischSkoda2012}.

\begin{theorem}
\label{thLGNEPNFaSNuSsFaSNuSin}
Let $N$ be an arbitrary set
and $\mathcal{P}$ a correspondence on $2^N$.
The following conditions are equivalent:
\begin{enumerate}[1)]
\item
The $\mathcal{P}$-restricted game $(N,\overline{u_{S}})$ is superadditive
for all $\emptyset \not= S \subseteq N$.
\item
\label{Claim3ThLGNEPNFaSNuSsFaSNuSin}
$\mathcal{P}(A)$ is a refinement of $\mathcal{P}(B)_{|A}$
for all subsets $A \subseteq B \subseteq N$.
\item
The $\mathcal{P}$-restricted game $(N,\overline{v})$ is superadditive
for all superadditive game $(N,v)$.
\end{enumerate}
\end{theorem}

As 
$\mathcal{P}_{\min}(A)$ is a refinement of $\mathcal{P}_{\min}(B)_{|A}$
for all subsets $A \subseteq B \subseteq N$,
Theorem~\ref{thLGNEPNFaSNuSsFaSNuSin} implies the following result.

\begin{corollary}
\label{corIGNEiafttsihfv}
Let $G=(N,E,w)$ be an arbitrary weighted graph.
The $\mathcal{P}_{\min}$-restricted game $(N,\overline{v})$ is superadditive
for every superadditive game $(N,v)$.
\end{corollary}

\begin{theorem}
\label{thLGNuPNFNFSNuSSNNuSFABFABF}
Let $N$ be an arbitrary set and 
$\mathcal{P}$ a correspondence on $2^N$.
Let $\mathcal{F}$ be a weakly union-closed family of subsets of $N$
with $\emptyset \notin \mathcal{F}$.
If the $\mathcal{P}$-restricted game $(N,\overline{u_{S}})$ is superadditive
for all $\emptyset \not= S \subseteq N$,
then the following conditions are equivalent:
\begin{enumerate}[1)]
\item
The $\mathcal{P}$-restricted game $(N, \overline{u_{S}})$ is $\mathcal{F}$-convex
for all $\emptyset \not= S \subseteq N $,
\item
For all $i \in N$,
for all $ A \subseteq B \subseteq N \setminus \lbrace i \rbrace$
such that $A$, $B$, and $A \cup \lbrace i \rbrace$ are in $\mathcal{F}$,
and for all $A' \in \mathcal{P}(A \cup \lbrace i \rbrace)_{|A}$,
$\mathcal{P}(A)_{|A'} = \mathcal{P}(B)_{|A'}$.
\end{enumerate}
\end{theorem}

We also recall the following lemmas proved by \citet{GrabischSkoda2012}.
We include the proofs
as these two results are extensively used throughout the paper.

\begin{lemma}
\label{lemmavB+i=1pvBiA}
Let us consider $A, B \subseteq N$
and a partition $\lbrace B_{1}, B_{2}, \ldots, B_{p} \rbrace$ of $B$.
Let $\mathcal{F}$ be a weakly union-closed family of subsets of $N$
with $\emptyset \notin \mathcal{F}$.
If $A, B_{i}$, and $ A\cap B_{i} \in \mathcal{F}$
for all $i \in \lbrace 1, \ldots, p \rbrace$,
then for every $\mathcal{F}$-convex game $(N,v)$
we have
\begin{equation}\label{eqlemsupermodularpropvAB+i=1pvABi}
v(A \cup B) + \sum_{i=1}^{p} v(A \cap B_{i}) \geq v(A) + \sum_{i=1}^{p} v(B_{i}).
\end{equation}
\end{lemma}

\begin{proof}
We prove the result by induction.
(\ref{eqlemsupermodularpropvAB+i=1pvABi}) is obviously satisfied for $p=1$.
Let us assume it is satisfied for $p$
and let us consider a partition $\lbrace  B_{1}, B_{2}, \ldots, B_{p}, B_{p+1} \rbrace$ of $B$.
We set $B' = B_{1} \cup B_{2} \cup \ldots \cup B_{p}$.
The $\mathcal{F}$-convexity of $v$
applied to $A \cup B'$ and $B_{p+1}$
provides the following inequality:
\begin{equation}
\label{eqvAB'Bp+1+vAB'Bp+1>=vAB4+vBp+1}
v((A \cup B') \cup B_{p+1}) + v((A \cup B')
\cap B_{p+1}) \geq v(A \cup B') + v(B_{p+1}).
\end{equation}
By induction (\ref{eqlemsupermodularpropvAB+i=1pvABi}) is valid for $B'$:
\begin{equation}
\label{eqvAB'+i=1pvABi>=vA+i=1pvBi}
v(A \cup B') + \sum_{i=1}^{p} v(A \cap B_{i})
\geq v(A) + \sum_{i=1}^{p}v(B_{i}).
\end{equation}
Adding (\ref{eqvAB'Bp+1+vAB'Bp+1>=vAB4+vBp+1}) and (\ref{eqvAB'+i=1pvABi>=vA+i=1pvBi})
we obtain the result for $p+1$.
\end{proof}

\begin{lemma}
\label{lemPNABNPAPBAvNvB-vB>=vA-vA}
Let us consider a correspondence $\mathcal{P}$ on $2^N$
and subsets $A \subseteq B \subseteq N$
such that $\mathcal{P}(A) = \mathcal{P}(B)_{|A}$.
Let $\mathcal{F}$ be a weakly union-closed family of subsets of $N$
with $\emptyset \notin \mathcal{F}$.
If $A \in \mathcal{F}$
and if all elements of $\mathcal{P}(A)$ and $\mathcal{P}(B)$ are in $\mathcal{F}$,
then for every $\mathcal{F}$-convex game $(N,v)$
we have
\begin{equation}
\label{eqlemvB-vB>=vA-vA}
v(B) - \overline{v}(B) \geq v(A) - \overline{v}(A).
\end{equation}
\end{lemma}

\begin{proof}
If $\mathcal{P}(B) = \lbrace B_{1}, B_{2}, \ldots, B_{p} \rbrace$,
then $\mathcal{P}(A) = \lbrace B_i \cap A \mid i = 1, \ldots, p, B_i \cap A \not= \emptyset \rbrace$
and Lemma~\ref{lemmavB+i=1pvBiA} implies (\ref{eqlemvB-vB>=vA-vA}).
\end{proof}

\citet{Skoda201702} established that  inheritance of convexity 
and inheritance of convexity restricted to unanimity games
are equivalent for $\mathcal{P}_M$ and $\mathcal{P}_{\min}$.

\begin{theorem}
\label{theoremEquivalenceInheritanceofConvexityForUnanimityGamesP_MAndP_min}
Let $G = (N,E)$
(resp. $G=(N,E,w)$)
be a graph
(resp. weighted graph).
The following statements are equivalent:
\begin{enumerate}[1)]
\item
The Myerson restricted game $(N, u_S^{M})$
(resp. $\mathcal{P}_{\min}$-restricted game $(N, \overline{u_S})$)
is convex
for all $\emptyset \not= S \subseteq N$.
\item
There is inheritance of convexity for $\mathcal{P}_M$
(resp. $\mathcal{P}_{\min}$).
\end{enumerate}
\end{theorem}

Finally,
we recall a characterization of inheritance of convexity
for Myerson's correspondence $\mathcal{P}_{M}$.
A graph $G = (N,E)$ is \emph{cycle-complete}
if for any cycle $C = \lbrace v_1, e_1, v_2, e_2, \ldots , e_m, v_1 \rbrace$
in $G$ the subset $\lbrace v_1, v_2, \ldots, v_m \rbrace \subseteq N$
of vertices of $C$ induces a complete subgraph in $G$.

\begin{theorem}
\label{NouwelandandBorm1991}
\citep{NouwelandandBorm1991}.
Let $G = (N,E)$ be a graph.
There is inheritance of convexity for $\mathcal{P}_{M}$
if and only if $G$ is cycle-complete.
\end{theorem}

\section{
Inheritance of $\mathcal{F}$-convexity
}
\label{SectionNecessaryCondF-ConvPmin}

Let $G = (N,E,w)$ be a connected weighted graph
and let 
$\mathcal{F}$ be the family of connected subsets of $N$.
In this section we recall
necessary conditions on the weight vector $w$
established by~\citet{Skoda2017} 
for inheritance of $\mathcal{F}$-convexity
from the original communication game $(N,v)$ 
to the $\mathcal{P}_{\min}$-restricted game $(N,\overline{v})$.
We assume that all weights are strictly positive
and denote by $w_{k}$ or $w_{ij}$
the weight of an edge $e_{k} = \lbrace i,j \rbrace$ in $E$.

A star $S_k$ corresponds to a tree
with one internal vertex and $k$ leaves.
We consider a star $S_{3}$
with vertices ${1, 2, 3, 4}$ and edges $e_{1} = \lbrace 1, 2 \rbrace$,
$e_{2} = \lbrace 1, 3 \rbrace$ and $e_{3} = \lbrace 1, 4 \rbrace$.

\vspace{\baselineskip}

\begin{mdframed}
\textbf{Star Condition.}
\it
For every star in $G$ of type $S_{3}$,
the edge-weights
satisfy
\[
 w_{1} \leq w_{2} = w_{3},
\]
after renumbering the edges if necessary.
\end{mdframed}

\vspace{\baselineskip}
\begin{mdframed}
\textbf{Path Condition.}
\it
For every elementary path
$\gamma = \lbrace 1, e_{1}, 2, e_{2}, 3, \ldots,  m,\\ e_{m}, m+1 \rbrace$
in $G$
and for all $i,j,k$
such that $1 \leq i < j < k \leq m$,
the edge-weights satisfy
\[
w_{j} \leq \max (w_{i}, w_{k}).
\]
\end{mdframed}

\vspace{\baselineskip}
For a given cycle
$C = \lbrace 1, e_{1}, 2, e_{2}, \ldots, m, e_{m}, 1\rbrace$ with $m \geq 3$,
we denote by $E(C)$ the set of edges $\lbrace e_{1}, e_{2}, \ldots, e_{m}\rbrace$ of $C$
and by $\hat{E}(C)$ the set composed of $E(C)$ and of the chords of $C$ in $G$.
\\

\begin{mdframed}
\textbf{Cycle Condition.}
\it
For every simple cycle
$C = \lbrace 1, e_{1}, 2, e_{2}, \ldots, m, e_{m}, 1\rbrace$ in $G$ with $m \geq 3$,
the edge-weights satisfy
\[
w_{1} \leq w_{2} \leq w_{3} = \cdots = w_{m} = \hat{M},
\]
after renumbering the edges if necessary,
where $\hat{M} = \max_{e \in \hat{E}(C)} w(e)$.
Moreover,
$w(e) = w_{2}$ for all chord incident to $2$,
and $w(e) = \hat{M}$ for all $e \in \hat{E}(C)$ non-incident to $2$.
\end{mdframed}

\vspace{\baselineskip}
For a given cycle $C$,
an edge $e$ in $\hat{E}(C)$ is a \emph{maximum weight edge} of~$C$
if $w(e) = \max_{e \in \hat{E}(C)} w(e)$.
Otherwise,
$e$ is a non-maximum weight edge of~$C$.
Moreover,
we call maximum
(resp. non-maximum)
weight chord of $C$
a maximum
(resp. non-maximum)
weight edge in $\hat{E}(C) \setminus E(C)$.

\vspace{\baselineskip}
\begin{mdframed}
\textbf{Pan Condition.}
\it
For all connected subgraphs of $G$
corresponding to the union of a simple cycle $C = \lbrace e_{1}, e_{2}, \ldots, e_{m} \rbrace$
with $m \geq 3$,
and an elementary path $P$ such that there is an edge $e$ in $P$
with $w(e) \leq \min_{1 \leq k \leq m} w_{k}$
and $|V(C) \cap V(P)| = 1$,
the edge-weights satisfy
\begin{enumerate}[(a)]
\item
\label{eqPanConditionEither}
either $w_{1} = w_{2} = w_{3} = \cdots = w_{m}= \hat{M}$,
\item
\label{eqPanConditionOr}
or $w_{1} = w_{2} < w_{3} = \cdots = w_{m} = \hat{M}$,
\end{enumerate}
\noindent
where $\hat{M} = \max_{e \in \hat{E}(C)} w(e)$.
If Condition (\ref{eqPanConditionOr}) is satisfied,
then $V(C) \cap V(P) = \lbrace 2 \rbrace$,
and if moreover $w(e) < w_{1}$,
then $\lbrace 1, 3 \rbrace$ is a maximum weight chord of $C$.
\end{mdframed}

\vspace{\baselineskip}
Two cycles are said \textit{adjacent} if they share at least one common edge.

\vspace{\baselineskip}
\begin{mdframed}
\textbf{Adjacent Cycles Condition.}
\it
For all pairs $\lbrace C,C' \rbrace$ of adjacent simple cycles in $G$
such that
\begin{enumerate}[(a)]
\item
\label{enumPropV(C)-V(C)'nonempty}
$V(C) \setminus V(C') \not= \emptyset$ and $V(C') \setminus V(C) \not= \emptyset$,
\item
\label{enumPropCatmost1non-maxweightchord}
$C$ has at most one non-maximum weight chord,
\item
\label{enumPropCC'nomaxweightchord}
$C$ and $C'$ have no maximum weight chord,
\item
\label{enumPropCandC'HaveNoCommonChord}
$C$ and $C'$ have no common chord.
\end{enumerate}
\noindent
$C$ and $C'$ cannot have two common non-maximum weight edges.
Moreover,
$C$ and $C'$ have a unique common non-maximum weight edge $e_{1}$ if and only if 
there are  non-maximum weight edges $e_{2} \in E(C) \setminus E(C')$ and $e_{2}' \in E(C') \setminus E(C)$
such that $e_{1}, e_{2}, e_{2}'$ are adjacent and
\begin{itemize}
\item $w_{1} = w_{2} = w_{2}'$ if $|E(C)| \geq 4$ and $|E(C')| \geq 4$,
\item $w_{1} = w_{2} \geq w_{2}'$ or $w_{1} = w_{2}' \geq w_{2}$
if $|E(C)| = 3$ or $|E(C')| = 3$.
\end{itemize}
\end{mdframed}

\begin{proposition}
\label{corPathCond}
Let $\mathcal{F}$ be the family of connected subsets of $N$.
If for all $\emptyset \not= S \subseteq N$,
the $\mathcal{P}_{\min}$-restricted game $(N, \overline{u_{S}})$ is $\mathcal{F}$-convex,
then the Star, Path, Cycle, Pan and Adjacent cycles conditions are satisfied.
\end{proposition}


Let us note that Proposition~\ref{corPathCond} only requires inheritance of $\mathcal{F}$-convexity
for the unanimity games
to obtain the necessity of the five previous conditions.
\citet{Skoda2017} also proved that these necessary conditions are sufficient
for inheritance  of $\mathcal{F}$-convexity
if we consider superadditive games.

\begin{theorem}
\label{thPathBranchCyclePanAdjacentCond}
Let $\mathcal{F}$ be the family of connected subsets of $N$.
The $\mathcal{P}_{\min}$-restricted game $(N, \overline{v})$ is $\mathcal{F}$-convex
for every superadditive and $\mathcal{F}$-convex game $(N,v)$
if and only if
the Star, Path, Cycle, Pan, and Adjacent cycles conditions are satisfied.
\end{theorem}

\section{Inheritance of convexity}
\label{SectionInheritanceOfConvexityWithPmin}
We consider in this section inheritance of convexity.
As convexity implies
superadditivity and
$\mathcal{F}$-convexity
(where $\mathcal{F}$ is the family of connected subsets of $N$),
the conditions stated in Section~\ref{SectionNecessaryCondF-ConvPmin} are necessary.
We now have to deal with disconnected subsets of $N$.
We establish supplementary necessary conditions
implying strong restrictions on edge-weights.
In particular,
we obtain that edge-weights can have at most three different values.
We first need the following lemmas.

\begin{lemma}
\label{LemP_MinMax(w1,w2)<w(e)e'InELinkingeTo2}
Let us assume that for all $\emptyset \not= S \subseteq N$
the $\mathcal{P}_{\min}$-restricted game $(N,\overline{u_{S}})$ is convex.
Let $e_{1} = \lbrace 1,2 \rbrace$ and $e_{2} = \lbrace 2,3 \rbrace$
be two adjacent edges, and $e$ be an edge such that
\begin{equation}\label{eqLemMax(w1,w2)<w(e)}
\max (w_{1}, w_{2}) < w(e).
\end{equation}
Then,
there exists an edge $e' \in E$ linking $e$ to vertex $2$.
Moreover,
if $e' \not= e_{1}$
and $e_{2}$,
then
$w(e') = \max(w_{1}, w_{2})$ if $w_{1} \not= w_{2}$
and $w(e') \leq w_{1} = w_{2}$ otherwise.
\end{lemma}

We represent in Figure~\ref{figLemma,wMax(w1,w2)<w(e),wMax(w1,w2)<w(e),wMax(w1,w2)<w(e)3}
the possible situations corresponding 
to Lemma~\ref{LemP_MinMax(w1,w2)<w(e)e'InELinkingeTo2}.
\begin{figure}[!h]
\centering
\subfloat{
\begin{pspicture}(0,-.3)(0,.6)
\tiny
\begin{psmatrix}[mnode=circle,colsep=.7,rowsep=.4]
{}	& {}\\
{$1$} 	& {$2$}	& {$3$}
\psset{arrows=-, shortput=nab,labelsep={.15}}
\tiny
\ncline{1,1}{1,2}^{$e$}
\ncline{1,2}{2,2}^{$e'$}
\ncline{2,1}{2,2}_{$e_{1}$}
\ncline{2,2}{2,3}_{$e_{2}$}
\end{psmatrix}
\end{pspicture}
}
\subfloat{
\begin{pspicture}(-1,-.3)(1,.6)
\tiny
\begin{psmatrix}[mnode=circle,colsep=.7,rowsep=.4]
{}\\
{$1$} 	& {$2$}	& {$3$}
\psset{arrows=-, shortput=nab,labelsep={.15}}
\tiny
\ncline{1,1}{2,1}^{$e$}
\ncline{2,1}{2,2}_{$e_{1}=e'$}
\ncline{2,2}{2,3}_{$e_{2}$}
\end{psmatrix}
\end{pspicture}
}
\subfloat{
\begin{pspicture}(0,-.3)(0,.6)
\tiny
\begin{psmatrix}[mnode=circle,colsep=.7,rowsep=.4]
{$1$} 	& {$2$}	& {$3$}
\psset{arrows=-, shortput=nab,labelsep={.15}}
\tiny
\ncarc[arcangle=60]{1,1}{1,3}^{$e$}
\ncline{1,1}{1,2}_{$e_{1}=e'$}
\ncline{1,2}{1,3}_{$e_{2}$}
\end{psmatrix}
\end{pspicture}
}
\caption{$e'$ linking $e$ to vertex $2$.}
\label{figLemma,wMax(w1,w2)<w(e),wMax(w1,w2)<w(e),wMax(w1,w2)<w(e)3}
\end{figure}
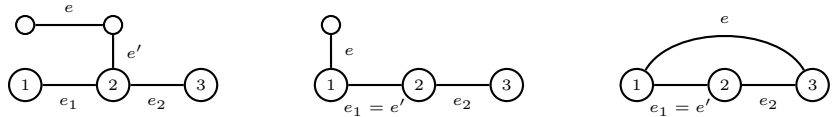

\begin{proof}
We set $e = \lbrace j, k \rbrace$.
Star condition implies $j\not= 2$ and $k\not=2$
(otherwise it contradicts (\ref{eqLemMax(w1,w2)<w(e)})).
By contradiction,
let us assume that there is no edge linking $e$ to $2$.
We can assume $w_{1} \leq w_{2} < w(e)$.
Let us consider $i=3$,
$A_{1}= \lbrace 2 \rbrace$, $A_{2} = \lbrace j,k \rbrace$,
$A = A_{1} \cup A_{2}$, and $B = A \cup \lbrace 1 \rbrace$
as represented in Figure~\ref{figLemmaProofwMax(w1,w2)<w(e)}.
\begin{figure}[!h]
\centering
\begin{pspicture}(0,-.4)(0,1.6)
\tiny
\begin{psmatrix}[mnode=circle,colsep=.7,rowsep=.5]
{$j$}	& {$k$}\\
{$1$} 	& {$2$}	& {$3$}
\psset{arrows=-, shortput=nab,labelsep={.1}}
\tiny
\ncline{1,1}{1,2}^{$e$}
\ncline[linestyle=dashed]{1,1}{2,1}
\ncline[linestyle=dashed]{1,2}{2,1}
\ncline[linestyle=dashed]{1,1}{2,3}
\ncline[linestyle=dashed]{1,2}{2,3}
\ncline{2,1}{2,2}_{$e_{1}$}
\ncline{2,2}{2,3}_{$e_{2}$}
\end{psmatrix}
\normalsize
\uput[0](-1.2,1){\textcolor{cyan}{$A_{2}$}}
\pspolygon[linecolor=cyan,linearc=.4,linewidth=.02](-2.9,.6)(-2.9,1.4)(-1.05,1.4)(-1.05,.6)
\uput[0](-1.5,-.45){\textcolor{cyan}{$A_{1}$}}
\pspolygon[linecolor=cyan,linearc=.3,linewidth=.02](-1.7,-.3)(-1.7,.3)(-1.1,.3)(-1.1,-.3)
\uput[0](-.1,0){\textcolor{blue}{$\scriptstyle{=i}$}}
\uput[0](-3.6,.55){\textcolor{cyan}{$B$}}
\pspolygon[linecolor=cyan,linearc=.4,linewidth=.02](-3,-.7)(-3,1.5)(-.5,1.5)(-.5,-.7)
\end{pspicture}
\caption{$w_{1} \leq w_{2} < w(e)$.}
\label{figLemmaProofwMax(w1,w2)<w(e)}
\end{figure}
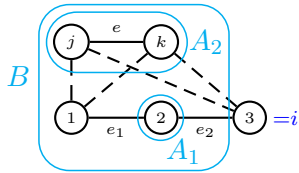
As no edge links $e$ to $2$,
we have $\mathcal{P}_{\min}(A) = \lbrace \lbrace 2 \rbrace, \lbrace j \rbrace, \lbrace k \rbrace \rbrace$.
Let us note that, as $w_{2} < w(e)$
(resp. $w_{1}< w(e)$),
there is a component $A'$
(resp. $B'$)
of $\mathcal{P}_{\min}(A \cup \lbrace i \rbrace)$
(resp. $\mathcal{P}_{\min}(B)$)
containing $A_{2}$.
Then,
$\mathcal{P}_{\min}(B)_{|A' \cap A} \not= \mathcal{P}_{\min}(A)_{|A' \cap A}$
as $\mathcal{P}_{\min}(A)$ corresponds to a singleton partition
but $B{'} \cap A'$ contains $A_{2}$.
It contradicts Theorem~\ref{thLGNuPNFNFSNuSSNNuSFABFABF}
applied with $\mathcal{F} = 2^N \setminus \lbrace \emptyset \rbrace$.
Therefore,
there exists an edge $e'$ linking $e$ to $2$.
Finally,
if $e' \not= e_{1}$ and $e_{2}$,
then Star condition applied to $\lbrace e_{1}, e_{2}, e' \rbrace$ implies
the result.
\end{proof}


\begin{lemma}
\label{lemMax(w1,w2)<Min(w(e),w(e')}
Let us assume that
for all $\emptyset \not= S \subseteq N$
the $\mathcal{P}_{\min}$-restricted game $(N,\overline{u_{S}})$ is convex.
Let $e_{1} = \lbrace 1,2 \rbrace$ and $e_{2} = \lbrace 2,3 \rbrace$
be two adjacent edges and let $e$ and $e'$ be two edges in $E$ such that
\begin{equation}
\max (w_{1}, w_{2}) < \min (w(e),w(e')).
\end{equation}
Then,
$w(e) = w(e')$.
\end{lemma}

\begin{proof}
We can assume $w_{1} \leq w_{2}$.
By contradiction,
let us assume $w(e) < w(e')$.
Applying Lemma~\ref{LemP_MinMax(w1,w2)<w(e)e'InELinkingeTo2} to $e$ (resp. $e'$),
there exists an edge $e_{2}'$ (resp. $e_{2}''$) linking $e$ (resp. $e'$) to $2$
such that $w_{2}' \leq \max (w_{1}, w_{2}) < w(e)$
(resp. $w_{2}'' \leq \max (w_{1}, w_{2}) < w(e')$)
($e_{2}'$ (resp. $e_{2}''$) may coincide with $e_{1}$ or $e_{2}$).
We set $e_{2}' = \lbrace 2, 2' \rbrace$ (resp. $e_{2}'' = \lbrace 2, 2'' \rbrace$)
where $2'$ (resp. $2''$) is an end-vertex of $e$ (resp. $e'$)
as represented in Figure~\ref{figLemmaMax(w1,w2)<Min(w(e),w(e'))}.
\begin{figure}[!h]
\centering
\begin{pspicture}(0,0)(0,1.4)
\tiny
\begin{psmatrix}[mnode=circle,colsep=.2,rowsep=.3]
{}	& & {$2'$} & & {$2''$} & &  {}\\
& {$1$}  & 	& {$2$}	& & {$3$}
\psset{arrows=-, shortput=nab,labelsep={.05}}
\tiny
\ncline{1,1}{1,3}^{$e$}
\ncline{1,5}{1,7}^{$e'$}
\ncline{1,3}{2,4}_{$e_{2}'$}
\ncline{2,4}{1,5}_{$e_{2}''$}
\ncline{2,2}{2,4}_{$e_{1}$}
\ncline{2,4}{2,6}_{$e_{2}$}
\end{psmatrix}
\end{pspicture}
\caption{$w_{2}' \leq \max (w_{1}, w_{2}) < w(e)$ and
$w_{2}'' \leq \max (w_{1}, w_{2}) < w(e')$.}
\label{figLemmaMax(w1,w2)<Min(w(e),w(e'))}
\end{figure}
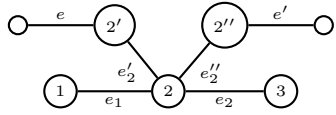
If $2' = 2''$,
then $e_{2}' = e_{2}''$
and as $w_{2}' < w(e) < w(e')$ it contradicts
the Star condition applied to $\lbrace e_{2}', e, e' \rbrace$.
Otherwise,
as $w_{2}' < w(e) < w(e')$,
Lemma~\ref{LemP_MinMax(w1,w2)<w(e)e'InELinkingeTo2}
applied to $e'$ and the pair of adjacent edges $\lbrace e_{2}', e \rbrace$
implies the existence of an edge $e'' \in E$ linking $e'$ to $2'$
($e''$ can coincide with $e$).
Let us first assume $e'' = \lbrace 2', 2'' \rbrace$
as represented in Figure~\ref{figLemmaMax(w1,w2)<Min(w(e),w(e'))2}.
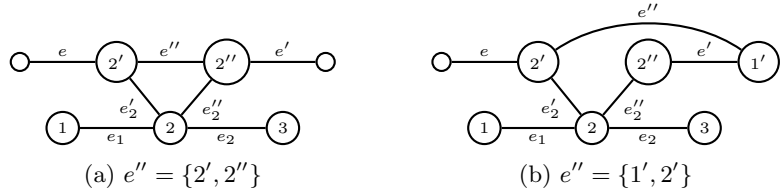
\begin{figure}[!h]
\centering
\subfloat[$e'' = \lbrace 2', 2'' \rbrace$]{
\begin{pspicture}(-.5,-.3)(.5,1.4)
\tiny
\begin{psmatrix}[mnode=circle,colsep=.2,rowsep=.3]
{}	& & {$2'$} & & {$2''$} & &  {}\\
& {$1$}  & 	& {$2$}	& & {$3$}
\psset{arrows=-, shortput=nab,labelsep={.05}}
\tiny
\ncline{1,1}{1,3}^{$e$}
\ncline{1,5}{1,7}^{$e'$}
\ncline{1,3}{1,5}^{$e''$}
\ncline{1,3}{2,4}_{$e_{2}'$}
\ncline{2,4}{1,5}_{$e_{2}''$}
\ncline{2,2}{2,4}_{$e_{1}$}
\ncline{2,4}{2,6}_{$e_{2}$}
\end{psmatrix}
\end{pspicture}
\label{figLemmaMax(w1,w2)<Min(w(e),w(e'))2}
}
\subfloat[$e'' = \lbrace 1', 2' \rbrace$]{
\begin{pspicture}(-.5,-.3)(.5,1.4)
\tiny
\begin{psmatrix}[mnode=circle,colsep=.2,rowsep=.3]
{}	& & {$2'$} & & {$2''$} & &  {$1'$}\\
& {$1$}  & 	& {$2$}	& & {$3$}
\psset{arrows=-, shortput=nab,labelsep={.05}}
\tiny
\ncline{1,1}{1,3}^{$e$}
\ncline{1,5}{1,7}^{$e'$}
\ncarc[arcangle=35]{1,3}{1,7}^{$e''$}
\ncline{1,3}{2,4}_{$e_{2}'$}
\ncline{2,4}{1,5}_{$e_{2}''$}
\ncline{2,2}{2,4}_{$e_{1}$}
\ncline{2,4}{2,6}_{$e_{2}$}
\end{psmatrix}
\end{pspicture}
\label{figLemmaMax(w1,w2)<Min(w(e),w(e'))3}
}
\caption{$e''$ linking $e'$ to $2'$.}
\end{figure}
As $w_{2}' < w(e)$ (resp. $w_{2}'' < w(e')$),
Star condition applied to $\lbrace e_{2}', e, e'' \rbrace$
(resp. $\lbrace e_{2}'', e', e'' \rbrace$)
implies $w(e'') = w(e)$ (resp. $w(e'') = w(e')$)
and then $w(e) = w(e')$,
a contradiction.
Let us now assume $e' = \lbrace 1', 2'' \rbrace$ and $e'' = \lbrace 1', 2' \rbrace$
as represented in Figure~\ref{figLemmaMax(w1,w2)<Min(w(e),w(e'))3}.
Then,
there is a cycle $C = \lbrace 2, e_{2}', 2', e'', 1', e', 2'', e_{2}'', 2 \rbrace$.
As $w_{2}' < w(e)$, Star condition applied to $\lbrace e_{2}', e, e'' \rbrace$
implies $w(e'') = w(e)$.
As $w(e) < w(e')$, we get $w(e'') < w(e')$.
Therefore,
$C$ has three non-maximum weight edges $e_{2}'$, $e_{2}''$, $e''$,
contradicting the Cycle condition.
\end{proof}

For a given weighted graph $G=(N,E,w)$,
let $\lbrace \sigma_{1}, \sigma_{2}, \ldots, \sigma_{k} \rbrace$
be the set of its edge-weights
such that $\sigma_{1} < \sigma_{2}< \ldots< \sigma_{k}$ with $1 \leq k \leq |E|$.
We denote by $E_{i}$ the set of edges in $E$ with weight $\sigma_{i}$,
and by $N_{i}$ the set of end-vertices of edges in $E_{i}$.

\begin{lemma}
\label{lem1)w1=sigma1(G)w2=Sigma2(G)2)GN1Connected}
Let $G= (N,E,w)$ be a connected weighted graph
satisfying the Path condition
and with at least two different edge-weights
$\sigma_{1}$, $\sigma_{2}$.
Then
\begin{enumerate}
\item
\label{item1lem1)w1=sigma1(G)w2=Sigma2(G)2)GN1Connected}
There exists a pair of adjacent edges $\lbrace e_{1}, e_{2} \rbrace$
with $w_{1} = \sigma_{1}$ and $w_{2} = \sigma_{2}$.
\item
$G_{N_{1}}$ is connected.
\end{enumerate}
\end{lemma}

\begin{proof}
\begin{enumerate}
\item
If there is no such pair,
then it contradicts the Path condition.
\item
Let $C'$, $C''$ be two distinct connected components of $G_{N_{1}}$.
By definition of $N_{1}$, $C'$ (resp. $C''$) contains at least one edge $e'$ (resp. $e''$)
of weight $\sigma_{1}$.
Let $\gamma$ be a shortest path in $G$ linking $e'$ to $e''$.
Path condition applied to $\gamma' = e' \cup \gamma \cup e''$
implies $w(e) = \sigma_{1}$
for all edge $e$ in $\gamma$.
Then,
$\gamma'$ is a path in $G_{N_{1}}$ linking $C'$ to $C''$,
a contradiction.
\qedhere
\end{enumerate}
\end{proof}

The following proposition is a direct consequence of Proposition~\ref{corPathCond}
and Lemmas~\ref{lemMax(w1,w2)<Min(w(e),w(e')}
and~\ref{lem1)w1=sigma1(G)w2=Sigma2(G)2)GN1Connected}.

\begin{proposition}
\label{propAtMost3DifferentEdgeWeights}
Let $G= (N,E,w)$ be a connected weighted graph.
Let us assume that
for all $\emptyset \not= S \subseteq N$
the $\mathcal{P}_{\min}$-restricted game $(N,\overline{u_{S}})$ is convex.
Then,
the edge-weights have at most three different values
$\sigma_{1} < \sigma_{2} < \sigma_{3}$.
Moreover,
if $|E_{1}| \geq 2$,
then they have at most two different values 
$\sigma_{1} < \sigma_{2}$.
\end{proposition}

Of course,
Proposition~\ref{propAtMost3DifferentEdgeWeights} implies that
if the edge-weights have three different values,
then there is only one edge with minimum weight $\sigma_{1}$.
We will now establish necessary conditions on adjacency and incidence of
edges in $E_1$, $E_2$, $E_3$.
We first need the two following lemmas. 

\begin{lemma}
\label{lemw1<w2<w3e_{1}e_{2}IncidentTove3Adjacenttoe1ore2ButNotIncidentTov}
Let $G= (N,E,w)$ be a connected weighted graph.
Let us assume that
for all $\emptyset \not= S \subseteq N$,
the
$\mathcal{P}_{\min}$-restricted game $(N,\overline{u_{S}})$ is convex
and  that the edge-weights have exactly three different values
$\sigma_{1} < \sigma_{2} < \sigma_{3}$.
Then,
there exist three edges $e_{1}, e_{2}, e_{3}$
with respective weights $\sigma_{1}$, $\sigma_{2}$, $\sigma_{3}$
such that $e_{1}$ and $e_{2}$ are incident to a vertex $v$
and $e_{3}$ is adjacent to $e_{1}$ or $e_{2}$ but not incident to $v$.
\end{lemma}

Three edges $e_{1}$, $e_{2}$, $e_{3}$ satisfying
Lemma~\ref{lemw1<w2<w3e_{1}e_{2}IncidentTove3Adjacenttoe1ore2ButNotIncidentTov}
correspond to three possible situations
represented in Figure~\ref{figleme_{1}e_{2}IncidentTove3Adjacenttoe1ore2ButNotIncidentTov}.

\begin{figure}[!h]
\centering
\subfloat[]{
\begin{pspicture}(0,-.3)(0,.5)
\tiny
\begin{psmatrix}[mnode=circle,colsep=0.5,rowsep=0.3]
{$1$}	& {$2$}	& {$3$} &  {$4$}
\psset{arrows=-, shortput=nab,labelsep={0.1}}
\tiny
\ncline{1,1}{1,2}_{$e_{1}$}
\ncline{1,2}{1,3}_{$e_{2}$}
\ncline{1,3}{1,4}_{$e_{3}$}
\normalsize
\end{psmatrix}
\end{pspicture}
}
\subfloat[]{
\begin{pspicture}(-1,-.3)(1,.5)
\tiny
\begin{psmatrix}[mnode=circle,colsep=0.3,rowsep=0]
	&	{$3$}\\
{$1$}	& 		&  {$2$} 
\psset{arrows=-, shortput=nab,labelsep={0.1}}
\tiny
\ncline{2,1}{2,3}_{$e_{1}$}
\ncline{2,1}{1,2}^{$e_{3}$}
\ncline{1,2}{2,3}^{$e_{2}$}
\normalsize
\end{psmatrix}
\end{pspicture}
}
\subfloat[]{
\begin{pspicture}(0,-.3)(0,.5)
\tiny
\begin{psmatrix}[mnode=circle,colsep=0.5,rowsep=0.3]
{$4$}	& {$1$}	& {$2$} &  {$3$}
\psset{arrows=-, shortput=nab,labelsep={0.1}}
\tiny
\ncline{1,1}{1,2}_{$e_{3}$}
\ncline{1,2}{1,3}_{$e_{1}$}
\ncline{1,3}{1,4}_{$e_{2}$}
\normalsize
\end{psmatrix}
\end{pspicture}
}
\caption{$w_{1} = \sigma_{1} < w_{2} = \sigma_{2} < w_{3} = \sigma_{3}$.}
\label{figleme_{1}e_{2}IncidentTove3Adjacenttoe1ore2ButNotIncidentTov}
\end{figure}
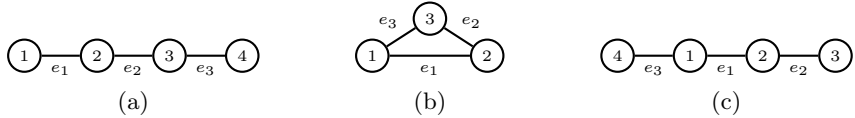

\begin{proof}
Claim~\ref{item1lem1)w1=sigma1(G)w2=Sigma2(G)2)GN1Connected} of Lemma~\ref{lem1)w1=sigma1(G)w2=Sigma2(G)2)GN1Connected}
implies the existence of $e_{1}$ and $e_{2}$.
Proposition~\ref{propAtMost3DifferentEdgeWeights} implies the uniqueness of $e_{1}$.
Let us assume $e_{1} = \lbrace 1, 2 \rbrace$
and $e_{2} = \lbrace 2, 3 \rbrace$,
and let $e_{3}$ be an edge of weight $\sigma_{3}$.
By Lemma~\ref{LemP_MinMax(w1,w2)<w(e)e'InELinkingeTo2}
there exists an edge $e'$ linking $e_{3}$ to $2$
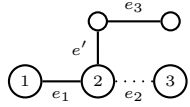
\begin{figure}[!h]
\centering
\begin{pspicture}(0,-.2)(0,1.2)
\tiny
\begin{psmatrix}[mnode=circle,colsep=0.5,rowsep=0.4]
			&	{}		&	{}\\
{$1$}	& {$2$}	& {$3$}
\psset{arrows=-, shortput=nab,labelsep={0.1}}
\tiny
\ncline{2,1}{2,2}_{$e_{1}$}
\ncline[linestyle=dotted]{2,2}{2,3}_{$e_{2}$}
\ncline{1,2}{2,2}_{$e'$}
\ncline{1,2}{1,3}^{$e_{3}$}
\normalsize
\end{psmatrix}
\end{pspicture}
\caption{
$e'$ linking $e_{3}$ to $2$ with
$w_{1} = \sigma_{1} < w_{2} = \sigma_{2} < w_{3} = \sigma_{3}$.}
\label{figProofLeme_{1}e_{2}IncidentTove3Adjacenttoe1ore2ButNotIncidentTov}
\end{figure}
and if $e' \not= e_{1}$ and $e' \not= e_{2}$,
then $w(e') = \max(w_{1}, w_{2}) = \sigma_{2}$.
In this last case,
we can substitute $e'$ for $e_{2}$
as represented in Figure~\ref{figProofLeme_{1}e_{2}IncidentTove3Adjacenttoe1ore2ButNotIncidentTov},
and then the $3$-tuple $\lbrace e_{1}, e', e_{3} \rbrace$ satisfies the conclusion of the lemma.
\end{proof}

\begin{lemma}
\label{propMax(w1,w2)<=w3=w4=...=wm}
Let us assume that
for all $\emptyset \not= S \subseteq N$,
the
$\mathcal{P}_{\min}$-restricted game $(N,\overline{u_{S}})$ is convex.
Then,
for all elementary path
$\gamma = \lbrace 1, e_{1}, 2, e_{2}, \ldots, m, e_{m},\\ m+1 \rbrace$
with $w_{1} < w_{m}$, we have
\begin{equation}
\label{eqPropMax(w_{1},w_{2})<=w(3)=W(4)=...=w(m)}
\max (w_{1}, w_{2}) \leq w_{3} = w_{4} = \cdots = w_{m}.
\end{equation}
\end{lemma}

\begin{proof}
The Path condition implies
$w_{j} \leq \max (w_{1}, w_{m}) = w_{m}$ for all $j$, $1 \leq j \leq m$.
Then,
(\ref{eqPropMax(w_{1},w_{2})<=w(3)=W(4)=...=w(m)}) is satisfied if $m=3$.
Let us assume $m \geq 4$
and
by contra\-diction $w_{m-1} < w_{m}$.
The Path condition implies $w_{2} \leq \max (w_{1}, w_{m-1}) < w_{m}$.
Therefore,
$\max(w_{1}, w_{2}) < w_{m}$.
Then,
by Lemma~\ref{LemP_MinMax(w1,w2)<w(e)e'InELinkingeTo2},
there exists an edge $e$ linking $e_{m}$ to $2$
with $w(e) \leq \max (w_{1}, w_{2})$.
Hence,
we have $w(e) < w_{m}$.
Let us first assume $e = \lbrace 2, m \rbrace$
as represented in Figure~\ref{figProofPrope=(2,m)}.
\begin{figure}[!h]
\centering
\subfloat[]{
\begin{pspicture}(0,-.3)(0,.8)
\tiny
\begin{psmatrix}[mnode=circle,colsep=.5,rowsep=.6]
{$1$}	& {$2$} & {$3$} & {$m$} &  {$\scriptscriptstyle{m+1}$}
\psset{arrows=-, shortput=nab,labelsep={.1}}
\tiny
\ncline{1,1}{1,2}_{$e_{1}$}
\ncline{1,2}{1,3}_{$e_{2}$}
\ncline[linestyle=dashed]{1,3}{1,4}
\ncline{1,4}{1,5}_{$e_{m}$}
\ncarc[arcangle=40]{1,2}{1,4}^{$e$}
\end{psmatrix}
\end{pspicture}
\label{figProofPrope=(2,m)}
}
\hspace{.6cm}
\subfloat[]{
\begin{pspicture}(0,-.3)(0,.8)
\tiny
\begin{psmatrix}[mnode=circle,colsep=.5,rowsep=.6]
{$1$}	& {$2$} & {$3$} & {$m$} &  {$\scriptscriptstyle{m+1}$}
\psset{arrows=-, shortput=nab,labelsep={.1}}
\tiny
\ncline{1,1}{1,2}_{$e_{1}$}
\ncline{1,2}{1,3}_{$e_{2}$}
\ncline[linestyle=dashed]{1,3}{1,4}
\ncline{1,4}{1,5}_{$e_{m}$}
\ncarc[arcangle=30]{1,2}{1,5}^{$e$}
\end{psmatrix}
\end{pspicture}
\label{figProofPrope=(2,m+1)}
}
\caption{$e = \lbrace 2, m \rbrace$ or $e = \lbrace 2, m+1 \rbrace$.}
\end{figure}
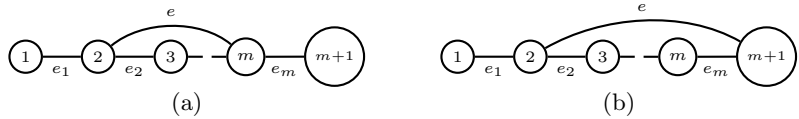
Then,
Star condition applied to $\lbrace e_{m-1}, e_{m}, e \rbrace$ implies
$w_{m-1} = w_{m}$, a contradiction.
Let us now assume $e = \lbrace 2, m+1 \rbrace$
as represented in Figure~\ref{figProofPrope=(2,m+1)}.
Then,
the cycle $C = \lbrace 2, e_{2}, 3 \ldots, m, e_{m}, m+1, e, 2 \rbrace$
contains at least three non-maximum weight edges ($e_{2}$, $e_{m-1}$, $e$),
contradicting the Cycle condition.
Hence,
we have $w_{m-1} = w_{m}$
and therefore $w_{m-1} > w_{1}$.
We can iterate to get
(\ref{eqPropMax(w_{1},w_{2})<=w(3)=W(4)=...=w(m)}).
\end{proof}


\begin{proposition}
\label{propEitherOnlyOneEdgeOfWeightw1OrAllEdgesOFWeightw1AreIncidentToSameVertex}
Let $G= (N,E,w)$ be a connected weighted graph.
Let us assume that
for all $\emptyset \not= S \subseteq N$,
the $\mathcal{P}_{\min}$-restricted game $(N,\overline{u_{S}})$ is convex
and that the edge-weights have exactly two different values
$\sigma_{1} < \sigma_{2}$.
Then,
either
$|E_{1}| = 1$
or all edges
in $E_{1}$
are incident to the same vertex $v$
and no edge
in $E_{2}$
is incident to $v$.
\end{proposition}

Graphs corresponding to the two possible situations
described in Proposition~\ref{propEitherOnlyOneEdgeOfWeightw1OrAllEdgesOFWeightw1AreIncidentToSameVertex}
are represented in Figure~\ref{figPropEitherOnlyOneEdgeOfWeightw1OrAllEdgesOFWeightw1AreIncidentToSameVertex}.
\begin{figure}[!h]
\centering
\subfloat{
\begin{pspicture}(-.5,0)(.5,.8)
\tiny
\begin{psmatrix}[mnode=circle,colsep=0.5,rowsep=0.4]
	{} &			 &			 &	{}\\
 {} & {$1$} & {$2$} & {}
\psset{arrows=-, shortput=nab,labelsep={0.05}}
\tiny
\ncline{1,1}{2,1}_{$\sigma_{2}$}
\ncline{2,1}{2,2}_{$\sigma_{2}$}
\ncline{2,2}{2,3}_{$\sigma_{1}$}^{$e_{1}$}
\ncline{2,3}{2,4}_{$\sigma_{2}$}
\ncline{1,1}{2,2}^{$\sigma_{2}$}
\ncline{2,3}{1,4}_{$\sigma_{2}$}
\normalsize
\end{psmatrix}
\end{pspicture}
}
\subfloat{
\begin{pspicture}(-.5,0)(.5,.8)
\tiny
\begin{psmatrix}[mnode=circle,colsep=0.5,rowsep=0.4]
{} 	& {} 		& {}\\
{}	& {$v$}	& {} 
\psset{arrows=-, shortput=nab,labelsep={0.05}}
\tiny
\ncline{1,1}{2,1}_{$\sigma_{2}$}
\ncline{1,1}{2,2}^{$\sigma_{1}$}
\ncline{1,2}{2,2}^{$\sigma_{1}$}
\ncline{1,2}{1,3}^{$\sigma_{2}$}
\ncline{1,3}{2,3}^{$\sigma_{2}$}
\ncline{2,2}{2,3}_{$\sigma_{1}$}
\ncline{2,1}{2,2}_{$\sigma_{1}$}
\end{psmatrix}
\end{pspicture}
}
\caption{Edge-weights with only two values $\sigma_{1} < \sigma_{2}$.}
\label{figPropEitherOnlyOneEdgeOfWeightw1OrAllEdgesOFWeightw1AreIncidentToSameVertex}
\end{figure}
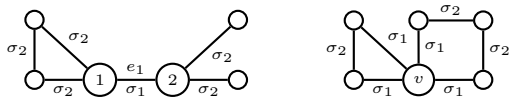

\begin{proof}
As $G$ is connected there is at least one pair of adjacent edges
$e_{1} = \lbrace 1, 2 \rbrace$,
and $e_{2} = \lbrace 2, 3 \rbrace$
with weights $\sigma_{1} < \sigma_{2}$.
Let $e$ be an edge in $E_{1} \setminus \lbrace e_{1} \rbrace$
non-incident to $1$.
Then,
$e$ cannot be incident to $2$ otherwise it contradicts the Star condition.
As $G$ is connected,
there exists a shortest path $\gamma$
linking $e$ to $1$ or $2$.
All edges in $\gamma$ have weight $\sigma_{1}$
by the Path condition
applied to $\lbrace e \rbrace \cup \gamma \cup \lbrace e_{1} \rbrace$.
Let $e'$ be the edge of $\gamma$ incident to $1$ or $2$.
As $w(e') = \sigma_{1}$,
the first part of the proof implies that $e'$ is necessarily incident to $1$.
Then,
Lemma~\ref{propMax(w1,w2)<=w3=w4=...=wm}
applied to
$\gamma' = \lbrace e \rbrace \cup \gamma \cup \lbrace e_{1}, e_{2}\rbrace$
as represented in Figure~\ref{figProofAllEdgesOfWeightw2AreIncidentToSameEnd-VertexOfe4-2}
implies $\sigma_{1} = \sigma_{2}$,
a contradiction.
\begin{figure}[!h]
\centering
\begin{pspicture}(0,0)(0,1.2)
\tiny
\begin{psmatrix}[mnode=circle,colsep=0.5,rowsep=0.5]
{}		&		{}\\
{$1$}	& {$2$}	& {$3$}
\psset{arrows=-, shortput=nab,labelsep={0.06}}
\tiny
\ncline{1,1}{2,1}^{$\sigma_{1}$}_{$e'$}
\ncline{1,1}{1,2}^{$e$}_{$\sigma_{1}$}
\ncline{2,1}{2,2}_{$e_{1}$}^{$\sigma_{1}$}
\ncline{2,2}{2,3}_{$e_{2}$}^{$\sigma_{2}$}
\normalsize
\end{psmatrix}
\end{pspicture}
\caption{$\gamma' = \lbrace e \rbrace \cup \gamma \cup \lbrace e_{1}, e_{2} \rbrace$.}
\label{figProofAllEdgesOfWeightw2AreIncidentToSameEnd-VertexOfe4-2}
\end{figure}
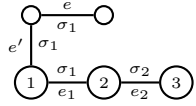
Therefore,
$e$ is incident to $1$,
and any edge of weight $\sigma_{2}$ incident to $1$ would contradict the Star condition.
\end{proof}

\begin{proposition}
\label{lemAllEdgesOfWeightw2AreIncidentToSameEndVertexofe1}
Let $G= (N,E,w)$ be a connected weighted graph.
Let us assume that
for all $\emptyset \not= S \subseteq N$,
the $\mathcal{P}_{\min}$-restricted game $(N,\overline{u_{S}})$ is convex
and that the edge-weights have exactly three different values
$\sigma_{1} < \sigma_{2} < \sigma_{3}$.
Then,
there is only one edge $e_{1}$
in $E_{1}$,
every edge in $E_{2}$
is incident to the same end-vertex $v$ of $e_{1}$,
and every edge in $E_{3}$
is non-incident to $v$ but linked to $v$
by $e_{1}$ or by an edge
in $E_{2}$.
\end{proposition}
We give in Figure~\ref{figAllEdgesOfWeightw2AreIncidentToSameEnd-VertexOfe1}
an example of a graph
satisfying the conditions of Proposition~\ref{lemAllEdgesOfWeightw2AreIncidentToSameEndVertexofe1}.
\begin{figure}[!h]
\centering
\begin{pspicture}(0,0)(0,1.3)
\tiny
\begin{psmatrix}[mnode=circle,colsep=0.7,rowsep=0.1]
			&			&	 {} & {}\\
{$1$}	& {$2$}	& {$3$}\\
			&			&	 {}
\psset{arrows=-, shortput=nab,labelsep={0.05}}
\tiny
\ncline{2,1}{2,2}_{$\sigma_{1}$}
\ncline{2,2}{2,3}_{$\sigma_{2}$}
\ncline{1,3}{2,2}_{$\sigma_{2}$}
\ncline{1,3}{1,4}^{$\sigma_{3}$}
\ncline{2,2}{3,3}_{$\sigma_{2}$}
\normalsize
\end{psmatrix}
\end{pspicture}
\caption{$\sigma_{1} < \sigma_{2} < \sigma_{3}$.}
\label{figAllEdgesOfWeightw2AreIncidentToSameEnd-VertexOfe1}
\end{figure}
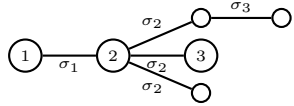
\begin{proof}
$|E_{1}| = 1$ by Proposition~\ref{propAtMost3DifferentEdgeWeights}.
By Lemma~\ref{lemw1<w2<w3e_{1}e_{2}IncidentTove3Adjacenttoe1ore2ButNotIncidentTov}
there exist $e_{1}, e_{2}, e_{3}$ in $E$
with $w_{i} = \sigma_{i}$ for $i \in \lbrace 1,2,3 \rbrace$
and such that one of the three situations
represented in Figure~\ref{figleme_{1}e_{2}IncidentTove3Adjacenttoe1ore2ButNotIncidentTov}
holds.
Let $e$ be an edge in $E_{2} \setminus \lbrace e_{2} \rbrace$
non-incident to $2$.
Let us first assume $e_{3} = \lbrace 3,4 \rbrace $
(Case a in Figure~\ref{figleme_{1}e_{2}IncidentTove3Adjacenttoe1ore2ButNotIncidentTov}).
Then,
$e$ cannot be incident to $3$
(resp. $4$)
otherwise $\lbrace e_{2}, e_{3},e \rbrace$
contradicts the Star
(resp. Path)
condition.
Finally,
$e$ cannot be incident to $1$
otherwise Lemma~\ref{propMax(w1,w2)<=w3=w4=...=wm}
applied to $\lbrace e, e_{1}, e_{2}, e_{3} \rbrace$
implies $\sigma_{2} = \sigma_{3}$, a contradiction.
Let us now assume $e_{3} = \lbrace 3, 1 \rbrace$
(Case b in Figure~\ref{figleme_{1}e_{2}IncidentTove3Adjacenttoe1ore2ButNotIncidentTov}).
Then,
$e$ cannot be incident to $1$ (resp. $3$)
otherwise $\lbrace e, e_{1}, e_{3} \rbrace$ (resp. $\lbrace e, e_{2}, e_{3} \rbrace$) contradicts the Star condition.
Finally,
if $e_{3} = \lbrace 4, 1 \rbrace$
(Case c in Figure~\ref{figleme_{1}e_{2}IncidentTove3Adjacenttoe1ore2ButNotIncidentTov}),
then we can establish as before with $e_{3} = \lbrace 3,4 \rbrace $
that $e$ cannot be incident to $1,3$, and~$4$.
As $G$ is connected,
there exists a shortest path $\gamma$ linking $e$ to $1, 2$ or~$3$.
The Path condition
applied to $ \lbrace e \rbrace \cup \gamma \cup \lbrace e_{1} \rbrace$
or $\lbrace e \rbrace \cup \gamma \cup \lbrace e_{2} \rbrace$
implies that all edges in $\gamma$ have weight $\sigma_{1}$ or $\sigma_{2}$.
As $e_{1}$ is the unique edge with weight $\sigma_{1}$,
we get $w(e) = \sigma_{2}$ for any edge $e$ in $\gamma$.
Let $e'$ be the edge of $\gamma$ incident to $1$, $2$, or~$3$.
As $w(e') = \sigma_{2}$,
the first part of the proof implies that 
$e'$ is necessarily incident to $2$.
If $e_{3} = \lbrace 3,4 \rbrace$
(resp.  $e_{3} = \lbrace 3,1 \rbrace$ or $\lbrace 4,1 \rbrace$),
we consider $\gamma' = \lbrace e \rbrace \cup \gamma \cup \lbrace e_{2}, e_{3} \rbrace$
(resp. $\gamma' = \lbrace e \rbrace \cup \gamma \cup \lbrace e_{1}, e_{3} \rbrace$)
as represented in Figure~\ref{figProofAllEdgesOfWeightw2AreIncidentToSameEnd-VertexOfe4}
(resp. Figure~\ref{figProofAllEdgesOfWeightw2AreIncidentToSameEnd-VertexOfe16})
with $\gamma$ reduced to $e'$.
Then,
Lemma~\ref{propMax(w1,w2)<=w3=w4=...=wm}
applied to $\gamma'$
implies $\sigma_{2} = \sigma_{3}$
(resp. $\sigma_{1} = \sigma_{3}$),
a contradiction.
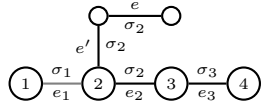
\begin{figure}[!h]
\centering
\begin{pspicture}(0,-.1)(0,1.3)
\tiny
\begin{psmatrix}[mnode=circle,colsep=0.5,rowsep=0.5]
& {}		&		{}\\
{$1$}	& {$2$}	& {$3$} &  {$4$}
\psset{arrows=-, shortput=nab,labelsep={0.06}}
\tiny
\ncline{1,2}{2,2}^{$\sigma_{2}$}_{$e'$}
\ncline{1,2}{1,3}^{$e$}_{$\sigma_{2}$}
\ncline[linecolor=gray]{2,1}{2,2}_{$e_{1}$}^{$\sigma_{1}$}
\ncline{2,2}{2,3}_{$e_{2}$}^{$\sigma_{2}$}
\ncline{2,3}{2,4}_{$e_{3}$}^{$\sigma_{3}$}
\normalsize
\end{psmatrix}
\end{pspicture}
\caption{$\gamma' = \lbrace e \rbrace \cup \gamma \cup \lbrace e_{2}, e_{3} \rbrace$.}
\label{figProofAllEdgesOfWeightw2AreIncidentToSameEnd-VertexOfe4}
\end{figure}
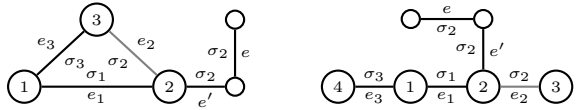
\begin{figure}[!h]
\centering
\begin{minipage}{4cm}
\begin{pspicture}(0,-.1)(0,.9)
\tiny
\begin{psmatrix}[mnode=circle,colsep=0.5,rowsep=0.4]
			& {$3$}	& 	  	& {}\\
{$1$}	&  			& {$2$}	& {}
\psset{arrows=-, shortput=nab,labelsep={0.05}}
\tiny
\ncline{2,3}{2,4}^{$\sigma_{2}$}_{$e'$}
\ncline{2,4}{1,4}_{$e$}^{$\sigma_{2}$}
\ncline{2,1}{2,3}_{$e_{1}$}^{$\sigma_{1}$}
\ncline[linecolor=gray]{1,2}{2,3}^{$e_{2}$}_{$\sigma_{2}$}
\ncline{1,2}{2,1}_{$e_{3}$}^{$\sigma_{3}$}
\normalsize
\end{psmatrix}
\end{pspicture}
\end{minipage}
\begin{minipage}{3.5cm}
\begin{pspicture}(0,-.1)(0,.9)
\tiny
\begin{psmatrix}[mnode=circle,colsep=0.5,rowsep=0.5]
& {} & {}\\
{$4$}	& {$1$}	& {$2$} & {$3$}
\psset{arrows=-, shortput=nab,labelsep={0.05}}
\tiny
\ncline{1,2}{1,3}_{$\sigma_{2}$}^{$e$}
\ncline{1,3}{2,3}^{$e'$}_{$\sigma_{2}$}
\ncline{2,2}{2,3}_{$e_{1}$}^{$\sigma_{1}$}
\ncline{2,1}{2,2}_{$e_{3}$}^{$\sigma_{3}$}
\ncline[linecolor=gray]{2,3}{2,4}_{$e_{2}$}^{$\sigma_{2}$}
\normalsize
\end{psmatrix}
\end{pspicture}
\end{minipage}
\caption{$\gamma' = \lbrace e \rbrace \cup \gamma \cup \lbrace e_{1}, e_{3} \rbrace$.}
\label{figProofAllEdgesOfWeightw2AreIncidentToSameEnd-VertexOfe16}
\end{figure}

Finally,
if an edge in $E_{3}$ is incident to $v$,
then it contradicts the Star condition.
Then,
Lemma~\ref{LemP_MinMax(w1,w2)<w(e)e'InELinkingeTo2} implies the result.
\end{proof}

A \textit{chordless cycle} in $G$
is an induced cycle,
\emph{i.e.},
a cycle corresponding to an induced subgraph of $G$.

\begin{remark}
By Proposition~\ref{lemAllEdgesOfWeightw2AreIncidentToSameEndVertexofe1}
any cycle containing $e_{1}$
has length at most $4$,
and there are only two possible chordless cycles containing $e_1$
represented in Figure~\ref{figSupplementaryNecessaryConditionsForCyclesm=3or42}.
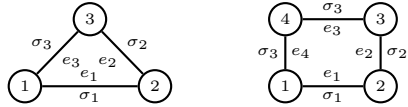
\begin{figure}[!h]
\centering
\subfloat{
\begin{pspicture}(-.5,-.1)(.5,1)
\tiny
\begin{psmatrix}[mnode=circle,colsep=0.4,rowsep=0.4]
&	{$3$}\\
{$1$}	& 	&	{$2$}
\psset{arrows=-, shortput=nab,labelsep={0.05}}
\tiny
\ncline{2,1}{2,3}^{$e_{1}$}_{$\sigma_{1}$}
\ncline{2,3}{1,2}^{$e_{2}$}_{$\sigma_{2}$}
\ncline{2,1}{1,2}_{$e_{3}$}^{$\sigma_{3}$}
\end{psmatrix}
\end{pspicture}
}
\subfloat{
\begin{pspicture}(-.5,-.1)(.5,1)
\tiny
\begin{psmatrix}[mnode=circle,colsep=0.4,rowsep=0.4]
{$4$} &	  & {$3$}\\
{$1$}	& 	&	{$2$}
\psset{arrows=-, shortput=nab,labelsep={0.05}}
\tiny
\ncline{2,1}{2,3}^{$e_{1}$}_{$\sigma_{1}$}
\ncline{2,3}{1,3}^{$e_{2}$}_{$\sigma_{2}$}
\ncline{1,1}{1,3}_{$e_{3}$}^{$\sigma_{3}$}
\ncline{1,1}{2,1}_{$\sigma_{3}$}^{$e_{4}$}
\normalsize
\end{psmatrix}
\end{pspicture}
}
\caption{$\tilde{C}_{3}$ and $\tilde{C}_{4}$ containing the edge of weight $\sigma_1$.}
\label{figSupplementaryNecessaryConditionsForCyclesm=3or42}
\end{figure}
Moreover,
by the Adjacent cycles condition
such a chordless cycle
is necessarily unique
(the existence of two such chordless cycles would imply $\sigma_1 = \sigma_2$).
\end{remark}

We end this section
with supplementary necessary conditions
corresponding to refinements of the Pan condition.
A cycle $C$ is
\emph{constant} if all edges in $E(C)$ have the same weight,
and non-constant otherwise.
\vspace{\baselineskip}

\begin{mdframed}
\textbf{Non-Constant Cycle Refined Pan Condition.}
\it
For all connected subgraphs corresponding to the union of a non-constant simple cycle
$C_m = \lbrace 1, e_{1}, 2, e_{2}, \ldots, m, e_{m}, 1 \rbrace$
with $m \geq 3$
and an elementary path $P$
containing an edge $e$ with $w(e) < \min_{1 \leq j \leq m} w_{j}$,
$e$ is incident to $2$ but  not a chord of $C_m$,
$C_m$ is a complete cycle,
and the edge-weights satisfy
\begin{equation}
\label{eqNonConstantCycleRefinedPanCondition}
w(e) < w_{1} = w_{2} < w_{3} = \cdots = w_{m}
= \hat{M} = \max_{e \in \hat{E}(C_m)} w(e).
\end{equation}
\end{mdframed}

\vspace{\baselineskip}
If $P$ is reduced to $e$,
then a pan satisfying the previous condition is represented
in Figure~\ref{figStrongPanConditionForNonConstantCycle}.
\begin{figure}[!h]
\centering
\begin{pspicture}(0,-.1)(0,1.8)
\tiny
\begin{psmatrix}[mnode=circle,colsep=0.2,rowsep=0.1]
	& { }	&  	&  {3} \\
{ } 	& 	&	& 	&  {$2$}  & &  {$\scriptstyle{m+1}$}\\
	& {$m$} 	& 	& {$1$}
\psset{arrows=-, shortput=nab,labelsep={0.05}}
\tiny
\ncline{3,2}{3,4}_{$\hat{M}$}
\ncline{2,1}{3,2}_{$\hat{M}$}
\ncline{2,1}{1,2}^{$\hat{M}$}
\ncline{3,4}{2,5}_{$w_{1}$}
\ncline{2,5}{1,4}_{$w_{2}$}
\ncline{1,2}{1,4}^{$\hat{M}$}
\ncline{1,2}{3,2}
\ncline{1,2}{3,4}
\ncline{1,2}{2,5}
\ncline{1,4}{2,1}
\ncline{1,4}{3,2}
\ncline{1,4}{3,4}
\ncline{2,1}{2,5}
\ncline{2,1}{3,4}
\ncline{3,2}{2,5}
\ncline{2,5}{2,7}_{$w(e)$}
\normalsize
\end{psmatrix}
\end{pspicture}
\caption{$w(e) < w_{1} = w_{2} < w_{3} = \cdots = w_{m} = \hat{M}$.}
\label{figStrongPanConditionForNonConstantCycle}
\end{figure}

\begin{proposition}
\label{lemStrongPanCondForNonConstantCycle}
If for all $\emptyset \not= S \subseteq N$
the $\mathcal{P}_{\min}$-restricted game $(N,\overline{u_{S}})$ is convex,
then the Non-Constant Cycle Refined Pan Condition is satisfied.
\end{proposition}
\begin{proof}
By the Cycle condition any chord of $C_m$ has weight $w_{2}$ or $\hat{M}$.
Therefore,
$e$ cannot be a chord of $C_m$.
Let us assume $e$ non-incident to $C_m$.
Let $P'$ be the shortest path induced by $P$ linking $e$ to $C_m$.
As $C_m$ is a non-constant cycle,
the Pan condition applied to $C_{m}$ and $P'$ implies (\ref{eqNonConstantCycleRefinedPanCondition})
and $V(C_m) \cap V(P') = \lbrace 2 \rbrace$.
If $e$ is not incident to $2$,
then Lemma~\ref{propMax(w1,w2)<=w3=w4=...=wm}
implies that $C_m$ is a constant cycle, a contradiction.
Let us set $e = \lbrace 2, m+1 \rbrace$.

Let us first prove $e_{j}' = \lbrace 2, j \rbrace \in \hat E(C_m)$ for all $j$, $4 \leq j \leq m$.
By Lemma~\ref{LemP_MinMax(w1,w2)<w(e)e'InELinkingeTo2},
it is sufficient to prove the existence of such a chord for $m=4$.
Indeed,
if $e_{j}' \notin \hat E(C_m)$ for a given index $j$,
then Lemma~\ref{LemP_MinMax(w1,w2)<w(e)e'InELinkingeTo2}
applied to $e_{j-1}$ (resp. $e_{j}$)
and to the pair of adjacent edges $\lbrace e_{1},e_{2} \rbrace$
implies that $e_{j-1}'$ (resp. $e_{j+1}'$) exists in $\hat E(C_m)$.
Then,
$\lbrace 2, e_{j-1}', j-1, e_{j-1}, j, e_{j}, j+1, e_{j+1}', 2 \rbrace$
defines a cycle of length $4$
as represented in Figure~\ref{figProofStrongPanConditionForNonConstantCyclem=4}.
\begin{figure}[!h]
\centering
\begin{pspicture}(0,-.1)(0,2.8)
\tiny
\begin{psmatrix}[mnode=circle,colsep=0.3,rowsep=0.2]
	& {$\scriptstyle{j-1}$}	&  	&  {3} \\
{$j$} 	& 	&	& 	&  {$2$} &  & {$\scriptstyle{m+1}$}\\
	& {$\scriptstyle{j+1}$} 	& 	& {$1$}
\psset{arrows=-, shortput=nab,labelsep={0.05}}
\tiny
\ncline[linecolor=gray]{3,2}{3,4}_{$e_{m}$}
\ncline{2,1}{3,2}_{$e_{j}$}
\ncline{2,1}{1,2}^{$e_{j-1}$}
\ncline[linecolor=gray]{3,4}{2,5}_{$e_{1}$}
\ncline[linecolor=gray]{2,5}{1,4}_{$e_{2}$}
\ncline[linecolor=gray]{1,2}{1,4}^{$e_{3}$}
\ncline[linestyle=dotted]{2,1}{2,5}^{$e_{j}'$}
\psset{labelsep={-0.02}}
\ncline{3,2}{2,5}^{$e_{j+1}'$}
\psset{labelsep={-0.06}}
\ncline{1,2}{2,5}^{$e_{j-1}'$}
\psset{labelsep={0.02}}
\ncline{2,5}{2,7}_{$w(e)$}
\end{psmatrix}
\end{pspicture}
\caption{$w(e) < w_{1} = w_{2} < w_{3} = \cdots = w_{m} =
\hat{M}
$}.
\label{figProofStrongPanConditionForNonConstantCyclem=4}
\end{figure}
By contradiction let us assume $e_{4}' \notin \hat E(C_{4})$.
By the Pan condition,
$\lbrace 1, 3\rbrace $ is a maximum weight chord of $C_{4}$.
Let us consider $i = 3$,
$A_{1} = \lbrace 2, 5 \rbrace$, $A_{2}= \lbrace 4 \rbrace$,
$A = A_{1} \cup A_{2}$, and $B = A \cup \lbrace 1 \rbrace$ 
as represented in Figure~\ref{figProofStrongPanConditionForNonConstantCycle}.
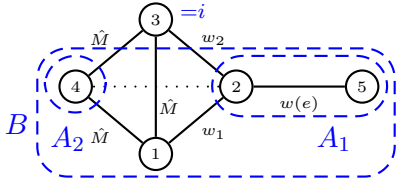
\begin{figure}[!h]
\centering
\begin{pspicture}(0,-.2)(0,2)
\tiny
\begin{psmatrix}[mnode=circle,colsep=0.6,rowsep=0.4]
	& {$3$} \\
{$4$} 	& 	&	 {$2$}  & &  {$5$}\\
	& {$1$}
\psset{arrows=-, shortput=nab,labelsep={0.1}}
\tiny
\ncline{1,2}{2,3}^{$w_{2}$}
\ncline{3,2}{2,3}_{$w_{1}$}
\ncline{2,1}{3,2}_{$\hat{M}$}
\ncline{2,1}{1,2}^{$\hat{M}$}
\ncline{1,2}{3,2}
\naput[npos=0.7,labelsep={0.01}]{$\hat{M}$}
\ncline[linestyle=dotted]{2,1}{2,3}
\ncline{2,3}{2,5}_{$w(e)$}
\end{psmatrix}
\normalsize
\uput[0](-1,.2){\textcolor{blue}{$A_{1}$}}
\pspolygon[framearc=1,linestyle=dashed,linecolor=blue,linearc=.4](-2.2,.5)(-2.2,1.3)(.1,1.3)(.1,.5)
\uput[0](-4.5,.2){\textcolor{blue}{$A_{2}$}}
\pspolygon[framearc=1,linestyle=dashed,linecolor=blue,linearc=.4](-4.4,.55)(-4.4,1.3)(-3.6,1.3)(-3.6,.55)
\uput[0](-2.8,1.9){\textcolor{blue}{$\scriptstyle{=i}$}}
\uput[0](-5.1,.4){\textcolor{blue}{$B$}}
\pspolygon[framearc=1,linestyle=dashed,linecolor=blue,linearc=.4](-4.5,-.3)(-4.5,1.4)(.2,1.4)(.2,-.3)
\end{pspicture}
\caption{$w(e) < w_{1} = w_{2} < w_{3} = w_{4} =
\hat{M}
$.}
\label{figProofStrongPanConditionForNonConstantCycle}
\end{figure}
If $\lbrace 4, 5 \rbrace \in E$,
then the Cycle condition applied to
$\lbrace 1,e_1,2,e,5,\lbrace 5,4\rbrace, 4,e_4,1 \rbrace$
implies $w(\lbrace 4,5\rbrace)=\hat{M}> w(e)$.
Hence,
we have either $\mathcal{P}_{\min}(A) = \lbrace \lbrace 2 \rbrace, \lbrace 4 , 5 \rbrace \rbrace$
or $\mathcal{P}_{\min}(A) = \lbrace \lbrace 2 \rbrace, \lbrace 4 \rbrace, \lbrace 5 \rbrace \rbrace$.
Moreover,
$\mathcal{P}_{\min}(A \cup \lbrace i \rbrace) = \lbrace \lbrace 2, 3, 4 \rbrace, \lbrace 5 \rbrace \rbrace$
or $\lbrace A \cup \lbrace i \rbrace \rbrace$,
and $\mathcal{P}_{\min}(B) = \lbrace \lbrace 1, 2, 4 \rbrace, \lbrace 5 \rbrace \rbrace$
or $\lbrace B \rbrace$.
If $\mathcal{P}_{\min}(A \cup \lbrace i \rbrace) = \lbrace \lbrace 2, 3, 4 \rbrace, \lbrace 5 \rbrace \rbrace$,
then taking $A' = \lbrace 2,3,4 \rbrace$
we get $\mathcal{P}_{\min}(A)_{|A'} = \lbrace \lbrace 2 \rbrace, \lbrace 4 \rbrace \rbrace
\not= \lbrace 2, 4 \rbrace = \mathcal{P}_{\min}(B)_{|A'}$
contradicting Theorem~\ref{thLGNuPNFNFSNuSSNNuSFABFABF}.
Otherwise,
taking $A' = A \cup \lbrace i \rbrace$,
we get
$\mathcal{P}_{\min}(A)_{|A'} =
\lbrace \lbrace 2 \rbrace, \lbrace 4 \rbrace,\lbrace 5 \rbrace \rbrace$
or
$\lbrace \lbrace 2 \rbrace, \lbrace 4, 5 \rbrace \rbrace$
and
$\mathcal{P}_{\min}(B)_{|A'} = \lbrace \lbrace 2, 4 \rbrace, \lbrace 5 \rbrace \rbrace$
or
$\lbrace \lbrace 2, 4, 5 \rbrace \rbrace$.
Therefore,
we always have
$\mathcal{P}_{\min}(A)_{|A'} \not= \mathcal{P}_{\min}(B)_{|A'}$
and it contradicts Theorem~\ref{thLGNuPNFNFSNuSSNNuSFABFABF}.

Let us now prove $\lbrace j, k \rbrace \in \hat E(C_{m})$ for all pairs of vertices $j$, $k$,
with $3 \leq j \leq m-1$ and $k = 1$ or $j+2 \leq k \leq m$.
We have $\lbrace 2, j \rbrace$ and $\lbrace 2, k \rbrace$ in $\hat E(C_{m})$.
\begin{figure}[!h]
\centering
\begin{pspicture}(0,-.2)(0,2.3)
\tiny
\begin{psmatrix}[mnode=circle,colsep=0.4,rowsep=0.4]
	& {$j$}	&  	&  {3} \\
{ } 	& 	&	& 	&  {$2$}  & &  { }\\
	& {$k$} 	& 	& {$1$}
\psset{arrows=-, shortput=nab,labelsep={0.1}}
\tiny
\ncline[linecolor=gray]{3,2}{3,4}_{$\hat{M}$}
\ncline{2,1}{3,2}_{$\hat{M}$}
\ncline{2,1}{1,2}^{$\hat{M}$}
\ncline[linecolor=gray]{3,4}{2,5}_{$w_{1}$}
\ncline[linecolor=gray]{2,5}{1,4}_{$w_{2}$}
\ncline[linecolor=gray]{1,2}{1,4}^{$\hat{M}$}
\ncline[linestyle=dotted]{1,2}{3,2}
\psset{labelsep={-0.02}}
\ncline{1,2}{2,5}_{$e_{j}'$}
\ncline{3,2}{2,5}^{$e_{k}'$}
\psset{labelsep={0.1}}
\ncline[linecolor=gray]{2,1}{2,5}
\ncline{2,5}{2,7}_{$w(e)$}
\normalsize
\end{psmatrix}
\end{pspicture}
\caption{$\tilde{C}_{m} = \lbrace 2, e_{j}', j, e_{j}, j+1, \ldots,
k, e_{k}', 2 \rbrace$ and $w_{j}' = w_{k}' = w_{1} = w_{2}$.}
\label{figProofStrongPanConditionForNonConstantCycle2}
\end{figure}
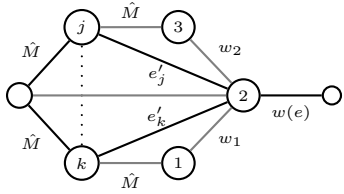
Then,
the Pan condition applied to
$\tilde{C}_{m} = \lbrace 2, e_{j}', j, e_{j}, j+1, \ldots, k, e_{k}', 2 \rbrace$ and $e$
as represented in Figure~\ref{figProofStrongPanConditionForNonConstantCycle2} implies that
$\lbrace j, k \rbrace$ is a maximum weight chord of $\tilde{C}_{m}$.
\end{proof}


We finally establish
necessary conditions on constant cycles and pans associated with constant cycles.

\begin{proposition}
\label{lemStrongPanConditionForCyclesWithConstantWeights}
Let us assume that
for all $\emptyset \not= S \subseteq N$,
the
$\mathcal{P}_{\min}$-restricted game $(N,\overline{u_{S}})$ is convex
and that the edge-weights have at most three different values
$\sigma_{1} < \sigma_{2} \leq \sigma_{3}$.
Then
\begin{enumerate}
\item
\label{itemLemmeE12w2=w3w2}
If $|E_{1}| \geq 2$ (then $\sigma_{2} = \sigma_{3}$),
then every cycle with constant weight $\sigma_{2}$ is complete.
\item
\label{itemLemmeE1=1E&=e1w2e1}
If $|E_{1}| = 1$ with $E_{1} = \lbrace e_{1} \rbrace$
and if there exists a cycle $C$ with constant weight $\sigma_{2}$,
then there are only two different edge-weights ($\sigma_{2} = \sigma_{3}$).
Moreover,
if $C$ is not incident to $e_{1}$
and not linked to $e_{1}$ by an edge,
then $C$ is complete.
\item
\label{itemLemmeE1=1E1=e1=12Cw2w322jECjVC}
If $|E_{1}| = 1$ with $E_{1} = \lbrace e_{1} \rbrace$ and $e_{1} = \lbrace 1, 2 \rbrace$,
then for every cycle $C$ with constant weight $\sigma_{2}$ or $\sigma_{3}$ 
and incident to $1$ (resp. $2$),
$\lbrace 1, j \rbrace \in E$ for all $j \in V(C)\setminus\lbrace 1\rbrace$
(resp. $\lbrace 2, j \rbrace \in E$ for all $j \in V(C)\setminus\lbrace 2\rbrace$).
\item
\label{itemLemmeE1=1E1=e1e1=12Cw2w32e122'w2w3}
If $|E_{1}| = 1$ with $E_{1} = \lbrace e_{1} \rbrace$ and $e_{1} = \lbrace 1, 2 \rbrace$,
then for every cycle $C$ with constant weight $\sigma_{2}$
or $\sigma_{3}$
and not adjacent to $e_{1}$
but linked to $e_{1}$ by an edge $e = \lbrace 2, k \rbrace$
(of weight $\sigma_{2}$)
with $k \in V(C)$,
one of the following conditions is satisfied:
\begin{enumerate}
\item
\label{itemItemLemme1jEjVC}
$\lbrace 1, j \rbrace \in E$ for all $j \in V(C)$.
\item
\label{itemItemLemme2jEjVC} 
$\lbrace 2, j \rbrace \in E$ for all $j \in V(C)$.
\item
\label{itemItemLemme2jEjVCkC}
There is no edge $\lbrace 2, j \rbrace$ in $E$
with $j \in V(C) \setminus \lbrace k \rbrace$
and $C$ is complete.
\end{enumerate}
\item
\label{itemLemmew_1<w_2<w_3e1=1,2w_1w_22C_me_11V(C_m)m=3C_3e_1}
Let us assume that the edge-weights have three different values $\sigma_{1} < \sigma_{2} < \sigma_{3}$.
Let $e_{1} = \lbrace 1, 2 \rbrace$ be the unique edge in $E_{1}$
and let us assume all edges in $E_{2}$ incident to~$2$.
Then,
every cycle $C_{m}$
with $e_{1} \notin E(C_{m})$ is complete
and $e_{1} \notin \hat{E}(C_{m})$.
\end{enumerate}
\end{proposition}

Situations corresponding to Claims~\ref{itemLemmeE1=1E1=e1=12Cw2w322jECjVC}
and \ref{itemLemmeE1=1E1=e1e1=12Cw2w32e122'w2w3}
in Proposition~\ref{lemStrongPanConditionForCyclesWithConstantWeights}
are represented in Figure~\ref{figlemStrongPanConditionForCyclesWithConstantWeights}.

\begin{figure}[!h]
\centering
\subfloat[Claim~\ref{itemLemmeE1=1E1=e1=12Cw2w322jECjVC}]{
\begin{pspicture}(0,-.7)(0,2)
\tiny
\begin{psmatrix}[mnode=circle,colsep=0.4,rowsep=0.25]
{}	&	  &  {} \\
	& 	&	& 	  {$2$} & & {$1$} \\
{}	& 	&  {}
\psset{arrows=-, shortput=nab,labelsep={0.1}}
\tiny
\ncline{3,1}{3,3}_{$\sigma_{2}$}
\ncline{3,3}{2,4}_{$\sigma_{2}$}
\ncline{2,4}{1,3}_{$\sigma_{2}$}
\ncline{1,3}{1,1}_{$\sigma_{2}$}
\ncline{1,1}{3,1}_{$\sigma_{2}$}
\ncline{3,1}{2,4}
\ncline{1,1}{2,4}
\ncline{2,4}{2,6}_{$e_{1}$}^{$\sigma_{1}$}
\normalsize
\end{psmatrix}
\end{pspicture}
}
\hspace{1cm}
\subfloat[Claim~\ref{itemItemLemme1jEjVC}]{
\begin{pspicture}(0,-.7)(0,2)
\tiny
\begin{psmatrix}[mnode=circle,colsep=0.4,rowsep=0.2]
{}	&	  &  {} \\
	& 	&	& 	  {$2'$} & & {$2$} & & {$1$} \\
{}	& 	&  {}
\psset{arrows=-, shortput=nab,labelsep={0.05}}
\tiny
\ncline{2,6}{2,8}_{$e_{1}$}^{$\sigma_{1}$}
\ncline{3,1}{3,3}_{$\sigma_{2}$}
\ncline{3,3}{2,4}_{$\sigma_{2}$}
\ncline{2,4}{1,3}_{$\sigma_{2}$}
\ncline{1,3}{1,1}_{$\sigma_{2}$}
\ncline{1,1}{3,1}_{$\sigma_{2}$}
\ncline{2,4}{2,6}_{$\sigma_{2}$}
\ncarc[arcangle=40]{1,1}{2,8}^{$\sigma_{2}$}
\ncarc[arcangle=25]{1,3}{2,8}^{$\sigma_{2}$}
\ncarc[arcangle=-25]{3,3}{2,8}_{$\sigma_{2}$}
\ncarc[arcangle=-40]{3,1}{2,8}_{$\sigma_{2}$}
\ncarc[arcangle=25]{2,4}{2,8}^{$\sigma_{2}$}
\normalsize
\end{psmatrix}
\end{pspicture}
}\\
\subfloat[Claim~\ref{itemItemLemme2jEjVC}]{
\begin{pspicture}(0,-.7)(0,2.1)
\tiny
\begin{psmatrix}[mnode=circle,colsep=0.4,rowsep=0.2]
{}	&	  &  {} \\
	& 	&	& 	  {$2'$} & & {$2$} & & {$1$} \\
{}	& 	&  {}
\psset{arrows=-, shortput=nab,labelsep={0.05}}
\tiny
\ncline{2,6}{2,8}_{$e_{1}$}^{$\sigma_{1}$}
\ncline{3,1}{3,3}_{$\sigma_{2}$}
\ncline{3,3}{2,4}_{$\sigma_{2}$}
\ncline{2,4}{1,3}_{$\sigma_{2}$}
\ncline{1,3}{1,1}_{$\sigma_{2}$}
\ncline{1,1}{3,1}_{$\sigma_{2}$}
\ncline{2,4}{2,6}_{$\sigma_{2}$}
\ncarc[arcangle=50]{1,1}{2,6}^{$\sigma_{2}$}
\ncarc[arcangle=20]{1,3}{2,6}^{$\sigma_{2}$}
\ncarc[arcangle=-20]{3,3}{2,6}_{$\sigma_{2}$}
\ncarc[arcangle=-50]{3,1}{2,6}_{$\sigma_{2}$}
\normalsize
\end{psmatrix}
\end{pspicture}
}
\hspace{1cm}
\subfloat[Claim~\ref{itemItemLemme2jEjVCkC}]{
\begin{pspicture}(0,-.7)(0,1.8)
\tiny
\begin{psmatrix}[mnode=circle,colsep=0.4,rowsep=0.2]
{}	&	  &  {} \\
	& 	&	& 	  {$2'$} & & {$2$} & & {$1$} \\
{}	& 	&  {}
\psset{arrows=-, shortput=nab,labelsep={0.1}}
\tiny
\ncline{3,1}{3,3}_{$\sigma_{2}$}
\ncline{3,3}{2,4}_{$\sigma_{2}$}
\ncline{2,4}{1,3}_{$\sigma_{2}$}
\ncline{1,3}{1,1}_{$\sigma_{2}$}
\ncline{1,1}{3,1}_{$\sigma_{2}$}
\ncarc[arcangle=35,linestyle=dotted]{2,4}{2,8}^{$\sigma_{2}$}
\ncline{1,1}{3,3}
\ncline{1,3}{3,3}
\ncline{1,3}{3,1}
\ncline{3,1}{2,4}
\ncline{1,1}{2,4}
\ncline{2,4}{2,6}^{$\sigma_{2}$}
\ncline{2,6}{2,8}_{$e_{1}$}^{$\sigma_{1}$}
\normalsize
\end{psmatrix}
\end{pspicture}
}
\caption{Situations of Claims~\ref{itemLemmeE1=1E1=e1=12Cw2w322jECjVC}
and~\ref{itemLemmeE1=1E1=e1e1=12Cw2w32e122'w2w3}
in Proposition~\ref{lemStrongPanConditionForCyclesWithConstantWeights}.}
\label{figlemStrongPanConditionForCyclesWithConstantWeights}
\end{figure}
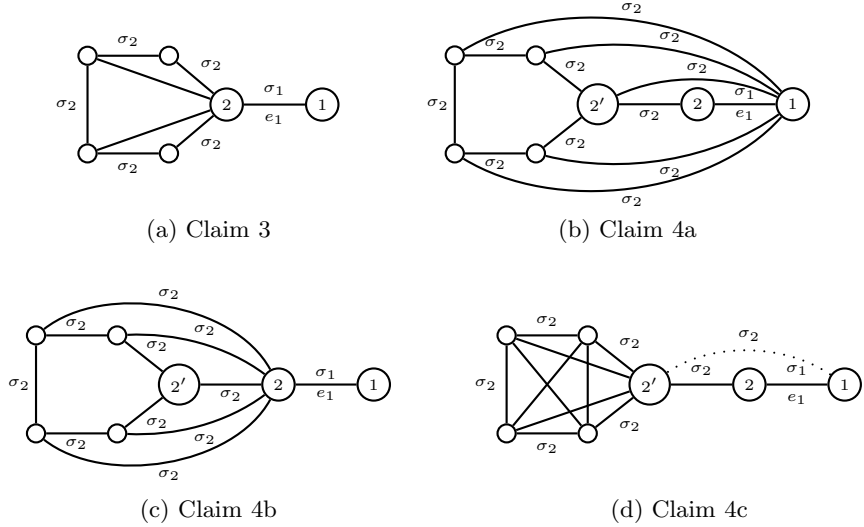

\begin{remark}
If the edge-weights have three different values $\sigma_{1} < \sigma_{2} < \sigma_{3}$,
then it follows from Claim~\ref{itemLemmew_1<w_2<w_3e1=1,2w_1w_22C_me_11V(C_m)m=3C_3e_1}
in Proposition~\ref{lemStrongPanConditionForCyclesWithConstantWeights}
that the cycle $C$ considered in Claim~\ref{itemLemmeE1=1E1=e1=12Cw2w322jECjVC}
is a triangle
(of constant weight $\sigma_{3}$).
As a triangle has no chord
Claim~\ref{itemLemmeE1=1E1=e1=12Cw2w322jECjVC} adds nothing
in the particular case of three different edge-weights.
But, $a\, priori$,
we have to keep this case
in Claim~\ref{itemLemmeE1=1E1=e1=12Cw2w322jECjVC}
to be able to prove Claim~\ref{itemLemmew_1<w_2<w_3e1=1,2w_1w_22C_me_11V(C_m)m=3C_3e_1}.
\end{remark}

\begin{proof}
\ref{itemLemmeE12w2=w3w2}.
By Proposition~\ref{propAtMost3DifferentEdgeWeights}
the edge-weights have at most two different values $\sigma_{1} < \sigma_{2}$.
By Proposition~\ref{propEitherOnlyOneEdgeOfWeightw1OrAllEdgesOFWeightw1AreIncidentToSameVertex},
all edges in $E_{1}$ are incident to the same vertex $j$
and no edge in $E_{2}$ is incident to $j$.
Let us consider a cycle $C$
with constant weight $\sigma_{2}$.
By Lemma~\ref{LemP_MinMax(w1,w2)<w(e)e'InELinkingeTo2},
every edge in $E_{2}$ is linked to $j$ by an edge in $E_{1}$
as represented in Figure~\ref{figProofLemStrongPanConditionForCyclesWithConstantWeights}.
\begin{figure}[!h]
\centering
\begin{pspicture}(0,0)(0,1.9)
\tiny
\begin{psmatrix}[mnode=circle,colsep=0.4,rowsep=0.2]
{}	&	  &  {} \\
	& 	&	& 	  {} & & {$j$} \\
{}	& 	&  {}
\psset{arrows=-, shortput=nab,labelsep={0.05}}
\tiny
\ncline{3,1}{3,3}_{$\sigma_{2}$}
\ncline{3,3}{2,4}^{$\sigma_{2}$}
\ncline{2,4}{1,3}^{$\sigma_{2}$}
\ncline{1,3}{1,1}^{$\sigma_{2}$}
\ncline{1,1}{3,1}_{$\sigma_{2}$}
\ncarc[arcangle=20]{1,3}{2,6}_{$\sigma_{1}$}
\ncarc[arcangle=40]{1,1}{2,6}^{$\sigma_{1}$}
\ncarc[arcangle=-20]{3,3}{2,6}_{$\sigma_{1}$}
\ncline{2,4}{2,6}^{$\sigma_{1}$}
\normalsize
\end{psmatrix}
\end{pspicture}
\caption{Every edge in $E_{2}$ is linked to $j$ by an edge in $E_{1}$.}
\label{figProofLemStrongPanConditionForCyclesWithConstantWeights}
\end{figure}
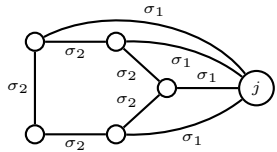
Let us consider a given game $(N,v)$,
the $\mathcal{P}_{\min}$-restricted game $(N,\overline{v})$,
and the Myerson restricted game $(N,v^{M})$.
We have
$\overline{v}(A \cup \lbrace j \rbrace) = v^{M}(A) + v({j}) = v^{M}(A)$
for all $A \subseteq N \setminus \lbrace j \rbrace$.
Hence,
for $i \in V(C)$ and $A \subseteq B \subseteq V(C) \setminus \lbrace i,j \rbrace$
the inequality
\begin{equation}
\label{equSBji-uSBj>=uSAji-uSAj}
\overline{v}(B \cup \lbrace j \rbrace \cup \lbrace i \rbrace)
- \overline{v}(B \cup \lbrace j \rbrace)
\geq
\overline{v}(A \cup \lbrace j \rbrace \cup \lbrace i \rbrace)
- \overline{v}(A \cup \lbrace j \rbrace)
\end{equation}
is equivalent to
\begin{equation}
\label{equSMBi-uSMB>=uSMAi-uSMA}
v^{M}(B \cup \lbrace i \rbrace) - v^{M}(B)
\geq
v^{M}(A \cup \lbrace i \rbrace) - v^{M}(A).
\end{equation}
As $(N,\overline{u_{S}})$ is convex,
(\ref{equSBji-uSBj>=uSAji-uSAj})
and therefore (\ref{equSMBi-uSMB>=uSMAi-uSMA}) are satisfied with $v = u_{S}$.
Then,
$u_{S}^{M}$ is convex if we restrict $G$ to $V(C)$,
and 
$C$ has to be complete
by Theorems~\ref{theoremEquivalenceInheritanceofConvexityForUnanimityGamesP_MAndP_min}
and~\ref{NouwelandandBorm1991}.\\

\noindent
\ref{itemLemmeE1=1E&=e1w2e1}.
If there are three different edge-weights,
then by Proposition~\ref{lemAllEdgesOfWeightw2AreIncidentToSameEndVertexofe1}
there is no cycle with constant weight $\sigma_{2}$.
Let us consider a cycle $C$
with constant weight $\sigma_{2}$ non-incident to $e_{1} = \lbrace 1, 2 \rbrace$
and not linked by an edge to $e_{1}$.
For any game $(N,v)$
and for $A \subseteq N \setminus \lbrace 1, 2 \rbrace$
such that there is no edge linking $A$ to $\lbrace 1, 2 \rbrace$,
we have
\begin{equation}
\label{eqv(A12)=v^M(A)+v(1)+v(2)=v^M(A)}
\overline{v}(A \cup \lbrace 1, 2 \rbrace) =
v^{M}(A) + v(\lbrace 1 \rbrace) + v(\lbrace 2 \rbrace) = 
v^{M}(A).
\end{equation}
Hence,
for $i \in V(C)$ and for $A \subseteq B \subseteq V(C) \setminus \lbrace i \rbrace$,
the subsets $A$, $B$, $A \cup \lbrace i \rbrace$,
and $B \cup \lbrace i \rbrace$
satisfy (\ref{eqv(A12)=v^M(A)+v(1)+v(2)=v^M(A)}).
Then,
the inequality
$  \overline{v}((B\cup \lbrace 1, 2 \rbrace) \cup \lbrace i \rbrace)
- \overline{v}(B\cup \lbrace 1, 2 \rbrace) 
\geq
\overline{v}((A\cup \lbrace 1, 2 \rbrace) \cup \lbrace i \rbrace)
- \overline{v}(A\cup \lbrace 1, 2 \rbrace)$
is equivalent to:
$v^{M}(B \cup \lbrace i \rbrace) - v^{M}(B)
\geq
v^{M}(A \cup \lbrace i \rbrace) - v^{M}(A)$.
Therefore,
taking $v = u_{S}$,
we can conclude as in the previous case.\\

\noindent
\ref{itemLemmeE1=1E1=e1=12Cw2w322jECjVC}.
Let us consider
$E_{1} = \lbrace e_{1}' \rbrace$
with $e_{1}' = \lbrace 1', 2 \rbrace$
and a cycle $C = \lbrace 1, e_{1}, 2, e_{2},\\ \ldots, m, e_{m}, 1 \rbrace$ incident to $2$.
We can assume w.l.o.g. that $C$ has constant weight $\sigma_{2}$.
Note that,
by the Cycle condition,
$e_{1}'$ cannot be a chord of $C$.
By contradiction let us assume $\lbrace 2, 4 \rbrace \notin E$.
Let us consider $i = 3$,
$A_{1} = \lbrace 4 \rbrace$,
$A_{2} = \lbrace 1', 2 \rbrace$,
$A = A_{1} \cup A_{2}$,
and $B = (V(C) \setminus \lbrace i \rbrace) \cup \lbrace 1'\rbrace$
as represented in Figure~\ref{figProofLemStrongPanConditionForCyclesWithConstantWeights3}.
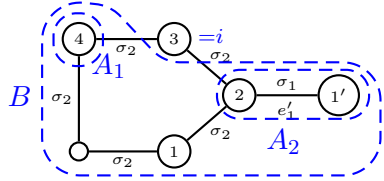
\begin{figure}[!h]
\centering
\begin{pspicture}(0,-.1)(0,1.9)
\tiny
\begin{psmatrix}[mnode=circle,colsep=0.4,rowsep=0.2]
{$4$}	&	  &  {$3$} \\
	& 	&	& 	  {$2$} & & {$1'$}\\
{}	& 	&  {$1$}
\psset{arrows=-, shortput=nab,labelsep={0.05}}
\tiny
\ncline{3,1}{3,3}_{$\sigma_{2}$}
\ncline{3,3}{2,4}_{$\sigma_{2}$}
\ncline{2,4}{1,3}_{$\sigma_{2}$}
\ncline{1,3}{1,1}^{$\sigma_{2}$}
\ncline{1,1}{3,1}_{$\sigma_{2}$}
\ncline{2,4}{2,6}_{$e_{1}'$}^{$\sigma_{1}$}
\normalsize
\pspolygon[framearc=1,linestyle=dashed,linecolor=blue,linearc=.3]
(-1.65,1.2)(-1.65,1.9)(-1,1.9)(-1,1.2)
\uput[0](.2,1.6){\textcolor{blue}{$\scriptstyle{=i}$}}
\uput[0](-1.2,1.2){\textcolor{blue}{$A_{1}$}}
\pspolygon[framearc=1,linestyle=dashed,linecolor=blue,linearc=.3]
(.5,.5)(.5,1.15)(2.5,1.15)(2.5,.5)
\uput[0](1.1,.2){\textcolor{blue}{$A_{2}$}}
\pspolygon[framearc=1,linestyle=dashed,linecolor=blue,linearc=.5]
(-1.8,-.25)(-1.8,2.1)(-.9,2)(-.2,1.25)(2.65,1.25)(2.65,-.25)
\uput[0](-2.3,.8){\textcolor{blue}{$B$}}
\end{psmatrix}
\end{pspicture}
\caption{Cycle $C$ incident to $2$.}
\label{figProofLemStrongPanConditionForCyclesWithConstantWeights3}
\end{figure}
Then,
$\mathcal{P}_{\min}(A) = \lbrace \lbrace 1' \rbrace, \lbrace 2 \rbrace, \lbrace 4 \rbrace  \rbrace$
or $\lbrace \lbrace 1',4 \rbrace, \lbrace 2 \rbrace \rbrace$,
$\mathcal{P}_{\min}(A \cup \lbrace i \rbrace) =
\lbrace \lbrace 1' \rbrace, \lbrace 2, 3, 4 \rbrace \rbrace$
or $\lbrace A \cup \lbrace i \rbrace \rbrace$,
and $\mathcal{P}_{\min}(B) = \lbrace \lbrace 1' \rbrace, V(C) \setminus \lbrace 3 \rbrace \rbrace$
or $\lbrace B \rbrace$.
Then,
for any $A' \in \mathcal{P}_{\min}(A \cup \lbrace i \rbrace)$ containing $2$,
we have
$\lbrace 2 \rbrace \in \mathcal{P}_{\min}(A)_{|A'}$
but $\lbrace 2 \rbrace \notin \mathcal{P}_{\min}(B)_{|A'}$,
contradicting Theorem~\ref{thLGNuPNFNFSNuSSNNuSFABFABF}.
Hence,
we have $e = \lbrace 2, 4 \rbrace \in E$
and the Star condition applied to $\lbrace e, e'_{1}, e_{2} \rbrace$
implies $w(e) = \sigma_{2}$.
Then,
by the same reasoning on the cycle
$\lbrace 1, e_{1}, 2, e, 4, e_{4}, \ldots, m, e_{m}, 1 \rbrace$,
we have $\lbrace 2, 5 \rbrace \in E$.
Iterating we get the result.\\

\noindent
\ref{itemLemmeE1=1E1=e1e1=12Cw2w32e122'w2w3}.
Let us consider $E_{1} = \lbrace e_{1}' \rbrace$
with $e_{1}' = \lbrace 1', 2' \rbrace$
and a cycle
$C = \lbrace 1, e_{1}, 2, e_{2},\\ \ldots, m, e_{m}, 1 \rbrace$
with constant weight $\sigma_{2}$ or $\sigma_{3}$ and not incident to $e_{1}$
but linked to $e_{1}$ by an edge $e = \lbrace 2, 2' \rbrace$ of weight $\sigma_{2}$.

If $e '= \lbrace 2', j \rbrace \in E$ for some $j \in V(C) \setminus \lbrace 2 \rbrace$,
then Claim~\ref{itemLemmeE1=1E1=e1=12Cw2w322jECjVC}
applied to the cycles
$C' = \lbrace 2', e, 2, e_{2}, 3, \ldots, e_{j-1}, j, e', 2' \rbrace$
and
$C'' = \lbrace 2', e', j, e_{j}, j+1 \ldots , e_{m}, 1, e_{1}, 2, e, 2' \rbrace$
incident to $e_{1}'$
as represented in Figure~\ref{figProofLemStrongPanConditionForCyclesWithConstantWeights4}
\begin{figure}[!h]
\centering
\begin{pspicture}(0,0)(0,2.2)
\tiny
\begin{psmatrix}[mnode=circle,colsep=0.4,rowsep=0.15]
{$j$}	&	  &  {$3$} \\
	& 	&	& 	  {$2$} & & {$2'$} & & {$1'$}\\
{$m$}	& 	&  {$1$}
\psset{arrows=-,shortput=nab,labelsep={0.05}}
\tiny
\ncline[linecolor=red]{3,1}{3,3}_{$\sigma_{2}$}^{$e_{m}$}
\ncline[linecolor=red]{3,3}{2,4}_{$\sigma_{2}$}^{$e_{1}$}
\ncline[linecolor=black,shadow=true,shadowsize=1pt,shadowangle=45,shadowcolor=blue]{2,4}{1,3}_{$\sigma_{2}$}^{$e_{2}$}
\ncline[linecolor=black,shadow=true,shadowsize=1pt,shadowangle=90,shadowcolor=blue]{1,3}{1,1}^{$\sigma_{2}$}
\ncline[linecolor=red]{1,1}{3,1}^{$\sigma_{2}$}
\ncline[linecolor=red,shadow=true,shadowsize=1pt,shadowangle=90,shadowcolor=blue]{2,4}{2,6}_{$e$}^{$\sigma_{2}$}
\ncline{2,6}{2,8}_{$e_{1}'$}^{$\sigma_{1}$}
\ncarc[linecolor=red,shadow=true,shadowsize=1pt,shadowangle=-95,shadowcolor=blue,arcangle=40]{1,1}{2,6}^{$e'$}
\normalsize
\uput[0](.8,1.3){\textcolor{blue}{$C'$}}
\uput[0](-2.1,.9){\textcolor{red}{$C''$}}
\end{psmatrix}
\end{pspicture}
\caption{$C'$ and $C''$.}
\label{figProofLemStrongPanConditionForCyclesWithConstantWeights4}
\end{figure}
implies
$\lbrace 2', k \rbrace \in E$ for all $k \in V(C)$.
Hence,
Claim~\ref{itemItemLemme2jEjVC} is satisfied.

Let us now assume $\lbrace 2', j \rbrace \notin E$
for all $j \in V(C) \setminus \lbrace 2 \rbrace$.
If $\lbrace 1', 2 \rbrace \in E$,
then we can apply the same reasoning as before
(interchanging the roles of $1'$ and $2'$).
Then,
either $\lbrace 1', k \rbrace \in E$ for all $k \in V(C)$
and Claim~\ref{itemItemLemme1jEjVC} is satisfied
or $\lbrace 1', k \rbrace \notin E$ for all $k \in V(C) \setminus \lbrace 2 \rbrace$.
Therefore,
we assume henceforth that the following condition is satisfied:
\begin{equation}
\label{eqTinewaitianew}
\begin{array}{l}
\textrm{There is no edge } \lbrace 2', l \rbrace
\textrm{ with } l \in V(C) \setminus \lbrace 2 \rbrace
\textrm{ and if } \lbrace 1', 2\rbrace \in E\\
\textrm{there is also no edge } \lbrace 1', l \rbrace
\textrm{ with } l \in V(C) \setminus \lbrace 2 \rbrace.
\end{array}
\end{equation}
We now prove that $C$ is complete,
\emph{i.e.},
that Claim~\ref{itemItemLemme2jEjVCkC} is satisfied.
By contradiction
let us assume $\lbrace 2, 4 \rbrace \notin E$
(the proof is similar to the one of Claim~\ref{itemLemmeE1=1E1=e1=12Cw2w322jECjVC}).
Let us consider $i = 3$,
$A_{1} = \lbrace 4 \rbrace$, $A_{2} = \lbrace 1', 2', 2 \rbrace$,
$A = A_{1} \cup A_{2}$, and $B = (V(C) \setminus \lbrace i \rbrace) \cup A_{2}$
as represented in Figure~\ref{figProofLemStrongPanConditionForCyclesWithConstantWeights5}.
\begin{figure}[!h]
\centering
\begin{pspicture}(0,-.1)(0,2.2)
\tiny
\begin{psmatrix}[mnode=circle,colsep=0.4,rowsep=0.3]
{$4$}	&	  &  {$3$} \\
	& 	&	& 	  {$2$} & & {$2'$} & & {$1'$}\\
{$m$}	& 	&  {$1$}
\psset{arrows=-, shortput=nab,labelsep={0.05}}
\tiny
\ncline{3,1}{3,3}_{$\sigma_{2}$}^{$e_{m}$}
\ncline{3,3}{2,4}_{$\sigma_{2}$}^{$e_{1}$}
\ncline{2,4}{1,3}_{$\sigma_{2}$}^{$e_{2}$}
\ncline{1,3}{1,1}^{$\sigma_{2}$}
\ncline{1,1}{3,1}_{$\sigma_{2}$}
\ncline{2,4}{2,6}_{$e$}^{$\sigma_{2}$}
\ncline{2,6}{2,8}_{$e_{1}'$}^{$\sigma_{1}$}
\normalsize
\pspolygon[framearc=1,linestyle=dashed,linecolor=blue,linearc=.3]
(-1.7,1.45)(-1.7,2.1)(-1,2.1)(-1,1.45)
\uput[0](.2,1.8){\textcolor{blue}{$\scriptstyle{=i}$}}
\uput[0](-1.2,1.3){\textcolor{blue}{$A_{1}$}}
\pspolygon[framearc=1,linestyle=dashed,linecolor=blue,linearc=.4]
(.45,.6)(.45,1.3)(3.85,1.3)(3.85,.6)
\uput[0](2.1,.3){\textcolor{blue}{$A_{2}$}}
\pspolygon[framearc=1,linestyle=dashed,linecolor=blue,linearc=.5]
(-1.8,-.25)(-1.8,2.3)(-.9,2.2)(-.2,1.4)(4,1.4)(4,-.25)
\uput[0](-2.3,1){\textcolor{blue}{$B$}}
\end{psmatrix}
\end{pspicture}
\caption{$C$ linked to $e_{1}'$ by an edge.}
\label{figProofLemStrongPanConditionForCyclesWithConstantWeights5}
\end{figure}
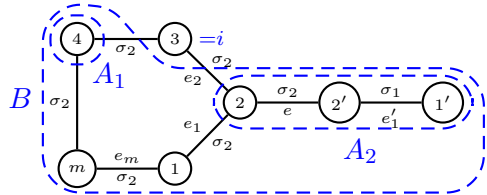
By (\ref{eqTinewaitianew}),
$\lbrace 2', 4 \rbrace \notin E$.
We consider several cases:
\begin{enumerate}
\item
\label{itemCaseIf1'2InEThen1'4NotinEThereforePminA=1'22'4}
If $\lbrace 1', 2 \rbrace \in E$,
then $\lbrace 1', 4 \rbrace \notin E$ by (\ref{eqTinewaitianew})
and therefore $\mathcal{P}_{\min}(A) = \lbrace \lbrace 1', 2, 2' \rbrace, \lbrace 4 \rbrace \rbrace$.
\item
\label{itemCaseIf1'2NotInEAnd1'4InEThenPminA=1'4,22'}
If $\lbrace 1', 2 \rbrace \notin E$ and $\lbrace 1', 4 \rbrace \in E$,
then $\mathcal{P}_{\min}(A) = \lbrace \lbrace 1', 4 \rbrace, \lbrace 2, 2' \rbrace \rbrace$.
\item
\label{itemCaseIf1'2NotInEAnd1'4NotInEThenPminA=1',22',4}
If $\lbrace 1', 2 \rbrace \notin E$ and $\lbrace 1', 4 \rbrace \notin E$,
then $\mathcal{P}_{\min}(A) =
 \lbrace \lbrace 1' \rbrace, \lbrace 2, 2' \rbrace, \lbrace 4 \rbrace \rbrace$.
\end{enumerate}
In every case $\mathcal{P}_{\min}(A \cup \lbrace i \rbrace) = \lbrace A \cup \lbrace i \rbrace \rbrace$
or $\lbrace \lbrace 1' \rbrace, A \cup \lbrace i \rbrace \setminus \lbrace 1' \rbrace \rbrace$
and $\mathcal{P}_{\min}(B) = \lbrace B \rbrace$
or $\lbrace \lbrace 1' \rbrace, B \setminus \lbrace 1' \rbrace \rbrace$.
Therefore,
taking $A' = A \cup \lbrace i \rbrace$
or $A \cup \lbrace i \rbrace \setminus \lbrace 1' \rbrace$,
we have
either $\mathcal{P}_{\min}(B)_{|A'} = \lbrace A \rbrace$
or $\lbrace \lbrace 1' \rbrace, A \setminus \lbrace 1' \rbrace \rbrace$
or $\lbrace A \setminus \lbrace 1' \rbrace \rbrace$.
As neither $A$ nor $A \setminus \lbrace 1' \rbrace$ is in $\mathcal{P}_{\min}(A)$,
we get $\mathcal{P}_{\min}(A)_{|A'} \not= \mathcal{P}_{\min}(B)_{|A'}$,
contradicting Theorem~\ref{thLGNuPNFNFSNuSSNNuSFABFABF}.
Hence,
$\lbrace 2, 4 \rbrace \in E$.
Then,
iterating as in the proof of Claim~\ref{itemLemmeE1=1E1=e1=12Cw2w322jECjVC}
we get $\lbrace 2, j \rbrace \in \hat E(C)$ for all $j \in V(C)$.
\label{PreviousProofOfClaim5InLemma}
Let us now assume $\lbrace j, k \rbrace \notin E$
for two vertices $j$ and $k$ with $3 \leq j \leq m-1$ and $k=1$ or $j+2 \leq k \leq m$.
Let us consider the edges $e'_j:=\lbrace 2 ,j\rbrace$ and $e'_k:=\lbrace 2 ,k\rbrace$ in $\hat E(C)$.
If $\lbrace 1',j\rbrace\in E$ and $\lbrace 1',k\rbrace\in E$,
then
we obtain two adjacent chordless cycles
$\tilde{C_j} = \lbrace 2, e_{j}', j,\lbrace j,1'\rbrace , 1', e_{1}',
2', e, 2 \rbrace$
and $\tilde{C_k} = \lbrace 2, e_{k}', k,\lbrace k,1'\rbrace , 1', e_{1}',
2', e, 2 \rbrace$
(by (\ref{eqTinewaitianew}) $\lbrace 2', j\rbrace$, $\lbrace 2', k\rbrace$, $\lbrace 1', 2\rbrace$ are not in $E$)
with a common edge $e'_1$ in  $E_1$
contradicting
the Adjacent cycles condition.
Hence,
we can assume that at most one of the edges
$\lbrace 1',j\rbrace$ or $\lbrace 1',k\rbrace$ is in $E$.
We now consider the cycle 
$\tilde{C} = \lbrace 2, e_{j}', j, e_{j}, j+1, \ldots, k, e_{k}', 2 \rbrace$
and $i=2$,
$A_{1} = \lbrace j \rbrace$, $A_{2} = \lbrace k \rbrace$, $A_{3} = \lbrace 1', 2' \rbrace$,
$A = A_{1} \cup A_{2} \cup A_{3}$,
$B_{1} = V(\tilde{C}) \setminus \lbrace 2 \rbrace$,
$B_{2} = A_{3}$,
and $B = B_{1}\cup B_{2}$
as represented in Figure~\ref{figProofLemStrongPanConditionForCyclesWithConstantWeights611}.
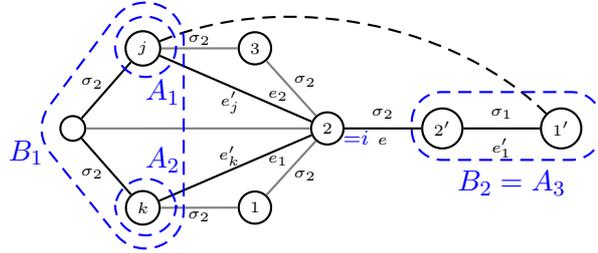
\begin{figure}[!h]
\centering
\begin{pspicture}(0,-.3)(0,3)
\tiny
\begin{psmatrix}[mnode=circle,colsep=0.5,rowsep=0.5]
	& {$j$}	&  	&  {3} \\
{ } 	& 	&	& 	&  {$2$}  & &  {$2'$} & & {$1'$}\\
	& {$k$} 	& 	& {$1$}
\psset{arrows=-, shortput=nab,labelsep={0.02}}
\tiny
\ncarc[arcangle=34,linestyle=dashed]{1,2}{2,9}
\ncline[linecolor=gray]{3,2}{3,4}_{$\sigma_{2}$}
\ncline{2,1}{3,2}_{$\sigma_{2}$}
\ncline{2,1}{1,2}^{$\sigma_{2}$}
\ncline[linecolor=gray]{3,4}{2,5}_{$\sigma_{2}$}^{$e_{1}$}
\ncline[linecolor=gray]{2,5}{1,4}_{$\sigma_{2}$}^{$e_{2}$}
\ncline[linecolor=gray]{1,2}{1,4}^{$\sigma_{2}$}
\psset{labelsep={-0.02}}
\ncline{1,2}{2,5}_{$e_{j}'$}
\ncline{3,2}{2,5}^{$e_{k}'$}
\psset{labelsep={0.1}}
\ncline[linecolor=gray]{2,1}{2,5}
\ncline{2,5}{2,7}_{$e$}^{$\sigma_{2}$}
\ncline{2,7}{2,9}_{$e_{1}'$}^{$\sigma_{1}$}
\normalsize
\uput[0](1,1){\textcolor{blue}{$\scriptstyle{=i}$}}
\pspolygon[framearc=1,linestyle=dashed,linecolor=blue,linearc=.4](-1.9,1.8)(-1.9,2.6)(-1.1,2.6)(-1.1,1.8)
\uput[0](-1.6,1.6){\textcolor{blue}{$A_{1}$}}
\pspolygon[framearc=1,linestyle=dashed,linecolor=blue,linearc=.4](-1.9,-.3)(-1.9,0.45)(-1.1,0.45)(-1.1,-.3)
\uput[0](-1.6,.7){\textcolor{blue}{$A_{2}$}}
\pspolygon[framearc=1,linestyle=dashed,linecolor=blue,linearc=.4](2,.7)(2,1.6)(4.5,1.6)(4.5,.7)
\uput[0](2.5,.4){\textcolor{blue}{$B_{2} = A_{3}$}}
\pspolygon[framearc=1,linestyle=dashed,linecolor=blue,linearc=.5]
(-3,1.1)(-1,3.8)(-1,-1.4)
\uput[0](-3.4,.8){\textcolor{blue}{$B_{1}$}}
\end{psmatrix}
\end{pspicture}
\caption{$\tilde{C} = \lbrace 2, e_{j}', j, e_{j}, j+1, \ldots,
k, e_{k}', 2 \rbrace$ and $\lbrace 1', j \rbrace \in E$.}
\label{figProofLemStrongPanConditionForCyclesWithConstantWeights611}
\end{figure}
To obtain
$\mathcal{P}_{\min}(A)$, $\mathcal{P}_{\min}(A\cup\lbrace i\rbrace)$
or $\mathcal{P}_{\min}(B)$
we only have to delete the edge $e'_1=\lbrace 1' , 2' \rbrace$ of weight $\sigma_1$.
As
$\lbrace 1',j\rbrace\notin E$
or $\lbrace 1',k\rbrace\notin E$, $\lbrace j ,k \rbrace$
cannot be a subset of any component of $\mathcal{P}_{\min}(A)$.
Therefore,
we can only have
$\mathcal{P}_{\min}(A)=\lbrace \lbrace 1' \rbrace, \lbrace 2' \rbrace, \lbrace j \rbrace, \lbrace k \rbrace \rbrace$
or $\lbrace \lbrace 1',j \rbrace, \lbrace 2' \rbrace, \lbrace k \rbrace \rbrace$
or $\lbrace \lbrace 1',k \rbrace, \lbrace 2' \rbrace,  \lbrace j \rbrace \rbrace$.
As $\lbrace 2,j \rbrace$ and $\lbrace 2,k \rbrace$ are in $\hat E(C)$
and as $i=2$,
$j$ and $k$ are connected in
$G_{A \cup \lbrace i \rbrace\setminus\lbrace 1' , 2' \rbrace}$.
Hence,
there exists $A'\in\mathcal{P}_{\min}(A\cup\lbrace i\rbrace)$
with $\lbrace j,k \rbrace\subseteq A'$.
As $j$ and $k$ are connected in
$G_{B \setminus\lbrace 1' , 2' \rbrace}$
there exists $B'\in\mathcal{P}_{\min}(B)$
with $\lbrace j,k \rbrace\subseteq B'$.
Hence,
$\lbrace j,k \rbrace\subseteq  (B'\cap A')\in\mathcal{P}_{\min}(B)_{|A'}$.
But $\lbrace j,k \rbrace$ cannot be a subset of any component of $\mathcal{P}_{\min}(A)$.
Hence,
$\mathcal{P}_{\min}(A)_{|A'}\neq\mathcal{P}_{\min}(B)_{|A'}$
contradicting Theorem~\ref{thLGNuPNFNFSNuSSNNuSFABFABF}.
Therefore,
$\lbrace j,k \rbrace\in \hat E (C)$.\\

\noindent
\ref{itemLemmew_1<w_2<w_3e1=1,2w_1w_22C_me_11V(C_m)m=3C_3e_1}.
By Proposition~\ref{lemAllEdgesOfWeightw2AreIncidentToSameEndVertexofe1},
there is a unique edge $e_{1} = \lbrace 1, 2 \rbrace$ in $E_{1}$,
all edges in $E_{2}$ are incident to the same end-vertex $2$ of $e_{1}$,
and all edges in $E_{3}$ are
linked to $2$ by $e_{1}$ or by an edge in $E_{2}$.
As $e_{1}\notin E(C_{m})$, an edge in $E(C_{m})$ has weight $\sigma_{2}$ or $\sigma_{3}$.

Let us assume $C_{m}$ non-constant.
As edges in $E_{2}$ are incident to $2$,
$2 \in V(C_{m})$ and $e_{1}$ is adjacent to $C_{m}$.
By Proposition~\ref{lemStrongPanCondForNonConstantCycle},
$e_{1} \notin \hat{E}(C_{m})$ and $C_{m}$ is complete.

Let us now assume $C_{m}$ constant.
As all edges in $E_{2}$ are incident to $2$ they cannot form a cycle,
therefore $E(C_{m}) \subseteq E_{3}$.
Then,
$2 \notin V(C_{m})$ and $e_{1} \notin \hat{E}(C_{m})$.
Let us assume $1 \notin V(C_{m})$.
Then,
an edge $e$ in $E(C_{m})$ cannot be linked to vertex $2$ by $e_{1}$,
therefore $e$ is linked to $2$ by an edge in $E_{2}$.
As $m \geq 3$ there exist at least two vertices $i$ and $j$ in $V(C_{m})$
such that $\lbrace 2, i \rbrace$ and $\lbrace 2, j \rbrace$ are in $E$.
$C_{m}$ gives two obvious paths  
$\gamma$ and $\gamma'$ linking $i$ and $j$.
Let us consider the cycles
$C_{m}' = \lbrace 2, i \rbrace \cup \gamma \cup \lbrace j, 2 \rbrace$
and
$C_{m}'' = \lbrace 2, i \rbrace \cup \gamma' \cup \lbrace j, 2 \rbrace$
as represented
in Figure~\ref{figItemProofLemEitherCmIsCompleteOrm=4C4HasConstantWeightw31'inV(C4)2'}.
\begin{figure}[!h]
\centering
\begin{pspicture}(0,-.2)(0,1.6)
\tiny
\begin{psmatrix}[mnode=circle,colsep=0.4,rowsep=0.1]
{}	&	  &  {$i$} \\
	& 	&	& 	  {} & & {$2$} & & {$1$} \\
{$j$}	& 	&  {}
\psset{arrows=-, shortput=nab,labelsep={0.05}}
\tiny
\ncline{2,6}{2,8}_{$e_{1}$}^{$\sigma_{1}$}
\ncline[linecolor=cyan]{3,1}{3,3}^{$\sigma_{3}$}
\ncline[linecolor=cyan]{3,3}{2,4}_{$\sigma_{3}$}
\ncline[linecolor=cyan]{2,4}{1,3}_{$\sigma_{3}$}
\ncline[linecolor=blue]{1,3}{1,1}_{$\sigma_{3}$}
\ncline[linecolor=blue]{1,1}{3,1}_{$\sigma_{3}$}
\ncarc[arcangle=20]{1,3}{2,6}^{$\sigma_{2}$}
\ncarc[arcangle=-30]{3,1}{2,6}_{$\sigma_{2}$}
\end{psmatrix}
\normalsize
\uput[0](-4.7,.9){\textcolor{blue}{$\gamma$}}
\uput[0](-3.6,.7){\textcolor{cyan}{$\gamma'$}}
\uput[0](-4.3,.6){$C_{m}$}
\end{pspicture}
\caption{$C_{m}' = \lbrace 2, i \rbrace \cup \gamma \cup \lbrace j, 2 \rbrace$
and $C_{m}'' = \lbrace 2, i \rbrace \cup \gamma' \cup \lbrace j, 2 \rbrace$.}
\label{figItemProofLemEitherCmIsCompleteOrm=4C4HasConstantWeightw31'inV(C4)2'}
\end{figure}
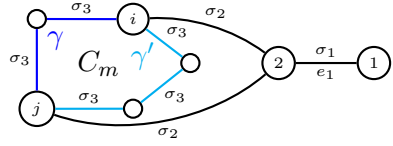
By Case~\ref{itemLemmeE12w2=w3w2}
(or Proposition~\ref{lemStrongPanCondForNonConstantCycle})
$C_{m}'$ and $C_{m}''$ are complete.
This implies $\lbrace i, j \rbrace \in E$
and $\lbrace 2, k \rbrace \in E$ for all $k \in V(C_{m})$.
Hence,
for any pair of vertices $i, j$ in $V(C_{m})$
the previous reasoning is valid and
implies $\lbrace i, j \rbrace \in E$.
Therefore,
$C_{m}$ is complete.
Let us now assume $1 \in V(C_{m})$.
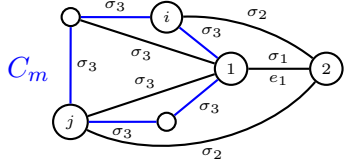
\begin{figure}[!h]
\centering
\begin{pspicture}(0,-.3)(0,1.8)
\tiny
\begin{psmatrix}[mnode=circle,colsep=0.4,rowsep=0.2]
{}	&	  &  {$i$} \\
	& 	&	& 	  {$1$} & & {$2$}\\
{$j$}	& 	&  {}
\psset{arrows=-, shortput=nab,labelsep={0.05}}
\tiny
\ncline{2,4}{2,6}_{$e_{1}$}^{$\sigma_{1}$}
\ncline[linecolor=blue]{3,1}{3,3}_{$\sigma_{3}$}
\ncline[linecolor=blue]{3,3}{2,4}_{$\sigma_{3}$}
\ncline[linecolor=blue]{2,4}{1,3}_{$\sigma_{3}$}
\ncline[linecolor=blue]{1,3}{1,1}_{$\sigma_{3}$}
\ncline[linecolor=blue]{1,1}{3,1}^{$\sigma_{3}$}
\ncline{1,1}{2,4}_{$\sigma_{3}$}
\ncline{3,1}{2,4}^{$\sigma_{3}$}
\ncarc[arcangle=20]{1,3}{2,6}^{$\sigma_{2}$}
\ncarc[arcangle=-38]{3,1}{2,6}_{$\sigma_{2}$}
\end{psmatrix}
\normalsize
\uput[0](-4.6,.7){\textcolor{blue}{$C_{m}$}}
\end{pspicture}
\caption{$C_{m}$ with $m = 5$.}
\label{figItemProofLemEitherCmIsCompleteOrm=4C4HasConstantWeightw31'inV(C4)3}
\end{figure}
Claim~\ref{itemLemmeE1=1E1=e1=12Cw2w322jECjVC}
implies $\lbrace 1, i \rbrace \in \hat{E}(C_{m})$ for all $i \in V(C_{m})$.
As $m \geq 3$,
there is at least one edge linking $2$ to $V(C_{m}) \setminus \lbrace 1 \rbrace$.
If there exist two edges $\lbrace 2, i \rbrace$ and $\lbrace 2, j \rbrace$
with $i$ and $j$ in $V(C_{m}) \setminus \lbrace 1 \rbrace$
as represented
in Figure~\ref{figItemProofLemEitherCmIsCompleteOrm=4C4HasConstantWeightw31'inV(C4)3},
then
$e_{1}$ is a common edge
for the triangles defined by $\lbrace 1,2,i \rbrace$ and $\lbrace 1,2,j \rbrace$,
contradicting the Adjacent cycles condition.
Hence,
there is exactly one edge $e$ linking $2$ to $V(C_{m}) \setminus \lbrace 1 \rbrace$.
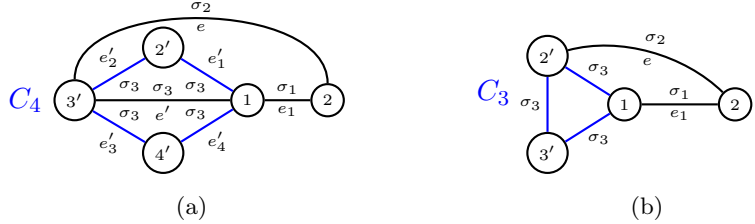
\begin{figure}[!h]
\centering
\subfloat[]{
\begin{pspicture}(-.5,-.4)(.3,2.2)
\tiny
\begin{psmatrix}[mnode=circle,colsep=.6,rowsep=0.1]
		&	{$2'$}\\
{$3'$}	& 	&	{$1$} & {$2$}\\
		& {$4'$}
\psset{arrows=-, shortput=nab,labelsep={0.05}}
\tiny
\ncline{2,3}{2,4}^{$\sigma_{1}$}_{$e_{1}$}
\ncline[linecolor=blue]{3,2}{2,3}^{$\sigma_{3}$}_{$e'_{4}$}
\ncline[linecolor=blue]{2,3}{1,2}^{$\sigma_{3}$}_{$e'_{1}$}
\ncline{2,1}{2,3}^{$\sigma_{3}$}_{$e'$}
\ncline[linecolor=blue]{2,1}{1,2}_{$\sigma_{3}$}^{$e'_{2}$}
\ncline[linecolor=blue]{2,1}{3,2}^{$\sigma_{3}$}_{$e'_{3}$}
\ncarc[arcangle=90]{2,1}{2,4}^{$\sigma_{2}$}_{$e$}
\end{psmatrix}
\normalsize
\uput[0](-4.6,.7){\textcolor{blue}{$C_{4}$}}
\end{pspicture}
\label{figLemCmCyclee1NotInCmThenCmIsCompleteIf1InV(Cm)Thenm=3-Figurea}
}
\hspace{1.4cm}
\subfloat[]{
\begin{pspicture}(-.3,-.4)(.5,2.2)
\tiny
\begin{psmatrix}[mnode=circle,colsep=.5,rowsep=0.1]
{$2'$}\\
					&	{$1$} & & {$2$}\\
{$3'$}
\psset{arrows=-, shortput=nab,labelsep={0.05}}
\tiny
\ncline[linecolor=blue]{1,1}{3,1}_{$\sigma_{3}$}
\ncline[linecolor=blue]{1,1}{2,2}^{$\sigma_{3}$}
\ncline[linecolor=blue]{3,1}{2,2}_{$\sigma_{3}$}
\ncline{2,2}{2,4}_{$e_{1}$}^{$\sigma_{1}$}
\ncarc[arcangle=30]{1,1}{2,4}^{$\sigma_{2}$}_{$e$}
\end{psmatrix}
\normalsize
\uput[0](-3.8,.8){\textcolor{blue}{$C_{3}$}}
\end{pspicture}
\label{figLemCmCyclee1NotInCmThenCmIsCompleteIf1InV(Cm)Thenm=3-Figureb}
}
\caption{$C_{4}$
and $\tilde{C}_{3} = \lbrace 1, e_{1}, 2, e, 3', e', 1 \rbrace$.
$C_{3}$ and $\tilde{C}_{3} = \lbrace 1, e_{1}, 2, e, 2', e_{1}', 1 \rbrace$.}
\label{figLemCmCyclee1NotInCmThenCmIsCompleteIf1InV(Cm)Thenm=3}
\end{figure}
If $m \geq 5$,
then there is at least one edge in $E(C_m)$ neither incident to $1$
nor linked to~$2$,
a contradiction.
If $m = 4$,
then we necessarily have $e = \lbrace 2, 3' \rbrace$
as represented in Figure~\ref{figLemCmCyclee1NotInCmThenCmIsCompleteIf1InV(Cm)Thenm=3-Figurea},
otherwise we get the same contradiction. 
Let us denote by $e'_{1} = \lbrace 1, 2' \rbrace$,
$e'_{2} = \lbrace 2', 3' \rbrace$,
$e'_{3} = \lbrace 3', 4' \rbrace$,
and $e'_{4} = \lbrace 4', 1 \rbrace$
the edges in $E(C_{4})$
and by $e' = \lbrace 1, 3' \rbrace$ the chord of $C_{4}$ incident to~$1$.
If $\lbrace 2',4' \rbrace \in E$,
we are done.
So,
let us assume $\lbrace 2', 4' \rbrace \notin E$.
Let us consider $i = 1$,
$A = \lbrace 2, 2', 4' \rbrace$,
and $B = A \cup \lbrace 3' \rbrace$.
Then,
$\mathcal{P}_{\min}(A) = \lbrace \lbrace 2 \rbrace, \lbrace 2' \rbrace, \lbrace 4' \rbrace \rbrace$,
$\mathcal{P}_{\min}(A \cup \lbrace i \rbrace) =
\lbrace \lbrace 2 \rbrace, \lbrace 1, 2', 4' \rbrace \rbrace$,
and $\mathcal{P}_{\min}(B) = \lbrace \lbrace 2 \rbrace, \lbrace 2', 3', 4' \rbrace \rbrace$.
Taking $A' = \lbrace 1, 2', 4' \rbrace$ we get
$\mathcal{P}_{\min}(A)_{|A'} =
\lbrace \lbrace 2' \rbrace, \lbrace 4' \rbrace \rbrace
\not= \lbrace 2', 4' \rbrace = \mathcal{P}_{\min}(B)_{A'}$
and it contradicts Theorem~\ref{thLGNuPNFNFSNuSSNNuSFABFABF}.
Hence,
$m = 3$
as represented in Figure~\ref{figLemCmCyclee1NotInCmThenCmIsCompleteIf1InV(Cm)Thenm=3-Figureb}.
\end{proof}

\section{Graphs satisfying inheritance of convexity}
\label{SectionWeighthedgraphsforwhichInheritanceofconvexity}

We provide characterizations of weighted graphs
satisfying inheritance of convexity with $\mathcal{P}_{\min}$.
We start with connected weighted graphs.

\subsection{Connected graphs with two edge-weights}

\begin{theorem}
\label{lemInheritanceOfConvexityForPminInheritanceOfConvexityForPMG1Cycle-Complete}
Let $G=(N,E,w)$ be a connected weighted graph.
Let us assume that the edge-weights have only two different values $\sigma_{1} < \sigma_{2}$
and $|E_{1}| \geq 2$.
Then,
there is inheritance of convexity for $\mathcal{P}_{\min}$
if and only if
\begin{enumerate}
\item
\label{itemLemInheritanceOfConvexityForPMOnG1=(N,E2)}
All edges in $E_{1}$ are incident to the same vertex $1$
and all edges in $E_{2}$ are linked to~$1$ by an edge in $E_{1}$.
\item
\label{itemLemInheritanceOfConvexityForPMOnG1=(N,E3)}
One of the  following two equivalent conditions is satisfied:
\begin{enumerate}
\item
\label{itemLemInheritanceOfConvexityForPMOnG1=(N,E3)a}
There is inheritance of convexity for $\mathcal{P}_{M}$
on the subgraph $G_{1} = (N, E \setminus E_{1})$.
\item
\label{itemLemInheritanceOfConvexityForPMOnG1=(N,E3)b}
$G_{1} = (N, E \setminus E_{1})$ is cycle-complete.
\end{enumerate}
\end{enumerate}
\end{theorem}

We give
in Figure~\ref{figPropInheritanceOfConvexityForPminInheritanceOfConvexityForPMG1Cycle-Complete}
an example of a graph satisfying
conditions \ref{itemLemInheritanceOfConvexityForPMOnG1=(N,E2)}
and \ref{itemLemInheritanceOfConvexityForPMOnG1=(N,E3)}
of Theorem~\ref{lemInheritanceOfConvexityForPminInheritanceOfConvexityForPMG1Cycle-Complete}.
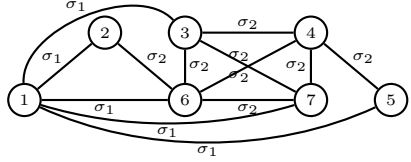
\begin{figure}[!h]
\centering
\begin{pspicture}(0,-.6)(0,1.4)
\tiny
\begin{psmatrix}[mnode=circle,colsep=.6,rowsep=0.4]
& {$2$} & {$3$} & & {$4$}\\
{$1$} &					&	{$6$} & & {$7$} & {$5$}
\psset{arrows=-, shortput=nab,labelsep={0.02}}
\tiny
\ncline{1,2}{2,1}_{$\sigma_{1}$}
\ncline{1,2}{2,3}^{$\sigma_{2}$}
\ncline{2,1}{2,3}_{$\sigma_{1}$}
\ncline{2,3}{2,5}_{$\sigma_{2}$}
\ncline{2,3}{1,3}_{$\sigma_{2}$}
\ncline{1,3}{1,5}^{$\sigma_{2}$}
\ncline{1,5}{2,5}_{$\sigma_{2}$}
\ncline{1,3}{2,5}_{$\sigma_{2}$}
\ncline{2,3}{1,5}^{$\sigma_{2}$}
\ncarc[arcangle=-17]{2,1}{2,5}_{$\sigma_{1}$}
\ncarc[arcangle=70]{2,1}{1,3}^{$\sigma_{1}$}
\ncline{1,5}{2,6}^{$\sigma_{2}$}
\ncarc[arcangle=-25]{2,1}{2,6}_{$\sigma_{1}$}
\end{psmatrix}
\normalsize
\end{pspicture}
\caption{Every edge in $E_{2}$ is linked to $1$
by an edge in $E_{1}$
and the cycle defined by $3, 4, 6, 7$ is complete.}
\label{figPropInheritanceOfConvexityForPminInheritanceOfConvexityForPMG1Cycle-Complete}
\end{figure}

\begin{proof}
Conditions~\ref{itemLemInheritanceOfConvexityForPMOnG1=(N,E3)a}
and \ref{itemLemInheritanceOfConvexityForPMOnG1=(N,E3)b} are equivalent
by Theorem~\ref{NouwelandandBorm1991}
\citep{NouwelandandBorm1991}.
By Proposition~\ref{propEitherOnlyOneEdgeOfWeightw1OrAllEdgesOFWeightw1AreIncidentToSameVertex},
Lemma~\ref{LemP_MinMax(w1,w2)<w(e)e'InELinkingeTo2}
and Proposition~\ref{lemStrongPanConditionForCyclesWithConstantWeights} 
(Claim~\ref{itemLemmeE12w2=w3w2}),
Conditions \ref{itemLemInheritanceOfConvexityForPMOnG1=(N,E2)}
and \ref{itemLemInheritanceOfConvexityForPMOnG1=(N,E3)} are necessary.
We now prove their sufficiency.
Let $(N,v)$ be a convex game.
We denote by $(N, \overline{v})$ (resp. $(N, v^{M})$)
the restricted game associated with $\mathcal{P}_{\min}$
(resp. $\mathcal{P}_{M}$)
on $G$ (resp. $G_{1}$).
Let us consider $i \in N$
and subsets $A \subseteq B \subseteq N \setminus \lbrace i \rbrace$.
We consider several cases to prove that the following inequality is satisfied:
\begin{equation}
\label{eqProofPropInheritanceOfConvexityForPminInheritanceOfConvexityForPMG1Cycle-Complete}
\overline{v}(B \cup \lbrace i \rbrace) - \overline{v}(B) \geq
\overline{v}(A \cup \lbrace i \rbrace) - \overline{v}(A).
\end{equation}

Let us first assume
$E(B) \subseteq E_{2}$
(resp. $E(A \cup \lbrace i \rbrace) \subseteq E_{2}$).
Then,
$\mathcal{P}_{\min}(A)$
and $\mathcal{P}_{\min}(B)$
(resp. $\mathcal{P}_{\min}(A \cup \lbrace i \rbrace)$)
are singleton partitions
and (\ref{eqProofPropInheritanceOfConvexityForPminInheritanceOfConvexityForPMG1Cycle-Complete})
is equivalent to
$\overline{v}(B \cup \lbrace i \rbrace) \geq
\overline{v}(A \cup \lbrace i \rbrace)$
(resp. $\overline{v}(B \cup \lbrace i \rbrace) - \overline{v}(B) \geq 0$.).
This last inequality is satisfied
as $(N, \overline{v})$ is superadditive
(cf. Corollary~\ref{corIGNEiafttsihfv}).

Let us now assume $E(B) \cap E_{1} \not= \emptyset$
and $E(A \cup \lbrace i \rbrace) \cap E_{1} \not= \emptyset$.
Then,
we also have $E(B \cup \lbrace i \rbrace) \cap E_{1} \not= \emptyset$.
By Condition~\ref{itemLemInheritanceOfConvexityForPMOnG1=(N,E2)}
any edge in $E_{1}$ is incident to $1$,
therefore
we have $1 \in B$ and $i \not= 1$
as $B \subseteq N \setminus \lbrace i \rbrace$.
Then,
as $E(A \cup \lbrace i \rbrace) \cap E_{1} \not= \emptyset$,
we necessarily have $1 \in A$.
If $E(A) \not= \emptyset$,
Condition~\ref{itemLemInheritanceOfConvexityForPMOnG1=(N,E2)}
implies $E(A) \cap E_{1} \not= \emptyset$ 
and then $\overline{v}(A) = v^{M}(A)$.
If $E(A) = \emptyset$,
then we trivially have $\overline{v}(A) = v^{M}(A)$.
Hence,
(\ref{eqProofPropInheritanceOfConvexityForPminInheritanceOfConvexityForPMG1Cycle-Complete})
is equivalent to
$
v^{M}(B \cup \lbrace i \rbrace) - v^{M}(B)
\geq v^{M}(A \cup \lbrace i \rbrace) - v^{M}(A)$,
and by Condition~\ref{itemLemInheritanceOfConvexityForPMOnG1=(N,E3)}
this last inequality is satisfied.
\end{proof}

\begin{theorem}
\label{PropInheritanceOfConvexityForPminOnGIffCofGWithWeightwEitherCompleteOrAllVerticesLinkedToTheSameEndVertexOfe1}
Let $G=(N,E,w)$ be a connected weighted graph.
Let us assume that the edge-weights have only two different values
$\sigma_{1} < \sigma_{2}$ and $|E_{1}| = 1$.
Let $e_{1} = \lbrace 1, 2 \rbrace$ be the unique edge in $E_{1}$.
Then,
there is inheritance of convexity for $\mathcal{P}_{\min}$
if and only if
\begin{enumerate}
\item
\label{itThThereExistsAtMostOneCycleCWithoutChordContaininge1}
There exists at most one chordless cycle  containing $e_1$.
\item
\label{itThCycleCOfGWithConstantEdge-Weightw2EitherCompleteOrAllVerticesLinkedToSameEndVertexOfe1}
For every cycle $C$  with constant weight $\sigma_{2}$
either $C$ is complete
or all vertices of $C$ are linked to the same end-vertex of $e_{1}$.
\end{enumerate}
\end{theorem}

We give in Figure~\ref{figPropInheritanceOfConvexityForPminOnGIffCofGWithWeightw2EitherCompleteOrAllVerticesLinkedToTheSameEndVertexOfe1}
an example of a graph satisfying conditions~\ref{itThThereExistsAtMostOneCycleCWithoutChordContaininge1}
and~\ref{itThCycleCOfGWithConstantEdge-Weightw2EitherCompleteOrAllVerticesLinkedToSameEndVertexOfe1}
of Theorem~\ref{PropInheritanceOfConvexityForPminOnGIffCofGWithWeightwEitherCompleteOrAllVerticesLinkedToTheSameEndVertexOfe1}.
\begin{figure}[!h]
\centering
\begin{pspicture}(0,-.3)(0,1.4)
\tiny
\begin{psmatrix}[mnode=circle,colsep=.7,rowsep=0.5]
{} & {} & 	&	  & {} & {}\\
{} & {} & {$1$} & {$2$} & {} & {}
\psset{arrows=-, shortput=nab,labelsep={0.02}}
\tiny
\ncline{2,3}{2,4}^{$\sigma_{1}$}_{$e_{1}$}
\ncline{2,2}{2,3}_{$\sigma_{2}$}
\ncline{1,1}{1,2}^{$\sigma_{2}$}
\ncline{1,1}{2,1}_{$\sigma_{2}$}
\ncline{1,1}{2,2}_{$\sigma_{2}$}
\ncline{1,2}{2,2}^{$\sigma_{2}$}
\ncline{2,1}{2,2}_{$\sigma_{2}$}
\ncline{2,1}{1,2}^{$\sigma_{2}$}
\ncline{1,6}{2,6}^{$\sigma_{2}$}
\ncline{2,5}{2,6}_{$\sigma_{2}$}
\ncline{1,5}{1,6}^{$\sigma_{2}$}
\ncline{1,5}{2,5}^{$\sigma_{2}$}
\ncarc[arcangle=20,linecolor=blue]{2,4}{1,5}_{$\sigma_{2}$}
\ncline[linecolor=blue]{2,4}{2,5}_{$\sigma_{2}$}
\ncarc[arcangle=80,linecolor=blue]{2,4}{1,6}^{$\sigma_{2}$}
\ncarc[arcangle=-35,linecolor=blue]{2,4}{2,6}_{$\sigma_{2}$}
\end{psmatrix}
\normalsize
\end{pspicture}
\caption{Either a constant cycle is complete or
all its vertices are linked to one end-vertex of $e_{1}$.}
\label{figPropInheritanceOfConvexityForPminOnGIffCofGWithWeightw2EitherCompleteOrAllVerticesLinkedToTheSameEndVertexOfe1}
\end{figure}
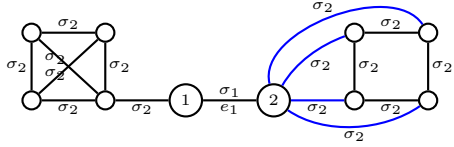

\begin{proof}
Condition~\ref{itThThereExistsAtMostOneCycleCWithoutChordContaininge1} is necessary
by the Adjacent cycles condition.
By Proposition~\ref{lemStrongPanConditionForCyclesWithConstantWeights}
(Claims~\ref{itemLemmeE1=1E&=e1w2e1}, \ref{itemLemmeE1=1E1=e1=12Cw2w322jECjVC},
and \ref{itemLemmeE1=1E1=e1e1=12Cw2w32e122'w2w3})
Condition~\ref{itThCycleCOfGWithConstantEdge-Weightw2EitherCompleteOrAllVerticesLinkedToSameEndVertexOfe1}
is also necessary.
We now prove their sufficiency.
Let $(N,v)$ be a convex game
and let us consider
$i \in N$ and subsets $A \subseteq B \subseteq N \setminus \lbrace i \rbrace$.
We consider several cases to prove that the following inequality is satisfied:
\begin{equation}
\label{eqProofv(BUi)-v(B)>=v(AUi)-v(A)}
\overline{v}(B \cup \lbrace i \rbrace) - \overline{v}(B)
\geq \overline{v}(A \cup \lbrace i \rbrace) - \overline{v}(A).
\end{equation}

Let us first assume
$E(B) \subseteq E_{2}$
(resp. $E(A \cup \lbrace i \rbrace) \subseteq E_{2}$).
Then,
we can conclude as in Case~1 in the proof
of Theorem~\ref{lemInheritanceOfConvexityForPminInheritanceOfConvexityForPMG1Cycle-Complete}.

Let us now assume $e_{1} \in E(B)$
and $e_{1} \in E(A \cup \lbrace i \rbrace)$.
Then,
we also have $e_{1} \in E(B \cup \lbrace i \rbrace)$.
If $e_{1} \notin E(A)$,
then $i = 1$ or $2$ as $e_{1} \in E(A \cup \lbrace i \rbrace)$
but it contradicts $e_{1} \in E(B)$ as $B \subseteq N \setminus \lbrace i \rbrace$.
Therefore,
we also have $e_{1} \in E(A)$ and (\ref{eqProofv(BUi)-v(B)>=v(AUi)-v(A)})
is equivalent to
$v^{M}(B \cup \lbrace i \rbrace) - v^{M}(B) \geq v^{M}(A \cup \lbrace i \rbrace) - v^{M}(A)$
where $(N,v^M)$ is
associated with $G_1=(N, E\setminus E_1)$.
Let $\mathcal{P}_{M}(A) = \lbrace A_{1}, A_{2}, \ldots, A_{p} \rbrace$
(resp.  $\mathcal{P}_{M}(B) = \lbrace B_{1}, B_{2}, \ldots, B_{q} \rbrace$)
be the partition of $A$
(resp. $B$)
into connected components in $G_1$.
If there is no link between $i$ and $A$,
then $\mathcal{P}_{M}(A \cup \lbrace i \rbrace) = \lbrace \mathcal{P}_{M}(A), \lbrace i \rbrace \rbrace$
and $v^{M}(A \cup \lbrace i \rbrace) - v^{M}(A) = v(\lbrace i \rbrace) = 0$.
Then,
(\ref{eqProofv(BUi)-v(B)>=v(AUi)-v(A)})
is equivalent to
$\overline{v}(B \cup \lbrace i \rbrace) - \overline{v}(B) \geq 0$
and this last inequality is satisfied
as $(N, \overline{v})$ is superadditive.
Otherwise,
we have
$\mathcal{P}_{M}(A \cup \lbrace i \rbrace) =
\lbrace A_{1} \cup \ldots \cup A_{r} \cup \lbrace i \rbrace, A_{r+1}, \ldots, A_{p} \rbrace$
(resp. $\mathcal{P}_{M}(B \cup \lbrace i \rbrace) =
\lbrace B_{1} \cup \ldots \cup B_{s} \cup \lbrace i \rbrace, B_{s+1}, \ldots, B_{q} \rbrace$)
with $1 \leq r \leq p$
(resp. $1 \leq s \leq q$),
after reordering if necessary.
Then,
setting $A' = A_{1} \cup \ldots \cup A_{r}$
and $B' = B_{1} \cup \ldots \cup B_{s}$,
(\ref{eqProofv(BUi)-v(B)>=v(AUi)-v(A)}) is equivalent to
\begin{equation}
\label{eqProofPropInheritanceOfConvexityForPminOnGIffCofGWithWeightw2EitherCompleteOrAllVerticesLinkedToTheSameEndVertexOfe1vB'Ui-SumvBj>=vA'Ui-SumvAj}
v(B' \cup \lbrace i \rbrace)
- \sum_{j=1}^{s} v(B_{j}) \geq
v(A' \cup \lbrace i \rbrace)
- \sum_{j=1}^{r} v(A_{j}).
\end{equation}
Let us observe that
obviously
$\mathcal{P}_{M}(A)$ is a refinement of $\mathcal{P}_{M}(B)_{|A}$.
To complete the proof we need the following claim.

\begin{claim}
\label{lemAssumptionOnCyclesWithConstantWeightW2ImpliesPM(B)_A'=PM(A)_A'}
Conditions~\ref{itThThereExistsAtMostOneCycleCWithoutChordContaininge1}
and~\ref{itThCycleCOfGWithConstantEdge-Weightw2EitherCompleteOrAllVerticesLinkedToSameEndVertexOfe1}
imply
$A_{j} \subseteq B_{j}$, for all $j$,
$1 \leq j \leq r$,
after renumbering if necessary.
\end{claim}

By contradiction,
let us assume that two components $A_{1}$ and $A_{2}$ of
$\mathcal{P}_{M}(A)_{|A'}$
are subsets of the same component $B_{1} \in \mathcal{P}_{M}(B)$,
after renumbering if necessary.
Let $\tilde{e}_{1} = \lbrace i, k_{1} \rbrace$ (resp. $\tilde{e}_{2} = \lbrace i, k_{2} \rbrace$)
be an edge linking $i$ to $A_{1}$ (resp. $A_{2}$).
As $i \notin \lbrace 1, 2 \rbrace$,
$\tilde{e}_{1}$ and $\tilde{e}_{2}$ are in $E_{2}$.
As $B_{1}$ is connected,
there exists an elementary path $\gamma$ in $\tilde{G}_{B_1}$
linking $k_{1} \in A_{1}$ to $k_{2} \in A_{2}$.
We obtain a simple cycle
$C = \lbrace i, \tilde{e}_{1}, k_{1} \rbrace \cup \gamma \cup \lbrace k_{2}, \tilde{e}_{2}, i \rbrace$ 
of constant weight $\sigma_2$.
If $C$ is complete,
then $\lbrace k_{1}, k_{2} \rbrace$ is a chord of $C$.
Condition~\ref{itThThereExistsAtMostOneCycleCWithoutChordContaininge1}
implies $\lbrace k_{1}, k_{2} \rbrace\neq e_1$. 
Then,
$\lbrace k_{1}, k_{2} \rbrace$ links $A_{1}$ to $A_{2}$ in $G_{1}$,
a contradiction.
If $C$ is not complete,
then by Condition~\ref{itThCycleCOfGWithConstantEdge-Weightw2EitherCompleteOrAllVerticesLinkedToSameEndVertexOfe1}
all vertices of $C$ are linked to the same end-vertex $v$ of $e_{1}$.
We can assume w.l.o.g. $v=1$
as represented
in Figure~\ref{figProofLemAssumptionOnCyclesWithConstantWeightW2ImpliesPM(B)_A'=PM(A)_A'}.
As $e_{1} \in E(A)$ and as $k_{1}$ and $k_{2}$ are in $A$,
we have $\lbrace 1, k_{1} \rbrace$ and $\lbrace 1, k_{2} \rbrace$ in $E(A)$.
We also have $\lbrace 1, k_{1} \rbrace\neq e_1$ and $\lbrace 1, k_{2} \rbrace\neq e_1$,
otherwise $e_{1}$ would be a chord of a cycle contradicting
Condition~\ref{itThThereExistsAtMostOneCycleCWithoutChordContaininge1}.
Then,
$A_{1}$ and $A_{2}$ are part of a connected component of $A$ in $G_{1}$,
a contradiction.
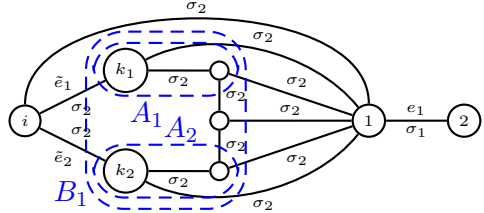
\begin{figure}[!h]
\centering
\begin{pspicture}(0,-.3)(0,2.4)
\tiny
\begin{psmatrix}[mnode=circle,colsep=.8,rowsep=0.1]
		&	{$k_{1}$} & {}\\
{$i$}	& 	&	{} &  & {$1$} & {$2$}\\
		& {$k_{2}$} & {}
\psset{arrows=-, shortput=nab,labelsep={0.05}}
\tiny
\ncline{2,1}{1,2}_{$\sigma_{2}$}^{$\tilde{e}_{1}$}
\ncline{2,1}{3,2}^{$\sigma_{2}$}_{$\tilde{e}_{2}$}
\ncline{1,2}{1,3}_{$\sigma_{2}$}
\ncline{1,3}{2,3}^{$\sigma_{2}$}
\ncline{2,3}{3,3}^{$\sigma_{2}$}
\ncline{3,2}{3,3}_{$\sigma_{2}$}
\ncline{2,3}{2,5}^{$\sigma_{2}$}
\ncline{1,3}{2,5}^{$\sigma_{2}$}
\ncline{3,3}{2,5}_{$\sigma_{2}$}
\ncline{2,5}{2,6}_{$\sigma_{1}$}^{$e_{1}$}
\ncarc[arcangle=-40]{3,2}{2,5}_{$\sigma_{2}$}
\ncarc[arcangle=40]{1,2}{2,5}^{$\sigma_{2}$}
\ncarc[arcangle=90]{2,1}{2,5}^{$\sigma_{2}$}
\end{psmatrix}
\normalsize
\pspolygon[framearc=1,linestyle=dashed,linecolor=blue,linearc=.4](-5.1,1)(-5.1,1.7)(-3.2,1.7)(-3.2,1)
\uput[0](-4.8,.75){\textcolor{blue}{$A_{1}$}}
\pspolygon[framearc=1,linestyle=dashed,linecolor=blue,linearc=.4](-5.1,-.35)(-5.1,.35)(-3.2,.35)(-3.2,-.35)
\uput[0](-4.35,.55){\textcolor{blue}{$A_{2}$}}
\pspolygon[framearc=1,linestyle=dashed,linecolor=blue,linearc=.5]
(-5.2,-.5)(-5.2,1.85)(-3.1,1.85)(-3.1,-.5)
\uput[0](-5.8,-.3){\textcolor{blue}{$B_{1}$}}
\end{pspicture}
\caption{$e_{1}$ in $E(A)$ and $k_{1}$ in $A_{1}$, $k_{2}$ in $A_{2}$.}
\label{figProofLemAssumptionOnCyclesWithConstantWeightW2ImpliesPM(B)_A'=PM(A)_A'}
\end{figure}

We now end the proof of
Theorem~\ref{PropInheritanceOfConvexityForPminOnGIffCofGWithWeightwEitherCompleteOrAllVerticesLinkedToTheSameEndVertexOfe1}.
By Claim~\ref{lemAssumptionOnCyclesWithConstantWeightW2ImpliesPM(B)_A'=PM(A)_A'},
we have $\mathcal{P}_{M}(A') = \mathcal{P}_{M}(B')_{|A'}$.
Then,
Lemma~\ref{lemPNABNPAPBAvNvB-vB>=vA-vA}
applied to $\mathcal{P}_{M}$
and the family $\mathcal{F}$
of connected subsets of $N$
implies
$v(B') - \sum_{j=1}^{s} v(B_{j})
\geq
v(A') - \sum_{j=1}^{r} v(A_{j})$.
The convexity of $(N,v)$ also implies
$v(B' \cup \lbrace i \rbrace) - v(B')
\geq v(A' \cup \lbrace i \rbrace) - v(A')$.
Adding these last inequalities
we obtain
(\ref{eqProofPropInheritanceOfConvexityForPminOnGIffCofGWithWeightw2EitherCompleteOrAllVerticesLinkedToTheSameEndVertexOfe1vB'Ui-SumvBj>=vA'Ui-SumvAj}).
\end{proof}

\begin{theorem}
\label{PropInheritanceOfConvexityForPminOnGIffCofGWithWeightw2EitherCompleteOrAllVerticesLinkedToTheSameEndVertexOfe12}
Let $G=(N,E,w)$ be a connected weighted graph.
Let us assume that the edge-weights have three different values
$\sigma_{1} < \sigma_{2}<\sigma_3$.
Then,
there is inheritance of convexity for $\mathcal{P}_{\min}$
if and only if
\begin{enumerate}
\item
\label{ThereExistsOnlyOneEdgeInE1}
There is only one edge
$e_{1} = \lbrace 1, 2 \rbrace$ in $E_1$.
\item
\label{EveryEdgeOfWeightSigma2IsIncidentToTheSameEnd-Vertex2Ofe1}
Every edge in $E_{2}$ is incident to the same end-vertex $2$ of $e_1$.
\item
\label{EveryEdgeOfWeightSigma3IsConnectedTo2Bye1OrByAnEdgeOfWeightSigma2}
Every edge
in $E_3$
is linked to $2$ by $e_1$ or by an edge in $E_2$.
\item
\label{ThereExistsAtMostOneChordlessCycleCmWithm=3or4Containinge1}
There exists at most one chordless cycle $\tilde{C}_m$ with $m=3$ or $4$
containing $e_1$.
\item
\label{EveryCycleWhichDoesNotContainse1IsComplete}
$G_1=(N, E\setminus E_1)$ is cycle-complete.
\end{enumerate}
Moreover,
these conditions imply:
\begin{enumerate}
\setcounter{enumi}{5}
\item
\label{itemThIfACycleCmDoesNotContaine1AndIf1inV(Cm)Thenm=3}
If a cycle $C_m$ does not contain $e_1$
and if $1 \in V(C_m)$,
then $m=3$ and such a cycle is unique,
has constant weight $\sigma_3$,
and is adjacent to a unique triangle $\tilde{C}_{3}$ containing $e_{1}$.
Moreover,
$\tilde{C}_{3} = \lbrace 1, e_{1}, 2, e_{2}, 3, e_{3}, 1 \rbrace$
with
$w_{i} = \sigma_{i}$ for $i \in \lbrace 1, 2, 3 \rbrace$,
and
$E(C_{3}) \cap E(\tilde{C}_{3}) = \lbrace e_{3} \rbrace$.
\end{enumerate}
\end{theorem}

\begin{remark}
The Star, Path, Cycle, Pan 
and Adjacent cycles  conditions are straightforward consequences
of Conditions~\ref{ThereExistsOnlyOneEdgeInE1} to \ref{EveryCycleWhichDoesNotContainse1IsComplete}
in Theorem~\ref{PropInheritanceOfConvexityForPminOnGIffCofGWithWeightw2EitherCompleteOrAllVerticesLinkedToTheSameEndVertexOfe12}.
\end{remark}

\begin{proof}[Proof of theorem~\ref{PropInheritanceOfConvexityForPminOnGIffCofGWithWeightw2EitherCompleteOrAllVerticesLinkedToTheSameEndVertexOfe12}]
By Proposition~\ref{lemAllEdgesOfWeightw2AreIncidentToSameEndVertexofe1}
and the Adjacent cycles condition,
Conditions~\ref{ThereExistsOnlyOneEdgeInE1}
to~\ref{ThereExistsAtMostOneChordlessCycleCmWithm=3or4Containinge1}
are necessary.
By Proposition~\ref{lemStrongPanConditionForCyclesWithConstantWeights}
(Claim~\ref{itemLemmew_1<w_2<w_3e1=1,2w_1w_22C_me_11V(C_m)m=3C_3e_1}),
Condition~\ref{EveryCycleWhichDoesNotContainse1IsComplete}
is necessary.
We now prove their sufficiency.
Let us consider a convex game $(N,v)$,
$i \in N$ and subsets $A \subseteq B \subseteq N \setminus \lbrace i \rbrace$.
We have to prove that the following inequality is satisfied:
\begin{equation}
\label{eqProofPropInheritanceOfConvexityForPminOnGIffCofGWithWeightw2EitherCompleteOrAllVerticesLinkedToTheSameEndVertexOfe1vB'Ui-vB'>=vA'Ui-vA'1}
\overline{v}(B \cup \lbrace i \rbrace) - \overline{v}(B)
\geq \overline{v}(A \cup \lbrace i \rbrace) -\overline{v}(A).
\end{equation}
Let us note that
if $i$ is not linked to $A$,
then
(\ref{eqProofPropInheritanceOfConvexityForPminOnGIffCofGWithWeightw2EitherCompleteOrAllVerticesLinkedToTheSameEndVertexOfe1vB'Ui-vB'>=vA'Ui-vA'1})
is trivially satisfied
as $(N, \overline{v})$ is superadditive (cf. Corollary~\ref{corIGNEiafttsihfv}).
If $i = 2$,
Conditions~\ref{ThereExistsOnlyOneEdgeInE1}
and~\ref{EveryEdgeOfWeightSigma2IsIncidentToTheSameEnd-Vertex2Ofe1}
imply $E(A) \subseteq E(B) \subseteq E_{3}$.
Then,
$\mathcal{P}_{\min}(A)$ and $\mathcal{P}_{\min}(B)$ are singletons partitions
and (\ref{eqProofPropInheritanceOfConvexityForPminOnGIffCofGWithWeightw2EitherCompleteOrAllVerticesLinkedToTheSameEndVertexOfe1vB'Ui-vB'>=vA'Ui-vA'1}) 
is equivalent to
$\overline{v}(B\cup\lbrace i\rbrace)\ge \overline{v}(A\cup\lbrace i\rbrace)$.
This last inequality is satisfied as $(N, \overline{v})$ is superadditive.

We thereafter assume $i$ linked to $A$ by at least one edge
and $i \not = 2$,
and consider several cases.\\

\noindent
\textbf{Case 1}
Let us assume
$2 \notin A$.
Conditions~\ref{ThereExistsOnlyOneEdgeInE1}
and~\ref{EveryEdgeOfWeightSigma2IsIncidentToTheSameEnd-Vertex2Ofe1}
imply $E(A) \subseteq E(A \cup \lbrace i \rbrace) \subseteq E_3$.
Then,
$\mathcal{P}_{\min}(A)$
and $\mathcal{P}_{\min}(A \cup \lbrace i \rbrace)$ are singleton partitions
and
$\overline{v}(A \cup\lbrace i\rbrace) - \overline{v}(A) = v(\lbrace i \rbrace) = 0$.
As $(N, \overline{v})$ is superadditive
(cf. Corollary~\ref{corIGNEiafttsihfv}),
(\ref{eqProofPropInheritanceOfConvexityForPminOnGIffCofGWithWeightw2EitherCompleteOrAllVerticesLinkedToTheSameEndVertexOfe1vB'Ui-vB'>=vA'Ui-vA'1}) is satisfied.\\

\noindent
\textbf{Case 2}
Let us assume $2\in A$ and $1 \in A$.
Then,
$e_{1}$ belongs to $E(A)$,
$E(A\cup \lbrace i\rbrace)$, $E(B)$, $E(B\cup \lbrace i\rbrace)$,
and (\ref{eqProofPropInheritanceOfConvexityForPminOnGIffCofGWithWeightw2EitherCompleteOrAllVerticesLinkedToTheSameEndVertexOfe1vB'Ui-vB'>=vA'Ui-vA'1})
is equivalent to
$v^M(B\cup\lbrace i\rbrace) - v^M(B) \geq
v^M(A\cup\lbrace i\rbrace) - v^M(A)$
where $(N,v^{M})$ is the Myerson restricted game associated
with $G_1 = ( N, E\setminus \lbrace e_1 \rbrace)$.
By Condition~\ref{EveryCycleWhichDoesNotContainse1IsComplete},
$G_1$ is cycle-complete.
Then,
$(N, v^M)$ is convex by Theorem~\ref{NouwelandandBorm1991}
and therefore
(\ref{eqProofPropInheritanceOfConvexityForPminOnGIffCofGWithWeightw2EitherCompleteOrAllVerticesLinkedToTheSameEndVertexOfe1vB'Ui-vB'>=vA'Ui-vA'1}) is satisfied.\\

\noindent
\textbf{Case 3}
Let us assume $2 \in A$ and $1 \in B \setminus A$.
Then,
$e_{1} \in E(B) \setminus E(A)$,
$i \notin \lbrace 1,2 \rbrace$,
$e_{1} \notin E(A \cup \lbrace i \rbrace)$,
and $e_1\in E(B \cup \lbrace i \rbrace)$.
As in Case~2,
$\overline{v}(B \cup \lbrace i \rbrace) - \overline{v}(B) =
v^{M}(B \cup \lbrace i \rbrace) - v^{M}(B)$
where $(N,v^{M})$ is
associated with~$G_1 = ( N, E\setminus \lbrace e_1 \rbrace)$.
As $i$ is linked to $A$ in $G$,
and as $i \notin \lbrace 1, 2 \rbrace$, 
$i$ is also linked to $A$ in $G_{1}$.
Let $\hat{A} = \lbrace A_{1}, \ldots, A_p \rbrace$
with $p \geq 1$
be the set of connected components of $A$ in $G_{1}$
linked to $i$.
By Conditions~\ref{EveryEdgeOfWeightSigma2IsIncidentToTheSameEnd-Vertex2Ofe1}
and \ref{EveryEdgeOfWeightSigma3IsConnectedTo2Bye1OrByAnEdgeOfWeightSigma2},
any edge in $G_{1}$ is either incident to $1$ or $2$
or linked to $2$ by an edge in $E_{2}$.
Therefore,
$\hat{A}$
is only made up of one component containing~$2$,
and possibly singleton components
(the component containing $2$ may be reduced to a singleton).
Note that if $\hat{A}$ contains a singleton different from $2$,
then the edge $\lbrace i, 2 \rbrace$ exists in $G_{1}$.
Hence,
there necessarily is an element in $\hat{A}$
containing $2$.
Let us assume $2 \in A_{1}$ after renumbering if necessary.
Then,
as $A_{2}, \ldots, A_{p}$ are singletons,
we get
$v^{M}(A \cup \lbrace i \rbrace) - v^{M}(A)=
v(\bigcup_{j=1}^{p} A_j \cup \lbrace i \rbrace)
-  v(A_1)$.
By Condition~\ref{EveryCycleWhichDoesNotContainse1IsComplete},
$G_1$ is cycle-complete.
Then,
$(N, v^M)$ is convex by Theorem~\ref{NouwelandandBorm1991}
and this implies
\begin{equation}
\label{eqvBi-vB>=vj=1qAji-vA1}
\overline{v}(B \cup \lbrace i \rbrace) - \overline{v}(B) \geq
v\left(\bigcup_{j=1}^{p} A_j \cup \lbrace i \rbrace \right)
-  v(A_1).
\end{equation}
As $e_{1} \notin E(A)$,
if $E(A)\neq\emptyset$,
then by Condition~\ref{EveryEdgeOfWeightSigma3IsConnectedTo2Bye1OrByAnEdgeOfWeightSigma2}
any edge in $E(A) \cap E_3$ is linked to $2$ by an edge in $E(A)\cap E_2$.
Hence,
$E(A)\cap E_2\neq\emptyset$ and
$\overline{v}(A) = v^M(A)$
where $(N,v^{M})$ is associated with $\tilde G_3:= (N, E_{3})$.
If $E(A) = \emptyset$,
we trivially have $\overline{v}(A) = v^M(A)$.
As $e_{1} \notin E(A \cup \lbrace i \rbrace)$,
we have by the same reasoning
$\overline{v}(A \cup \lbrace i \rbrace) = v^M(A \cup \lbrace i \rbrace)$
where $(N, v^{M})$ is associated with $\tilde G_3$.
Let $\tilde{A}$ be
the set of connected components of $A$ in $\tilde{G}_{3}$
linked to~$i$.
Note that,
by Conditions~\ref{EveryEdgeOfWeightSigma2IsIncidentToTheSameEnd-Vertex2Ofe1}
and \ref{EveryEdgeOfWeightSigma3IsConnectedTo2Bye1OrByAnEdgeOfWeightSigma2},
$\lbrace 2 \rbrace$ is a singleton component in $\tilde{G}_{3}$
and cannot belong to $\tilde{A}$.
If $i$ is linked to $A_1 \setminus \lbrace 2 \rbrace$,
then $\tilde{A} =
\lbrace \tilde{A}_{1,1}, \tilde{A}_{1,2}, \ldots, \tilde{A}_{1,r},
A_{2}, \ldots, A_{p} \rbrace$
with $r \geq 1$
and $\emptyset \not= \tilde{A}_{1,j} \subset A_{1}$ for all $j$, $1\leq j \leq r$.
Let us assume $r \geq 2$.
There exists $k_{1} \in \tilde{A}_{1,1}$
(resp. $k_{2} \in \tilde{A}_{1,2}$)
such that $\lbrace i, k_{1} \rbrace \in E$
(resp. $\lbrace i, k_{2} \rbrace \in E$).
As $\tilde{A}_{1,1} \subseteq A_{1}$
and $\tilde{A}_{1,2} \subseteq A_{1}$,
there is a path $\gamma$ in $A_{1}$ linking $k_1$ to $k_2$.
Then,
$\lbrace i,k_{1} \rbrace \cup \gamma \cup \lbrace k_{2}, i \rbrace$
induces a cycle $C$ in $G_{1}$
as represented in Figure~\ref{figCase3-2}.
\begin{figure}[!h]
\centering
\begin{pspicture}(0,-.3)(0,1.9)
\tiny
\begin{psmatrix}[mnode=circle,colsep=0.4,rowsep=0.1]
& {$k_1$}	&	  &  {} \\
{$i$}	& & 	&	& 	  {}\\
& {$k_2$}	& 	& {}
\psset{arrows=-, shortput=nab,labelsep={0.05}}
\tiny
\ncline{1,2}{1,4}^{$\sigma_{3}$}
\ncline{2,1}{1,2}^{$\sigma_{3}$}
\ncline{2,1}{3,2}_{$\sigma_{3}$}
\ncline{3,2}{3,4}_{$\sigma_{3}$}
\ncline{3,4}{2,5}_{$\sigma_{3}$}
\ncline{2,5}{1,4}_{$\sigma_{3}$}
\end{psmatrix}
\normalsize
\pspolygon[framearc=1,linestyle=dashed,linecolor=blue,linearc=.3]
(-2.5,-.5)(-2.5,1.8)(.3,1.8)(.3,-.5)
\uput[0](.2,1){\textcolor{blue}{$A_{1}$}}
\pspolygon[framearc=1,linestyle=dashed,linecolor=blue,linearc=.3]
(-2.4,.9)(-2.4,1.7)(-.6,1.7)(-.6,.9)
\uput[0](-.7,1.5){\textcolor{blue}{$\tilde{A}_{1,1}$}}
\pspolygon[framearc=1,linestyle=dashed,linecolor=blue,linearc=.3]
(-2.4,-.4)(-2.4,.4)(-.6,.4)(-.6,-.4)
\uput[0](-.7,-.2){\textcolor{blue}{$\tilde{A}_{1,2}$}}
\end{pspicture}
\caption{Cycle $C$.}
\label{figCase3-2}
\end{figure}
By Condition~\ref{EveryCycleWhichDoesNotContainse1IsComplete},
$C$ is complete in $G$.
As $1 \notin A$ and $2 \notin \tilde{A}_{1,1}$ and $2 \notin \tilde{A}_{1,2}$,
$\lbrace k_{1}, k_{2}\rbrace \in E_{3}$ 
and links $\tilde{A}_{1,1}$ to $\tilde{A}_{1,2}$ in $\tilde{G}_{3}$,
a contradiction.
Hence,
$r=1$ and we have
\begin{equation}
\label{eqv(Ai)-v(A)=v(A_11j=2qAji)-v(A_11)}
\overline{v}(A \cup \lbrace i \rbrace) - \overline{v}(A)=
v\left(\tilde{A}_{1,1} \cup \bigcup_{j=2}^{p} A_j \cup \lbrace i \rbrace \right)
- v(\tilde{A}_{1,1}).
\end{equation}
If $i$ is not linked to $A_1 \setminus \lbrace 2 \rbrace$,
then
(\ref{eqv(Ai)-v(A)=v(A_11j=2qAji)-v(A_11)}) is still satisfied
setting $\tilde{A}_{1,1} = \emptyset$.
As $\tilde{A}_{1,1} \subseteq A_{1}$,
we have
$( \tilde{A}_{1,1} \cup \bigcup_{j=2}^{p} A_j \cup \lbrace i \rbrace)
\cap A_1 = \tilde{A}_{1,1}$
and
$(\tilde{A}_{1,1} \cup \bigcup_{j=2}^{p} A_j \cup \lbrace i \rbrace)
\cup A_{1} =\bigcup_{j=1}^{p} A_j \cup \lbrace i \rbrace$.
Therefore,
the convexity of $(N,v)$ implies
\begin{equation}
\label{eqvj=1qAji-vA1>=vA11j=2qAji-vA11}
v \left(\bigcup_{j=1}^{p} A_j \cup \lbrace i \rbrace\right) - v(A_1) \geq
v\left(
\tilde{A}_{1,1} \cup \bigcup_{j=2}^{p} A_j \cup \lbrace i \rbrace
\right)
- v(\tilde{A}_{1,1}).
\end{equation}
Finally,
(\ref{eqvBi-vB>=vj=1qAji-vA1}), (\ref{eqvj=1qAji-vA1>=vA11j=2qAji-vA11}),
and (\ref{eqv(Ai)-v(A)=v(A_11j=2qAji)-v(A_11)})
imply
(\ref{eqProofPropInheritanceOfConvexityForPminOnGIffCofGWithWeightw2EitherCompleteOrAllVerticesLinkedToTheSameEndVertexOfe1vB'Ui-vB'>=vA'Ui-vA'1}).\\

\noindent
\textbf{Case 4}
Let us assume
$2\in A$, $1 \notin B$, and $i \not= 1$.
Then,
$e_{1} \notin E(A)$,
$e_1 \notin E(A \cup \lbrace i \rbrace)$,
$e_{1} \notin E(B)$,
and $e_1 \notin E(B \cup \lbrace i \rbrace)$.
By the same reasoning as in Case~2,
we have
$\overline{v}(A) = v^M(A)$,
$\overline{v}(A \cup \lbrace i \rbrace) = v^M(A \cup \lbrace i \rbrace)$,
$\overline{v}(B) = v^M(B)$,
and
$\overline{v}(B \cup \lbrace i \rbrace) = v^M(B \cup \lbrace i \rbrace)$
where $(N,v^{M})$ is associated with $\tilde G_3:= (N, E_{3})$.
Then,
(\ref{eqProofPropInheritanceOfConvexityForPminOnGIffCofGWithWeightw2EitherCompleteOrAllVerticesLinkedToTheSameEndVertexOfe1vB'Ui-vB'>=vA'Ui-vA'1})
is equivalent to
$v^M(B \cup \lbrace i \rbrace) - v^M(B) \geq
v^M(A \cup \lbrace i \rbrace) - v^M(A)$.
Let $C_m$ be a cycle in $\tilde G_3$.
Then,
$E(C_m) \subseteq E_{3}$ and
by Condition~\ref{EveryCycleWhichDoesNotContainse1IsComplete},
$C_m$ is a complete cycle in $G$.
By Conditions~\ref{EveryEdgeOfWeightSigma2IsIncidentToTheSameEnd-Vertex2Ofe1}
and~\ref{EveryEdgeOfWeightSigma3IsConnectedTo2Bye1OrByAnEdgeOfWeightSigma2},
an edge is in $E_{2}$ if and only if it is incident to $2$.
Therefore,
$2 \notin V(C_m)$
and any chord of $C_{m}$ is in $E_{3}$.
Then,
$C_m$ is also a complete cycle in $\tilde G_3$.
Hence,
$\tilde G_3$ is cycle-complete.
Then,
$(N, v^M)$ is convex by Theorem~\ref{NouwelandandBorm1991}
and
(\ref{eqProofPropInheritanceOfConvexityForPminOnGIffCofGWithWeightw2EitherCompleteOrAllVerticesLinkedToTheSameEndVertexOfe1vB'Ui-vB'>=vA'Ui-vA'1}) is satisfied.\\

\noindent
\textbf{Case 5}
Let us assume $2 \in A$ and $i=1$.
Then,
$e_1 \notin E(A)$, $e_1 \notin E(B)$
but $e_1 \in E(A \cup \lbrace i \rbrace)$
and $e_1 \in E(B \cup \lbrace i \rbrace)$.
Let $\hat{A} = \lbrace A_{1}, \ldots, A_{p} \rbrace$
(resp. $\hat{B} = \lbrace B_{1}, \ldots, B_{q} \rbrace$)
be the set of connected components of $A$
(resp. $B$)
in $G_1 = ( N, E\setminus \lbrace e_1 \rbrace)$.
By Conditions~\ref{EveryEdgeOfWeightSigma2IsIncidentToTheSameEnd-Vertex2Ofe1}
and \ref{EveryEdgeOfWeightSigma3IsConnectedTo2Bye1OrByAnEdgeOfWeightSigma2},
any edge in $G_{1}$ is either incident to $1$ or $2$
or linked to $2$ by an edge in $E_{2}$.
Therefore,
$\hat{A}$
(resp. $\hat{B}$)
is only made up of one component containing~$2$,
and possibly singleton components
(the component containing $2$ may be reduced to a singleton).
Let $A_{1}$
(resp. $B_{1}$)
be the component containing $2$.
Then,
we have $A_{1} \subseteq B_{1}$
and $A_{2}, \ldots, A_{p}$
(resp. $B_{2}, \ldots, B_{q}$)
are singletons.
Let $\tilde{A}$
(resp. $\tilde{B}$)
be the set of connected components of $A$
(resp. $B$)
in $\tilde G_3:= (N, E_{3})$.
By Conditions~\ref{EveryEdgeOfWeightSigma2IsIncidentToTheSameEnd-Vertex2Ofe1}
and \ref{EveryEdgeOfWeightSigma3IsConnectedTo2Bye1OrByAnEdgeOfWeightSigma2},
$\lbrace 2 \rbrace$ is a singleton component in $\tilde{G}_{3}$.
Therefore,
we have
$\tilde{A} = \lbrace \tilde{A}_{1,1}, \tilde{A}_{1,2}, \ldots,
\tilde{A}_{1,r}, A_{2}, \ldots, A_{p} \rbrace$
(resp. 
$\tilde{B} = \lbrace \tilde{B}_{1,1}, \tilde{B}_{1,2}, \ldots,
\tilde{B}_{1,s}, B_{2}, \ldots, B_{q} \rbrace$)
where
$\lbrace \tilde{A}_{1,1}, \tilde{A}_{1,2}, \ldots, \tilde{A}_{1,r} \rbrace$
(resp. $\lbrace \tilde{B}_{1,1}, \tilde{B}_{1,2}, \ldots, \tilde{B}_{1,s} \rbrace$)
with $r \geq 1$
(resp. $s \geq 1$)
is the partition of $A_{1}$
(resp. $B_{1}$)
in $\tilde{G}_{3}$
and $\tilde{A}_{1,1} = \tilde{B}_{1,1} = \lbrace 2 \rbrace$.
Note that
$\tilde{A}_{1,j}$
(resp. $\tilde{B}_{1,j}$)
is linked to $2$ in $G$ (and $G_{1}$)
\emph{i.e.},
there exists $k_j \in \tilde{A}_{1,j}$
(resp. $l_j \in \tilde{B}_{1,j}$)
such that $\lbrace 2, k_j \rbrace \in E$
(resp. $\lbrace 2, l_j \rbrace \in E$)
for all $j$, $2 \leq j \leq r$
(resp. $2 \leq j \leq s$).

\begin{claim}
\label{Claim1-2}
We can assume
$\tilde{A}_{1,j} \subseteq \tilde{B}_{1,j}$ for all $j$, $1 \leq j \leq r$,
after renumbering if necessary.
\end{claim}
\begin{proof}[Proof of Claim~\ref{Claim1-2}]
We have $\tilde{A}_{1,1}=\tilde{B}_{1,1} = \lbrace 2 \rbrace$.
By contradiction,
let us assume $\tilde{A}_{1,2} \subseteq \tilde{B}_{1,2}$
and $\tilde{A}_{1,3} \subseteq \tilde{B}_{1,2}$,
after renumbering if necessary.
Let $\gamma$ be a simple path in $\tilde{B}_{1,2}$
linking $k_2 \in \tilde{A}_{1,2}$ to $k_3 \in \tilde{A}_{1,3}$.
Then,
$\lbrace 2, k_{2} \rbrace \cup \gamma \cup \lbrace k_3, 2 \rbrace$ 
induces a cycle $C$ in $G_1$
as represented in Figure~\ref{figCase6-1-2}.
\begin{figure}[!h]
\centering
\begin{pspicture}(0,-.3)(0,2)
\tiny
\begin{psmatrix}[mnode=circle,colsep=0.4,rowsep=0.1]
& {}	&	  &  {$k_{2}$} \\
{}	& & 	&	& 	  {$2$}\\
& {}	& 	& {$k_{3}$}
\psset{arrows=-, shortput=nab,labelsep={0.05}}
\tiny
\ncline{1,2}{1,4}^{$\sigma_{3}$}
\ncline{2,1}{1,2}^{$\sigma_{3}$}
\ncline{2,1}{3,2}_{$\sigma_{3}$}
\ncline{3,2}{3,4}_{$\sigma_{3}$}
\ncline{3,4}{2,5}_{$\sigma_{2}$}
\ncline{2,5}{1,4}_{$\sigma_{2}$}
\end{psmatrix}
\normalsize
\pspolygon[framearc=1,linestyle=dashed,linecolor=blue,linearc=.3]
(-2.6,.9)(-2.6,1.7)(-.8,1.7)(-.8,.9)
\uput[0](-3.5,1.5){\textcolor{blue}{$\tilde{A}_{1,2}$}}
\pspolygon[framearc=1,linestyle=dashed,linecolor=blue,linearc=.3]
(-2.6,-.4)(-2.6,.4)(-.8,.4)(-.8,-.4)
\uput[0](-3.5,-.2){\textcolor{blue}{$\tilde{A}_{1,3}$}}
\pspolygon[framearc=1,linestyle=dashed,linecolor=blue,linearc=.4]
(-3.5,-.5)(-3.5,1.8)(-.7,1.8)(-.7,-.5)
\uput[0](-4.5,.65){\textcolor{blue}{$\tilde{B}_{1,2}$}}
\end{pspicture}
\caption{Cycle $C$.}
\label{figCase6-1-2}
\end{figure}
By Condition~\ref{EveryCycleWhichDoesNotContainse1IsComplete},
$C$ is complete in $G$.
Then,
$\lbrace k_2, k_3 \rbrace \in E(A) \cap E_{3}$
and links $\tilde{A}_{1,2}$ to $\tilde{A}_{1,3}$ in $\tilde{G}_{3}$,
a contradiction.
\end{proof}

The partition of $A_{1}$
(resp. $B_{1}$)
in $\tilde{G}_{3}$
is $\mathcal{P}_{M}(A_{1})=
\lbrace \lbrace 2 \rbrace, \tilde{A}_{1,2}, \ldots, \tilde{A}_{1,r} \rbrace$
(resp. $\mathcal{P}_{M}(B_{1})=
\lbrace \lbrace 2 \rbrace, \tilde{B}_{1,2}, \ldots, \tilde{B}_{1,s} \rbrace$).
Claim~\ref{Claim1-2} implies
$\mathcal{P}_{M}(B_{1})_{|A_{1}}=\mathcal{P}_{M}(A_{1})$.
Then,
as $(N,v)$ is convex,
Lemma~\ref{lemPNABNPAPBAvNvB-vB>=vA-vA} implies
\begin{equation}
\label{v(B_1)-sum_{j=1}^{s}v(tilde{B}_{1,j})geqv(A_1)-sum_{j=1}^{r}v(tilde{A}_{1,j})}
v(B_1) - \sum_{j=1}^{s}v(\tilde{B}_{1,j}) \geq
v(A_1) - \sum_{j=1}^{r}v(\tilde{A}_{1,j}).
\end{equation}
As $e_{1} \notin E(A)$
(resp. $e_{1} \notin E(B)$)
we have $\overline{v}(A) = v^{M}(A)$
(resp. $\overline{v}(B) = v^{M}(B)$)
where $(N,v^{M})$ is associated with $\tilde{G}_{3}$.
We get $\overline{v}(A) = \sum_{j=1}^{r} v(\tilde{A}_{1,j})$
(resp. $\overline{v}(B) = \sum_{j=1}^{s} v(\tilde{B}_{1,j})$).
As $e_{1} \in E(A \cup \lbrace i \rbrace)$
(resp. $e_{1} \in E(B \cup \lbrace i \rbrace)$),
we have $\overline{v}(A \cup \lbrace i \rbrace) = v^{M}(A \cup \lbrace i \rbrace)$
(resp. $\overline{v}(B \cup \lbrace i \rbrace) = v^{M}(B \cup \lbrace i \rbrace)$)
where $(N, v^{M})$ is now associated with $G_{1}$.\\

\noindent
\textbf{Case 5.1}
Let us first assume $i$ is not linked to $A$ in $G_{1}$.
Then,
$\overline{v}(A \cup \lbrace i \rbrace) =
v(\lbrace i \rbrace) + \sum_{j=1}^{p} v(A_j) = v(A_1)$.
If $i$ is not linked to $B$ in $G_{1}$,
then we also have
$\overline{v}(B \cup \lbrace i \rbrace) = v(B_1)$
and
(\ref{eqProofPropInheritanceOfConvexityForPminOnGIffCofGWithWeightw2EitherCompleteOrAllVerticesLinkedToTheSameEndVertexOfe1vB'Ui-vB'>=vA'Ui-vA'1})
is satisfied
as it is equivalent to (\ref{v(B_1)-sum_{j=1}^{s}v(tilde{B}_{1,j})geqv(A_1)-sum_{j=1}^{r}v(tilde{A}_{1,j})}).
If $i$ is linked to $B$ in $G_{1}$,
let $\hat{B}'$
be the set of connected components of $B$ linked to $i$ in $G_{1}$.
We have either
$\hat{B}'=\lbrace B_{1}, \ldots, B_{q'} \rbrace$
with $1 \leq q' \leq q$
or $\hat{B}'=\lbrace B_{2}, \ldots, B_{q'} \rbrace$
with $2 \leq q' \leq q$,
after renumbering if necessary
and therefore,
either
$\overline{v}(B \cup \lbrace i \rbrace) - \overline{v}(B)=
v(\bigcup_{j=1}^{q'} B_{j} \cup \lbrace i \rbrace) - \sum_{j=1}^{s}v(\tilde{B}_{1,j})$
or
$\overline{v}(B \cup \lbrace i \rbrace) - \overline{v}(B)=
v(B_{1})+
v(\bigcup_{j=2}^{q'} B_{j} \cup \lbrace i \rbrace) - \sum_{j=1}^{s}v(\tilde{B}_{1,j})$.
As $(N,v)$ is superadditive,
we have in any case
\begin{equation}
\label{eqoverline{v}(Bcupi)-overline{v}(B)geqv(B_{1})+v(cup_{j=2}^{q'}B_{j}cupi)-sum_{j=1}^{s}v(tilde{B}_{1,j})}
\overline{v}(B \cup \lbrace i \rbrace) - \overline{v}(B) \geq
v(B_{1})+
v\left(\bigcup_{j=2}^{q'} B_{j} \cup \lbrace i \rbrace\right)
- \sum_{j=1}^{s}v(\tilde{B}_{1,j}).
\end{equation}
As $(N,v)$ is superadditive and zero-normalized,
we have $v\left(\bigcup_{j=2}^{q'} B_{j} \cup \lbrace i \rbrace\right) \geq 0$.
Then,
(\ref{eqoverline{v}(Bcupi)-overline{v}(B)geqv(B_{1})+v(cup_{j=2}^{q'}B_{j}cupi)-sum_{j=1}^{s}v(tilde{B}_{1,j})})
and (\ref{v(B_1)-sum_{j=1}^{s}v(tilde{B}_{1,j})geqv(A_1)-sum_{j=1}^{r}v(tilde{A}_{1,j})})
imply (\ref{eqProofPropInheritanceOfConvexityForPminOnGIffCofGWithWeightw2EitherCompleteOrAllVerticesLinkedToTheSameEndVertexOfe1vB'Ui-vB'>=vA'Ui-vA'1}).\\

\noindent
\textbf{Case 5.2}
Let us now assume $i$ is linked to $A$ in $G_{1}$.
Let $\hat{A}'$ be the set of connected components of $A$ linked to $i$ in $G_{1}$.
We have either
$\hat{A}'=\lbrace A_{1}, \ldots, A_{p'} \rbrace$
with $1 \leq p' \leq p$
or $\hat{A}'=\lbrace A_{2}, \ldots, A_{p'} \rbrace$
with $2 \leq p' \leq p$,
after renumbering if necessary.

\begin{claim}
\label{Claim5}
There is at most one element of $\hat{A}'$ included in $B_{1}$.
\end{claim}
\begin{proof}[Proof of Claim~\ref{Claim5}]
Let us assume w.l.o.g. $A_{2} \subseteq B_{1}$
and  $A_{3} \subseteq B_{1}$.
Let $\tilde{e}_2 = \lbrace i, k_{2} \rbrace$
(resp. $\tilde{e}_3 = \lbrace i, k_{3} \rbrace$)
be an edge linking $i$ to $A_{2}$
(resp. $A_{3}$) in $G_{1}$.
As $B_{1}$ is connected
there is a path $P$ connecting $k_{2}$ to $k_{3}$.
Then,
$\tilde{e}_2$,
$\tilde{e}_3$,
and $P$ induce a cycle $C$ in $G_{1}$.
By Condition~\ref{EveryCycleWhichDoesNotContainse1IsComplete},
$C$ is complete in $G$. 
Then,
$\lbrace k_{2}, k_{3}\rbrace$
links $A_{2}$ and $A_{3}$ in $G_{1}$,
a contradiction.
\end{proof}

Let us first assume $A_1 \in \hat{A}'$.
Then,
$B_{1} \in \hat{B}'$,
and therefore
(\ref{eqProofPropInheritanceOfConvexityForPminOnGIffCofGWithWeightw2EitherCompleteOrAllVerticesLinkedToTheSameEndVertexOfe1vB'Ui-vB'>=vA'Ui-vA'1})
is equivalent to
\begin{equation}
\label{eqv(cup_{j=1}^{q'}B_jcupi)-sum_{j=1}^{s}v(tilde{B}_{1,j})>=v(cup_{j=1}^{p'}A_jcupi)-sum_{j=1}^{r}v(tilde{A}_{1,j})}
v\left(\bigcup_{j=1}^{q'} B_j \cup \lbrace i \rbrace\right) - \sum_{j=1}^{s}v(\tilde{B}_{1,j}) \geq
v\left(\bigcup_{j=1}^{p'} A_j \cup \lbrace i \rbrace\right) - \sum_{j=1}^{r}v(\tilde{A}_{1,j}).
\end{equation}
The partition of $\bigcup_{j=1}^{p'} A_{j}$
(resp. $\bigcup_{j=1}^{q'} B_{j}$)
in $\tilde{G}_{3}$
is $\mathcal{P}_{M}(\bigcup_{j=1}^{p'} A_{j})=
\lbrace \lbrace 2 \rbrace, \tilde{A}_{1,2},\\ \ldots, \tilde{A}_{1,r}, A_{2}, \ldots, A_{p'} \rbrace$
(resp. $\mathcal{P}_{M}(\bigcup_{j=1}^{q'} B_{j})=
\lbrace \lbrace 2 \rbrace, \tilde{B}_{1,2}, \ldots, \tilde{B}_{1,s}, B_{2}, \ldots,\\ B_{q'} \rbrace$).
By Claim~\ref{Claim1-2},
we have
$\tilde{A}_{1,j} \subseteq \tilde{B}_{1,j}$ for all $j$, $1 \leq j \leq r$.
As $A_{1} \subseteq B_{1}$,
Claim~\ref{Claim5} implies $A_{j} \not \subseteq \tilde{B}_{1,k}$
for all $j$,
$2 \leq j \leq p'$,
and all $k$,
$2 \leq k \leq s$.
Therefore,
we have
$\mathcal{P}_{M}(\bigcup_{j=1}^{q'} B_{j})_{|\bigcup_{j=1}^{p'} A_{j}}=\mathcal{P}_{M}(\bigcup_{j=1}^{p'} A_{j})$,
and as $(N,v)$ is convex
Lemma~\ref{lemPNABNPAPBAvNvB-vB>=vA-vA} implies
\begin{equation}
\label{eqv(cup_{j=1}^{q'} B_{j})-sum_{j=1}^{s}v(tilde{B}_{1,j})>=v(cup_{j=1}^{p'}A_{j})-sum_{j=1}^{r}v(tilde{A}_{1,j})}
v\left(\bigcup_{j=1}^{q'} B_{j}\right) - \sum_{j=1}^{s}v(\tilde{B}_{1,j}) \geq
v\left(\bigcup_{j=1}^{p'} A_{j}\right) - \sum_{j=1}^{r}v(\tilde{A}_{1,j}).
\end{equation}
Moreover,
the convexity of $(N,v)$ implies
\begin{equation}
\label{eqv(cup_{j=1}^{q'}B_{j}cuplbraceirbrace)-v(cup_{j=1}^{q'}B_{j})>=v(cup_{j=1}^{p'}A_{j}cuplbraceirbrace)-v(cup_{j=1}^{p'}A_{j})}
v\left(\bigcup_{j=1}^{q'} B_{j} \cup \lbrace i \rbrace \right) - v\left(\bigcup_{j=1}^{q'} B_{j}\right) \geq
v\left(\bigcup_{j=1}^{p'} A_{j} \cup \lbrace i \rbrace \right) - v\left(\bigcup_{j=1}^{p'} A_{j}\right).
\end{equation}
(\ref{eqv(cup_{j=1}^{q'} B_{j})-sum_{j=1}^{s}v(tilde{B}_{1,j})>=v(cup_{j=1}^{p'}A_{j})-sum_{j=1}^{r}v(tilde{A}_{1,j})})
and
(\ref{eqv(cup_{j=1}^{q'}B_{j}cuplbraceirbrace)-v(cup_{j=1}^{q'}B_{j})>=v(cup_{j=1}^{p'}A_{j}cuplbraceirbrace)-v(cup_{j=1}^{p'}A_{j})}) imply (\ref{eqv(cup_{j=1}^{q'}B_jcupi)-sum_{j=1}^{s}v(tilde{B}_{1,j})>=v(cup_{j=1}^{p'}A_jcupi)-sum_{j=1}^{r}v(tilde{A}_{1,j})}).
Let us now assume $A_{1} \notin \hat{A}'$.
If $B_{1} \notin \hat{B}'$,
then (\ref{eqProofPropInheritanceOfConvexityForPminOnGIffCofGWithWeightw2EitherCompleteOrAllVerticesLinkedToTheSameEndVertexOfe1vB'Ui-vB'>=vA'Ui-vA'1})
is equivalent to
\begin{equation}
\label{v(B_{1})+v(cup_{j=2}^{q'}B_jcupi)-sum_{j=1}^{s}v(tilde{B}_{1,j})>=v(A_{1})+v(cup_{j=2}^{p'}A_jcupi)-sum_{j=1}^{r}v(tilde{A}_{1,j})}
v(B_{1})+v\left(\bigcup_{j=2}^{q'} B_j \cup \lbrace i \rbrace\right)
-\sum_{j=1}^{s}v(\tilde{B}_{1,j}) \geq
v(A_{1})+v\left(\bigcup_{j=2}^{p'} A_j \cup \lbrace i \rbrace\right)
-\sum_{j=1}^{r}v(\tilde{A}_{1,j}).
\end{equation}
We have $A_{j} \not \subseteq B_{1}$ for all $j$, $2 \leq j \leq p'$
(otherwise $B_{1}$ is linked to $i$,
a contradiction).
Hence,
$\bigcup_{j=2}^{p'} A_j \subseteq \bigcup_{j=2}^{q'} B_j$
and the superadditivity of $(N,v)$ implies 
$v(\bigcup_{j=2}^{q'} B_j \cup \lbrace i \rbrace) \geq
v(\bigcup_{j=2}^{p'} A_j \cup \lbrace i \rbrace)$.
This last inequality
and (\ref{v(B_1)-sum_{j=1}^{s}v(tilde{B}_{1,j})geqv(A_1)-sum_{j=1}^{r}v(tilde{A}_{1,j})})
imply (\ref{v(B_{1})+v(cup_{j=2}^{q'}B_jcupi)-sum_{j=1}^{s}v(tilde{B}_{1,j})>=v(A_{1})+v(cup_{j=2}^{p'}A_jcupi)-sum_{j=1}^{r}v(tilde{A}_{1,j})}).
Finally,
if $B_{1} \in \hat{B}'$,
then (\ref{eqProofPropInheritanceOfConvexityForPminOnGIffCofGWithWeightw2EitherCompleteOrAllVerticesLinkedToTheSameEndVertexOfe1vB'Ui-vB'>=vA'Ui-vA'1})
is equivalent to
\begin{equation}
\label{eqv(cup_{j=1}^{q'}B_jcupi)-sum_{j=1}^{s}v(tilde{B}_{1,j})>=v(A_{1})+v(cup_{j=2}^{p'}A_jcupi)-sum_{j=1}^{r}v(tilde{A}_{1,j})}
v\left(\bigcup_{j=1}^{q'} B_j \cup \lbrace i \rbrace\right)
-\sum_{j=1}^{s}v(\tilde{B}_{1,j}) \geq
v(A_{1})+v\left(\bigcup_{j=2}^{p'} A_j \cup \lbrace i \rbrace\right)
-\sum_{j=1}^{r}v(\tilde{A}_{1,j}).
\end{equation}
If $\bigcup_{j=2}^{p'} A_j \subseteq \bigcup_{j=2}^{q'} B_j$,
then (\ref{v(B_{1})+v(cup_{j=2}^{q'}B_jcupi)-sum_{j=1}^{s}v(tilde{B}_{1,j})>=v(A_{1})+v(cup_{j=2}^{p'}A_jcupi)-sum_{j=1}^{r}v(tilde{A}_{1,j})})
is satisfied
and it implies (\ref{eqv(cup_{j=1}^{q'}B_jcupi)-sum_{j=1}^{s}v(tilde{B}_{1,j})>=v(A_{1})+v(cup_{j=2}^{p'}A_jcupi)-sum_{j=1}^{r}v(tilde{A}_{1,j})})
as $(N,v)$ is superadditive.
Otherwise,
we can assume w.l.o.g. $A_{2} \subseteq B_{1}$.
Claim~\ref{Claim5} implies
$\bigcup_{j=3}^{p'} A_{j} \subseteq \bigcup_{j=2}^{q'} B_{j}$.
Then,
we have
$(\bigcup_{j=2}^{p'} A_j \cup \lbrace i \rbrace)\cap B_{1} = A_{2}$
and the convexity of $(N,v)$ implies
\begin{equation}
\label{eqv(cup_{j=2}^{p'}A_jcupicupB_{1})+v(A_{2})>=v(cup_{j=2}^{p'}A_jcupi)+v(B_{1})}
v\left(\bigcup_{j=2}^{p'} A_j \cup \lbrace i \rbrace \cup B_{1}\right)
+v(A_{2})
\geq v\left(\bigcup_{j=2}^{p'} A_j \cup \lbrace i \rbrace\right)
+v(B_{1}).
\end{equation}
Moreover,
we have
$\bigcup_{j=2}^{p'} A_j \cup B_{1} \subseteq \bigcup_{j=1}^{q'} B_{j}$
and $A_{2}$ is a singleton.
As $(N,v)$ is superadditive and zero-normalized,
(\ref{eqv(cup_{j=2}^{p'}A_jcupicupB_{1})+v(A_{2})>=v(cup_{j=2}^{p'}A_jcupi)+v(B_{1})}) implies
$v(\bigcup_{j=1}^{q'} B_{j} \cup \lbrace i \rbrace)
\geq v(\bigcup_{j=2}^{p'} A_j \cup \lbrace i \rbrace)
+v(B_{1})$,
and therefore
\begin{equation}
\label{eqv(cup_{j=1}^{q'}B_{j}cupi)-sum_{j=1}^{s}v(tilde{B}_{1,j})>=v(cup_{j=2}^{p'}A_jcupi)+v(B_{1})-sum_{j=1}^{s}v(tilde{B}_{1,j})}
v\left(\bigcup_{j=1}^{q'} B_{j} \cup \lbrace i \rbrace \right)
-\sum_{j=1}^{s}v(\tilde{B}_{1,j})
\geq v\left(\bigcup_{j=2}^{p'} A_j \cup \lbrace i \rbrace\right)
+v(B_{1})
-\sum_{j=1}^{s}v(\tilde{B}_{1,j}).
\end{equation}
Finally,
(\ref{eqv(cup_{j=1}^{q'}B_{j}cupi)-sum_{j=1}^{s}v(tilde{B}_{1,j})>=v(cup_{j=2}^{p'}A_jcupi)+v(B_{1})-sum_{j=1}^{s}v(tilde{B}_{1,j})})
and
(\ref{v(B_1)-sum_{j=1}^{s}v(tilde{B}_{1,j})geqv(A_1)-sum_{j=1}^{r}v(tilde{A}_{1,j})})
imply
(\ref{eqv(cup_{j=1}^{q'}B_jcupi)-sum_{j=1}^{s}v(tilde{B}_{1,j})>=v(A_{1})+v(cup_{j=2}^{p'}A_jcupi)-sum_{j=1}^{r}v(tilde{A}_{1,j})}).\\

We now prove
that Conditions~\ref{ThereExistsOnlyOneEdgeInE1} to~\ref{EveryCycleWhichDoesNotContainse1IsComplete}
imply
Condition~\ref{itemThIfACycleCmDoesNotContaine1AndIf1inV(Cm)Thenm=3}.
By Conditions~\ref{ThereExistsOnlyOneEdgeInE1} to~\ref{EveryEdgeOfWeightSigma3IsConnectedTo2Bye1OrByAnEdgeOfWeightSigma2},
$e_{1}$ is the unique edge in $E_{1}$,
all edges in $E_{2}$ are incident to $2$,
and all edges in $E_{3}$ are linked to $2$
by $e_{1}$
or by an edge in $E_{2}$.
By Condition~\ref{EveryCycleWhichDoesNotContainse1IsComplete},
$C_{m}$ is complete in $G_1$,
so that $e_1$ cannot be a chord of $C_m$.
Moreover,
as $1 \in V(C_m)$ and  $e_{1} \notin \hat{E}(C_m)$,
we have $2 \notin V(C_m)$.
Then,
$C_m$ has constant weight $\sigma_3$
and any chord of $C_{m}$ has weight $\sigma_{3}$.
If there exist two edges $\lbrace 2, i \rbrace$ and $\lbrace 2, j \rbrace$
with $i$ and $j$ in $V(C_{m}) \setminus \lbrace 1 \rbrace$
as represented
in Figure~\ref{FigCmWithm=5CompleteCyclea} with $m = 5$,
then
$e_{1}$ is a common edge
for the triangles defined by $\lbrace 1,2,i \rbrace$ and $\lbrace 1,2,j \rbrace$,
contradicting
Condition~\ref{ThereExistsAtMostOneChordlessCycleCmWithm=3or4Containinge1}.
Hence,
there is at most one edge linking $2$ to $V(C_{m}) \setminus \lbrace 1 \rbrace$.
This implies $m = 3$
and there is exactly one edge linking $2$ to $V(C_{m}) \setminus \lbrace 1 \rbrace$
as represented in Figure~\ref{figC3LinkedTo2b}
(otherwise there is at least one edge in $\hat{E}(C_m)$ neither incident to $1$
nor linked to $2$ by an edge in $E_2$).
\begin{figure}[!h]
\centering
\subfloat[]{
\begin{pspicture}(-.5,-.5)(.6,1.5)
\tiny
\begin{psmatrix}[mnode=circle,colsep=0.4,rowsep=0.2]
{}	&	  &  {$i$} \\
	& 	&	& 	  {$1$} & & {$2$}\\
{$j$}	& 	&  {}
\psset{arrows=-, shortput=nab,labelsep={0.05}}
\tiny
\ncline{2,4}{2,6}_{$e_{1}$}^{$\sigma_{1}$}
\ncline[linecolor=blue]{3,1}{3,3}_{$\sigma_{3}$}
\ncline[linecolor=blue]{3,3}{2,4}_{$\sigma_{3}$}
\ncline[linecolor=blue]{2,4}{1,3}_{$\sigma_{3}$}
\ncline[linecolor=blue]{1,3}{1,1}_{$\sigma_{3}$}
\ncline[linecolor=blue]{1,1}{3,1}_{$\sigma_{3}$}

\ncline{1,1}{3,3}_{$\sigma_{3}$}
\ncline{3,1}{1,3}^{$\sigma_{3}$}
\ncline{1,3}{3,3}^{$\sigma_{3}$}

\ncline{1,1}{2,4}_{$\sigma_{3}$}
\ncline{3,1}{2,4}^{$\sigma_{3}$}
\ncarc[arcangle=20]{1,3}{2,6}^{$\sigma_{2}$}
\ncarc[arcangle=-38]{3,1}{2,6}_{$\sigma_{2}$}
\end{psmatrix}
\normalsize
\uput[0](-4.7,.7){\textcolor{blue}{$C_{m}$}}
\end{pspicture}
\label{FigCmWithm=5CompleteCyclea}
}
\subfloat[]{
\begin{pspicture}(-.5,-.5)(.5,1.5)
\tiny
\begin{psmatrix}[mnode=circle,colsep=.5,rowsep=0.1]
{$3$}\\
					&	{$1$} & & {$2$}\\
{$3'$}
\psset{arrows=-, shortput=nab,labelsep={0.05}}
\tiny
\ncline[linecolor=blue]{1,1}{3,1}_{$\sigma_{3}$}
\ncline[linecolor=blue]{1,1}{2,2}^{$\sigma_{3}$}_{$e_3$}
\ncline[linecolor=blue]{3,1}{2,2}_{$\sigma_{3}$}
\ncline{2,2}{2,4}_{$e_{1}$}^{$\sigma_{1}$}
\ncarc[arcangle=30]{1,1}{2,4}^{$\sigma_{2}$}_{$e_2$}
\end{psmatrix}
\normalsize
\uput[0](-3.8,.8){\textcolor{blue}{$C_{3}$}}
\end{pspicture}
\label{figC3LinkedTo2b}
}
\caption{$C_{5}$ and $C_3$ linked to $2$ by an edge in $E_2$.}
\end{figure}
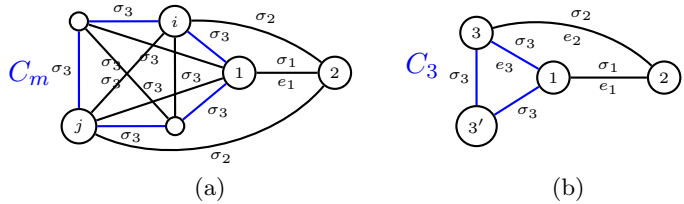
Finally,
let us assume that there is a second triangle $C_3'$ with $1 \in V(C_3')$
and $e_1 \notin E(C_3')$.
If $C_3'$ is adjacent to $C_3$,
then $C_3$ and $C_3'$ induce a cycle of size $4$ incident to $1$ and not containing $e_1$, a contradiction.
Otherwise,
there exists $e \in E_2$
(resp. $e' \in E_2$)
linking $2$ to $V(C_3) \setminus \lbrace 1 \rbrace$ 
(resp. $V(C_3') \setminus \lbrace 1 \rbrace$).
Then,
there are two triangles containing $e_1$,
contradicting Condition~\ref{ThereExistsAtMostOneChordlessCycleCmWithm=3or4Containinge1}.
\end{proof}

\begin{remark}
Let us note that
Conditions~\ref{itemLemInheritanceOfConvexityForPMOnG1=(N,E2)}
and~\ref{itemLemInheritanceOfConvexityForPMOnG1=(N,E3)}
of
Theorem~\ref{lemInheritanceOfConvexityForPminInheritanceOfConvexityForPMG1Cycle-Complete},
Conditions~\ref{itThThereExistsAtMostOneCycleCWithoutChordContaininge1}
and~\ref{itThCycleCOfGWithConstantEdge-Weightw2EitherCompleteOrAllVerticesLinkedToSameEndVertexOfe1}
of Theorem~\ref{PropInheritanceOfConvexityForPminOnGIffCofGWithWeightwEitherCompleteOrAllVerticesLinkedToTheSameEndVertexOfe1},
and
Conditions~\ref{ThereExistsOnlyOneEdgeInE1}
to~\ref{EveryCycleWhichDoesNotContainse1IsComplete}
of Theorem~\ref{PropInheritanceOfConvexityForPminOnGIffCofGWithWeightw2EitherCompleteOrAllVerticesLinkedToTheSameEndVertexOfe12}
are consequences of the necessary conditions established in Section~\ref{SectionInheritanceOfConvexityWithPmin}.
To obtain these last conditions we only needed to 
assume inheritance of convexity with $\mathcal{P}_{\min}$
for the family of unanimity games.
Therefore,
Theorems~\ref{lemInheritanceOfConvexityForPminInheritanceOfConvexityForPMG1Cycle-Complete},
\ref{PropInheritanceOfConvexityForPminOnGIffCofGWithWeightwEitherCompleteOrAllVerticesLinkedToTheSameEndVertexOfe1},
and~\ref{PropInheritanceOfConvexityForPminOnGIffCofGWithWeightw2EitherCompleteOrAllVerticesLinkedToTheSameEndVertexOfe12}
imply that
for the correspondence $\mathcal{P}_{\min}$
there is inheritance of convexity if and only if
there is
inheritance of convexity for the family of unanimity games.
This result was already observed
in~\citep{Skoda201702}.
\end{remark}

\subsection{Disconnected graphs}
We now consider the case of disconnected weighted graphs.
A connected component is said to be \emph{constant}
if all its edges have the same weight.
We prove that if there is inheritance of convexity for $\mathcal{P}_{\min}$,
then the underlying graph $G$ has to be connected or has only one component with non-constant weight. 

\begin{proposition}
\label{corPathCond (Path Condition)nonconnected}
Let $G=(N,E,w)$ be a weighted graph.
Let us assume that for all $\emptyset \not= S \subseteq N$
the $\mathcal{P}_{\min}$-restricted game $(N,\overline{u_{S}})$ is convex.
\begin{enumerate}
\item
\label{itemLemmeE12w2=w3w2component1}
If $G$ has a connected component with three different weights $\sigma_1<\sigma_2<\sigma_3$,
then all other connected components of $G$ are singletons.
\item
\label{itemLemmeE12w2=w3w2component2}
Suppose that $G$ has a connected component with two different weights $\sigma_1<\sigma_2$
and that the component has two distinct edges $e_1$ and $e_2$
with $w_1 = w_2 = \sigma_1$.
Then all other connected components of $G$ are singletons.
\item
\label{itemLemmeE12w2=w3w2component3}
If $G$ has a connected component with two different weights $\sigma_1<\sigma_2$
and if $|E_{1}|=1$,
then all other connected components of $G$ have constant weight $\sigma_2$ or are singletons.
\end{enumerate}
\end{proposition}
To prove Proposition~\ref{corPathCond (Path Condition)nonconnected} we need the following lemma.

\begin{lemma}
\label{LemP_MinMax(w1,w2)<w(e)e'InELinkingeTo21}
Let us assume that for all $\emptyset \not= S \subseteq N$
the $\mathcal{P}_{\min}$-restricted game $(N,\overline{u_{S}})$ is convex.
Let $e_{1} = \lbrace 1,2 \rbrace$ and $e_{2} = \lbrace 2,3 \rbrace$
be two adjacent edges with $w_1<w_2$ and $e$ be an edge non-incident to~$1$.
Then,
we have $w(e)\geq w_2$ and if moreover $e$ is not linked to~$2$ by an edge,
then $w(e)=w_2$.
\end{lemma}

\begin{proof}
We set $e = \lbrace j, k \rbrace$.
If $e$ is incident to 2,
Star condition implies $w(e)=w_2$.
If $e$ is incident to 3,
Path condition implies $w(e)\ge w_2$.
Hence,
we can assume $e$ non-incident to the vertices $1$, $2$, and $3$.
By contradiction,
let us assume $w(e)<w_{2}$.
Let us consider $i=1$ and the subsets
$A= \lbrace 2,3 \rbrace$ and $B = A \cup \lbrace j,k \rbrace$,
as represented in Figure~\ref{figLemmaProofwMax(w1,w2)<w(e)-2}.
\begin{figure}[!h]
\centering
\begin{pspicture}(0,-.3)(0,1.6)
\tiny
\begin{psmatrix}[mnode=circle,colsep=.7,rowsep=.4]
{$j$}	& {$k$}\\
{$1$} 	& {$2$}	& {$3$}
\psset{arrows=-, shortput=nab,labelsep={.1}}
\tiny
\ncline{1,1}{1,2}^{$e$}
\ncline[linestyle=dotted]{1,1}{2,1}
\ncline[linestyle=dotted]{1,2}{2,1}
\ncline[linestyle=dotted]{1,1}{2,3}
\ncline[linestyle=dotted]{1,2}{2,3}
\ncline{2,1}{2,2}_{$e_{1}$}
\ncline{2,2}{2,3}_{$e_{2}$}
\end{psmatrix}
\normalsize
\uput[0](-2.2,.3){\textcolor{cyan}{$A$}}
\pspolygon[linecolor=cyan,linearc=.4,linewidth=.02](-1.7,-.4)(-1.7,.4)(.1,.4)(.1,-.4)
\uput[0](-3.35,0){\textcolor{blue}{$\scriptstyle{i=}$}}
\uput[0](-3.5,1){\textcolor{cyan}{$B$}}
\pspolygon[linecolor=cyan,linearc=.5,linewidth=.02](-1.5,-.5)(-3.9,1.3)(-1.1,1.3)(1.1,-.5)
\end{pspicture}
\caption{$w_{1} < w_{2}$ and $w(e) < w_{2}$.}
\label{figLemmaProofwMax(w1,w2)<w(e)-2}
\end{figure}
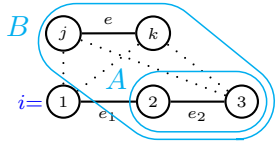
As $w(e) < w_2$,
there is a block $B'$ of $\mathcal{P}_{\min}(B)$ such that $A\subseteq B'$.
As $w_{1} < w_2$,
we have $\sigma(A\cup\lbrace i\rbrace)<w_2$
and
$\mathcal{P}_{\min}(A\cup\lbrace i\rbrace)=
\lbrace A\cup\lbrace i\rbrace\rbrace$ or $\lbrace A,\lbrace i\rbrace\rbrace$.
Then,
taking $A'=A$
we get $\mathcal{P}_{\min}(B)_{|A'}=\lbrace A\rbrace \not=
\lbrace \lbrace 2 \rbrace, \lbrace 3 \rbrace\rbrace = \mathcal{P}_{\min}(A)_{|A'}$
and it contradicts Theorem~\ref{thLGNuPNFNFSNuSSNNuSFABFABF}
applied with $\mathcal{F} = 2^N \setminus \lbrace \emptyset \rbrace$.
If $e$ is not linked to 2 by an edge,
Lemma~\ref{LemP_MinMax(w1,w2)<w(e)e'InELinkingeTo2} implies $w(e)\le \max(w_1,w_2)=w_2$ and therefore $w(e)=w_2$.
\end{proof} 
We can now prove Proposition~\ref{corPathCond (Path Condition)nonconnected}.
\begin{proof}
[Proof of Proposition~\ref{corPathCond (Path Condition)nonconnected}]
\begin{enumerate}
\item
\label{itemLemmeE12w2=w3w2component11}
Let us consider edges $e_1$, $e_2$, $e_3$ in the same connected component of weights $\sigma_1<\sigma_2<\sigma_3$.
By Theorem~\ref{PropInheritanceOfConvexityForPminOnGIffCofGWithWeightw2EitherCompleteOrAllVerticesLinkedToTheSameEndVertexOfe12}
we can assume $e_1 = \lbrace 1, 2 \rbrace$, $e_2 = \lbrace 2, 3 \rbrace$,
and $e_3$ incident to $1$ or~$3$. 
Let $e$ be an edge that is not connected to $e_1$.
Applying Lemma~\ref{LemP_MinMax(w1,w2)<w(e)e'InELinkingeTo21}
to the pair of edges $\lbrace e_1,e_2 \rbrace$
we get $w(e)=\sigma_2$.
If $e_{3}$ is incident to $1$
(resp. $3$),
then
Lemma~\ref{LemP_MinMax(w1,w2)<w(e)e'InELinkingeTo21}
applied to $\lbrace e_{1}, e_{3} \rbrace$
(resp. $\lbrace e_2,e_3 \rbrace$)
implies $w(e)=\sigma_3$,
a contradiction.
\item
\label{itemLemmeE12w2=w3w2component12}
Let us now consider edges $e_1$, $e_2$, $e_3$ in the same connected component
of weights $w_{1}=w_{2}=\sigma_{1}$
and $w_3=\sigma_2$.
By Theorem~\ref{lemInheritanceOfConvexityForPminInheritanceOfConvexityForPMG1Cycle-Complete}
we can assume $e_1 = \lbrace 1, 2 \rbrace$, $e_2= \lbrace 1, 3 \rbrace$,
and $e_3$ incident to $3$.
By contradiction,
let $e$ be an edge that is not connected to $e_1$.
Lemma~\ref{LemP_MinMax(w1,w2)<w(e)e'InELinkingeTo21} applied to $\lbrace e_2,e_3 \rbrace$
implies $w(e)=w_3=\sigma_2$.
Then,
by Lemma~\ref{LemP_MinMax(w1,w2)<w(e)e'InELinkingeTo2}
applied to $\lbrace e_1,e_2 \rbrace$
(we have $w(e)=\sigma_2>\sigma_1=w_1=w_2$),
$e$ has to be linked to $1$,
a contradiction.
\item
Let $e_1 = \lbrace 1, 2 \rbrace$ be the unique edge with weight $\sigma_1$
and let us consider an edge $e_2 = \lbrace 2, 3 \rbrace$ of weight $w_2=\sigma_2>\sigma_1$ adjacent to $e_1$.
Let $e$ be an edge that is not connected to $e_1$.
Applying Lemma~\ref{LemP_MinMax(w1,w2)<w(e)e'InELinkingeTo21} to $\lbrace e_1,e_2 \rbrace$,
we get $w(e) = \sigma_2$.
\qedhere
\end{enumerate}
\end{proof}

\subsection{Complexity analysis}
Using the characterizations previously obtained,
we finally investigate the complexity of the following decision problem:
``Given a weighted graph $G=(N,E,w)$,
is there inheritance of convexity for $\mathcal{P}_{\min}$?''.
Throughout this section
we assume that $G$ is represented by its adjacency matrix
$A = (a_{ij})$
defined by
\[
a_{ij}=
\left
\lbrace
\begin{array}{cl}
w_{ij} &\textrm{if } \lbrace i, j \rbrace \in E,\\
0 & \textrm{otherwise.}
\end{array}
\right.
\]
We first show that cycle-completeness of $G_1=(N, E \setminus E_1)$
can be verified in polynomial time.
Although $G_1$ corresponds to an unweighted subgraph of $G$,
it can be represented by the matrix obtained from $A$
by assigning value~$0$ to all entries associated with edges in $E_1$.
We recall that
a connected graph is \textit{biconnected} if it remains connected after the removal of
any vertex and its incident edges.
A \textit{biconnected component} of a graph is a maximal biconnected subgraph.
We say that 
a biconnected component is complete
if it corresponds to a complete subgraph.
Noting that a graph is cycle-complete if and only if
all its biconnected components are complete,
we can easily check the cycle-completeness of a given graph.
\citet{Tarjan1972} proposed a polynomial algorithm based on a depth-first search procedure
for finding all the biconnected components of an undirected graph.
With the adjacency matrix representation,
Tarjan's algorithm would compute all biconnected components of $G_1$ in $O(n^2)$ time.
Then,
verifying completeness of a given component only requires to check the entries
of the corresponding submatrix.
As two biconnected components cannot have any edge in common,
we can check if all biconnected components are complete
in $O(n^2)$ time.
This implies the following result.

\begin{lemma}
\label{Cycle-CompletenessCanBeVerifiedInO(n^2)Time}
Cycle-completeness of $G_1 = (N, E \setminus E_1)$
can be verified in $O(n^2)$ time.
\end{lemma}

We now consider the remaining conditions on cycles required
in Theorems~\ref{PropInheritanceOfConvexityForPminOnGIffCofGWithWeightwEitherCompleteOrAllVerticesLinkedToTheSameEndVertexOfe1}
and~\ref{PropInheritanceOfConvexityForPminOnGIffCofGWithWeightw2EitherCompleteOrAllVerticesLinkedToTheSameEndVertexOfe12}.

\begin{lemma}
\label{LemmaCondition1(resp.Condition4OtTheorem21(resp.Theorem22)CanBeVerifiedInO(n^2)Time}
Let $G=(N,E,w)$ be a connected weighted graph.
Let us assume $|E_{1}| = 1$
and let $e_{1} = \lbrace 1, 2 \rbrace$ be the unique edge in $E_{1}$.
The existence and uniqueness of a chordless cycle containing $e_1$ can be verified in $O(n^2)$ time.
\end{lemma}

\begin{proof}
We can check the existence of a path
linking $1$ and $2$ in $G_1 = (N, E \setminus E_1)$
with a Breadth First Search (BFS) algorithm in $O(n^2)$ time.
If it exists,
then the BFS algorithm returns a shortest path $P$
and
$P \cup \lbrace e_1 \rbrace$ corresponds to a chordless cycle.
Then,
it remains to check that
there is no path $P' \not= P$ linking $1$ and $2$ in $G_1$
such that $P' \cup \lbrace e_1 \rbrace$ corresponds to a chordless cycle.
Let us note that if a path $P'\not=P$ linking $1$ and $2$ in $G_1$ contains all vertices of $P$,
then  $P' \cup \lbrace e_1 \rbrace$ cannot correspond to a chordless cycle
as at least one edge of $P$ is necessarily a chord of $P' \cup \lbrace e_1 \rbrace$.
Moreover,
if there is a path $P'\not=P$ linking $1$ and $2$ in $G_1$
which does not contain all vertices of $P$,
then
at least one vertex in $V(P) \setminus \lbrace 1, 2 \rbrace$
is not an articulation point\footnote{
A vertex in a graph is an articulation point if its removal disconnects the graph.} in $G_1$.
Tarjan's algorithm returns the articulation points (and the biconnected components)
of $G_1$ in $O(n^2)$ time.
Then,
it is sufficient to check 
that each vertex in $V(P) \setminus \lbrace 1, 2 \rbrace$
belongs to the set of articulation points of $G_1$.
\end{proof}

\begin{lemma}
\label{LemmaForEveryBiconnectedComponentCOfG_1EitherCIsCompleteOrAllVerticesOfCAreLinked}
Let $G=(N,E,w)$ be a connected weighted graph.
Let us assume $|E_{1}| = 1$.
Let $e_{1} = \lbrace 1, 2 \rbrace$ be the unique edge in $E_{1}$,
and let $G_1 = (N, E \setminus E_1)$.
Let us assume that the following condition is satisfied:
\begin{enumerate}
\item
\label{itemLemmaThereExistsAtMostOneCycleContaininge1}
There exists at most one chordless cycle containing $e_{1}$.
\end{enumerate}
Then the following conditions are equivalent:
\begin{enumerate}
\setcounter{enumi}{1}
\item
\label{itemLemmaForEveryCycleCInG_1EitherCIsCompleteOrAllVerticesOfCAreLinkedToTheSameEnd-VertexOfe_1-2}
For every cycle $C$ in $G_1$
either $C$ is complete or
all vertices of $C$ are linked to the same end-vertex of $e_{1}$.
\item
\label{itemLemmaForEvery2-ConnectedComponentCInG_1EitherCIsCompleteOrAllVerticesOfCAreLinkedToTheSameEnd-VertexOfe_1}
For every biconnected component $\tilde{C}$ of $G_1$
with at least three vertices
either $\tilde{C}$ is complete or
all vertices of $\tilde{C}$ are linked to the same end-vertex of $e_{1}$.
\end{enumerate}
Moreover,
these conditions can be verified in $O(n^2)$ time.
\end{lemma}

\begin{proof}
As any cycle belongs to a biconnected component,
Condition~\ref{itemLemmaForEvery2-ConnectedComponentCInG_1EitherCIsCompleteOrAllVerticesOfCAreLinkedToTheSameEnd-VertexOfe_1}
obviously implies Condition~\ref{itemLemmaForEveryCycleCInG_1EitherCIsCompleteOrAllVerticesOfCAreLinkedToTheSameEnd-VertexOfe_1-2}.
Let us assume Condition~\ref{itemLemmaForEveryCycleCInG_1EitherCIsCompleteOrAllVerticesOfCAreLinkedToTheSameEnd-VertexOfe_1-2}
satisfied
and let $\tilde{C}$ be a non-complete biconnected component in $G_1$.
Then
there exist $i$ and $j$ in $V(\tilde{C})$
with $\lbrace i,j \rbrace \notin E \setminus E_1$.
Let $k$ be a vertex in $V(\tilde{C}) \setminus \lbrace i, j \rbrace$.
As $\tilde{C}$ is a biconnected component,
there exists a simple cycle $C$
containing $i$, $j$, and $k$.
As $\lbrace i,j \rbrace \notin E \setminus E_1$,
Condition~\ref{itemLemmaForEveryCycleCInG_1EitherCIsCompleteOrAllVerticesOfCAreLinkedToTheSameEnd-VertexOfe_1-2}
implies that $i$, $j$, and $k$
are linked to the same end-vertex of $e_1$.
We can repeat this reasoning for any $k$ in $V(\tilde{C}) \setminus \lbrace i, j \rbrace$.
Finally,
all vertices in $V(\tilde{C})$
are either all linked to $1$
or all linked to $2$,
otherwise we get a contradiction
to Condition~\ref{itemLemmaThereExistsAtMostOneCycleContaininge1}.
Let us now investigate the complexity.
By Lemma~\ref{LemmaCondition1(resp.Condition4OtTheorem21(resp.Theorem22)CanBeVerifiedInO(n^2)Time},
Condition~\ref{itemLemmaThereExistsAtMostOneCycleContaininge1}
can be verified in $O(n^2)$ time.
Then,
Condition~\ref{itemLemmaForEvery2-ConnectedComponentCInG_1EitherCIsCompleteOrAllVerticesOfCAreLinkedToTheSameEnd-VertexOfe_1}
(which is equivalent to Condition~\ref{itThCycleCOfGWithConstantEdge-Weightw2EitherCompleteOrAllVerticesLinkedToSameEndVertexOfe1})
can be checked in $O(n^2)$ time as follows.
We can obtain all biconnected components of $G_1$
with Tarjan's algorithm in $O(n^2)$ time.
Then,
for any biconnected component $\tilde{C}$,
we check if $a_{1i} \not= 0$ for all $i$ in $V(\tilde{C})$
or $a_{2i} \not= 0$ for all $i$ in $V(\tilde{C})$.
If these conditions are not satisfied
then we check the completeness of $\tilde{C}$.
\end{proof}

\begin{proposition}
\label{PropInheritanceConvexityForPminPolynomialTime}
Inheritance of convexity for $\mathcal{P}_{\min}$ can be decided in $O(n^2)$ time.
\end{proposition}

\begin{proof}
Let $\lbrace \sigma_{1}, \sigma_{2}, \ldots, \sigma_{k} \rbrace$,
with $k \leq |E|$,
be the set of edge-weights in $G$ 
with $\sigma_{1} < \sigma_{2} < \ldots < \sigma_{k}$.
We apply the following procedure.
We first count in $A$
the number $k$ of edge-weights
and the number $n_{i}$ of occurences of $\sigma_{i}$,
\emph{i.e.},
$n_{i} = |E_{i}|$,
for all $i$, $1 \leq i \leq k$.
If $k > 3$ or if $k =3$ and $n_{1} >1$,
then we stop
as there is no inheritance of convexity for $\mathcal{P}_{\min}$
by Proposition~\ref{propAtMost3DifferentEdgeWeights}.
Otherwise,
we use the characterizations given
in Theorems~\ref{lemInheritanceOfConvexityForPminInheritanceOfConvexityForPMG1Cycle-Complete},
\ref{PropInheritanceOfConvexityForPminOnGIffCofGWithWeightwEitherCompleteOrAllVerticesLinkedToTheSameEndVertexOfe1},
and
\ref{PropInheritanceOfConvexityForPminOnGIffCofGWithWeightw2EitherCompleteOrAllVerticesLinkedToTheSameEndVertexOfe12}
to solve the decision problem in the remaining cases
described below.
If a contradiction is found for a given case,
then we stop the associated procedure
as it implies there is no inheritance of convexity.

Let us assume $k=2$ and $n_{1} \geq 2$.
We have to check Conditions~\ref{itemLemInheritanceOfConvexityForPMOnG1=(N,E2)}
and~\ref{itemLemInheritanceOfConvexityForPMOnG1=(N,E3)b}
of Theorem~\ref{lemInheritanceOfConvexityForPminInheritanceOfConvexityForPMG1Cycle-Complete}.
Let $i^*$ be the smallest index $i \in \lbrace 1, \ldots, n \rbrace$
such that row $i$ of $A$ contains at least two values $\sigma_1$.
If there are indices $i \not= i^*$ and $j \not= i^*$
with $a_{ij} = \sigma_1$,
then we stop
as it contradicts Condition~\ref{itemLemInheritanceOfConvexityForPMOnG1=(N,E2)}.
If row~$i^*$ contains a value $\sigma_{2}$,
we stop.
If there are indices $i \not= i^*$
and $j \not= i^*$
such that $a_{ij} = \sigma_{2}$,
$a_{i^{*}j} \not= \sigma_{1}$
and $a_{ii^{*}} \not= \sigma_{1}$,
we stop as
it still contradicts Condition~\ref{itemLemInheritanceOfConvexityForPMOnG1=(N,E2)}.
Otherwise,
it only remains to check Condition~\ref{itemLemInheritanceOfConvexityForPMOnG1=(N,E3)b},
\emph{i.e.},
whether $G_{1}$ is cycle-complete.
By Lemma~\ref{Cycle-CompletenessCanBeVerifiedInO(n^2)Time},
it can be done in $O(n^2)$ time.

Let us now assume $k=2$ and $n_{1} = 1$.
By Lemmas~\ref{LemmaCondition1(resp.Condition4OtTheorem21(resp.Theorem22)CanBeVerifiedInO(n^2)Time}
and~\ref{LemmaForEveryBiconnectedComponentCOfG_1EitherCIsCompleteOrAllVerticesOfCAreLinked},
Conditions~\ref{itThThereExistsAtMostOneCycleCWithoutChordContaininge1} 
and~\ref{itThCycleCOfGWithConstantEdge-Weightw2EitherCompleteOrAllVerticesLinkedToSameEndVertexOfe1}
of Theorem~\ref{PropInheritanceOfConvexityForPminOnGIffCofGWithWeightwEitherCompleteOrAllVerticesLinkedToTheSameEndVertexOfe1}
can be verified in $O(n^2)$ time.


Let us finally assume $k=3$.
We have to verify Conditions~\ref{ThereExistsOnlyOneEdgeInE1}
to~\ref{EveryCycleWhichDoesNotContainse1IsComplete}
of Theorem~\ref{PropInheritanceOfConvexityForPminOnGIffCofGWithWeightw2EitherCompleteOrAllVerticesLinkedToTheSameEndVertexOfe12}.
As $n_{1} = 1$,
Condition~\ref{ThereExistsOnlyOneEdgeInE1}
is satisfied.
Let us assume w.l.o.g. $E_{1} = \lbrace e_1 \rbrace$
with
$e_{1} = \lbrace 1, 2 \rbrace$.
Let $i^*$ be the smallest index $i \in \lbrace 1, \ldots, n \rbrace$
such that row $i$ of $A$ contains the value $\sigma_2$.
If $i^* \notin \lbrace 1, 2 \rbrace$,
we stop as it contradicts Condition~\ref{EveryEdgeOfWeightSigma2IsIncidentToTheSameEnd-Vertex2Ofe1}.
Otherwise,
if there are indices $i>i^*$ and $j>i$ such that $a_{ij} = \sigma_2$,
then we stop as
it still contradicts Condition~\ref{EveryEdgeOfWeightSigma2IsIncidentToTheSameEnd-Vertex2Ofe1}.
If row~$i^*$ contains $\sigma_{3}$
or if there are indices $i \in \lbrace 3, \ldots, n \rbrace$
and $j \in \lbrace 3, \ldots, n  \rbrace$
such that $a_{ij} = \sigma_{3}$,
$a_{ii^{*}} \not= \sigma_{2}$
and $a_{i^{*}j} \not= \sigma_{2}$,
we stop as
it contradicts Condition~\ref{EveryEdgeOfWeightSigma3IsConnectedTo2Bye1OrByAnEdgeOfWeightSigma2}.
Otherwise,
it only remains to check Conditions~\ref{ThereExistsAtMostOneChordlessCycleCmWithm=3or4Containinge1}
and \ref{EveryCycleWhichDoesNotContainse1IsComplete}.
By Lemmas~\ref{Cycle-CompletenessCanBeVerifiedInO(n^2)Time}
and~\ref{LemmaCondition1(resp.Condition4OtTheorem21(resp.Theorem22)CanBeVerifiedInO(n^2)Time},
these last conditions
can be verified in $O(n^2)$ time.
\end{proof}

\bibliographystyle{apalike}
\bibliography{biblio}

\end{document}